\documentclass[journal]{IEEEtran}
\pdfoutput=1
\IEEEoverridecommandlockouts
\usepackage[utf8]{inputenc} 
\usepackage[table,dvipsnames]{xcolor}
\usepackage{graphicx}
\graphicspath{{figs/}}

\usepackage{epstopdf}

\usepackage{tikz}
\usetikzlibrary{arrows}
\usepackage{tabularx}
\usepackage{subfig}
\usepackage{dcolumn}
\usepackage{multirow}
\usepackage{pgfplots}
\usepackage{float}
\usepackage{amsthm}
\usepackage[cmex10]{amsmath}
\usepackage{amssymb}
\usepackage[numbers,sort&compress]{natbib}
\usepackage{braket}
\usepackage{array}
\usepackage{longtable}
\usepackage{arydshln}
\usepackage{bm}
\usepackage{bbm}

\DeclareMathOperator{\mmse}{mmse}

\usepackage[T1]{fontenc}
\usepackage{url}
\usepackage{ifthen}
\usepackage{cuted}

\usepackage[unicode]{hyperref}
\makeatletter
\newcommand{\distas}[1]{\mathbin{\overset{#1}{\kern\z@\sim}}}%
\newsavebox{\mybox}\newsavebox{\mysim}
\newcommand{\distras}[1]{%
  \savebox{\mybox}{\hbox{\kern3pt$\scriptstyle#1$\kern3pt}}%
  \savebox{\mysim}{\hbox{$\sim$}}%
  \mathbin{\overset{#1}{\kern\z@\resizebox{\wd\mybox}{\ht\mysim}{$\sim$}}}%
}
\makeatother
\newtheorem{theorem}{Theorem}[section]
\newtheorem{lemma}[theorem]{Lemma}
\newtheorem{claim}{Claim}
\newtheorem{define}{Definiton}
\newtheorem{prop}{Proposition}
\newtheorem{remark}{Remark}

\def\cF{\mathcal{F}}
\def\cQ{\mathcal{Q}}
\newcommand{\Pc}[1]{P_{c_{[#1]}}}

\def\Var{\mathrm{Var}}
\newcommand{\Pcp}[1]{P^{\perp}_{c_{[#1]}}}
\newcommand{\Pd}[2]{P_{c'_{[#1]},c_{[#2]}}}
\newcommand{\Pdd}[2]{P_{c_{[#1]},c_{[#2]}}}
\def\PP{P^{\perp}}
\def\PPA{P^{\perp}_{A_1}}
\newcommand{\norm}[1]{\left\| #1 \right\|}
\newcommand{\norma}[1]{\biggr\| #1 \biggr\|}
\newcommand{\Pb}[1]{\mathbb{P}\left[#1 \right]}
\newcommand{\Prb}[1]{\mathbb{P}\biggr[#1 \biggr]}
\newcommand{\dop}[2]{\left\langle #1, #2 \right\rangle}
\newcommand{\Exp}[1]{\exp\left(#1\right)}
\newcommand{\Ex}[1]{\mathbb{E}\left[#1 \right]}
\newcommand{\Exa}[1]{\mathbb{E}\biggr[#1 \biggr]}
\def\EE{\mathbb{E}}

\newcommand{\Exxa}[2]{\mathbb{E}_{#1}\biggr[#2 \biggr]}

\newcommand{\txb}[1]{#1}

\def\phd{\left(1+P'_{n}\sum_{i\in D^c}|H_i|^2\right)}
\def\phdd{\left(1+P'\sum_{i\in D^c}|H_i|^2\right)}
\def\phs{\left(1+P'\sum_{i\in S\setminus S_2^*}|H_i|^2\right)}
\def\phsa{\biggr(1+P'\sum_{i\in S\setminus S_2^*}|H_i|^2\biggr)}
\def\phss{P'\sum_{i\in S}|H_i|^2}
\def\hc{\sum_{i\in S}H_i c_i}
\def\hcc{\sum_{i\in S_2^*}H_i c_i}
\def\cn{\mathcal{CN}}
\def\cs{\mathcal{(CS)}}
\def\cd{\mathcal{C}}
\def\mate{\mathcal{E}}
\def\eqdef{\triangleq}
\def\PT{P_{tot}}

\def\PTX{P_{tot,\nu}}

\def\Ieee{IEEEeqnarray*}
\def\Ieeen{\IEEEyesnumber}

\begin{document}
\title{Fundamental limits of many-user MAC with finite payloads and fading\\
\thanks{This material is based upon work supported by the National Science
 Foundation under Grant No CCF-17-17842 and a grant from Skoltech--MIT Joint Next Generation Program (NGP). This work was also supported in part by the MIT-IBM Watson AI Lab.}}

%\author{
%  \IEEEauthorblockN{Suhas S Kowshik\\}
%  \IEEEauthorblockA{Massachusetts Institute of Technology\\
%                    Dept. of EECS\\
%                    Cambridge, Massachusetts, USA\\
%                    Email: suhask@mit.edu\\}
%                    \thanks{This paper [in part and compressed form] is accepted to be published at IEEE International Symposium on Information Theory (ISIT) 2019 \citep{kowshik2019quasi} }
%                    \thanks{Parts of the results on classical regime in this paper (theorems \ref{th:cap_joint} \ref{th:pu1}) have appeared in the one of the authors' Masters thesis \citep{kowshik2018coding}}
%  \and
%  \IEEEauthorblockN{Yury Polyanskiy\\}
%  \IEEEauthorblockA{Massachusetts Institute of Technology\\
%                    Dept. of EECS\\ 
%                    Cambridge, Massachusetts, USA\\
%                    Email: yp@mit.edu}
%}
\author{Suhas~S~Kowshik and~Yury~Polyanskiy,~\IEEEmembership{Senior~Member,~IEEE} \thanks{S. S. Kowshik and Y. Polyanskiy are with the Department of Electrical Engineering and Computer Science, Massachusetts Institute of Technology, Cambridge, MA 02139, USA.}%
\thanks{This paper was presented in part at the 2019 IEEE International Symposium on Information Theory (ISIT)~\citep{kowshik2019quasi}.}%
\thanks{Copyright~\copyright~2021 IEEE. Personal use of this material is permitted.  However, permission to use this material for any other purposes must be obtained from the IEEE by sending a request to pubs-permissions@ieee.org}
}


\maketitle

\begin{abstract}
Consider a (multiple-access) wireless communication system where users are connected to a unique base station over a
shared-spectrum radio links. Each user has a fixed number $k$ of bits to send to the base station, and his signal gets
attenuated by a random channel gain (quasi-static fading). In this paper we consider the many-user asymptotics of
Chen-Chen-Guo'2017, where the number of users grows linearly with the blocklength. Differently, though, we adopt a per-user
probability of error (PUPE) criterion (as opposed to classical joint-error probability criterion). Under
PUPE the finite energy-per-bit communication is possible, and we are able to derive bounds on the tradeoff between
energy and spectral efficiencies. We reconfirm the curious behaviour (previously observed for non-fading MAC) of the
possibility of almost perfect multi-user interference (MUI) cancellation for user densities below a critical threshold.
Further, we demonstrate the suboptimality of standard solutions such as orthogonalization (i.e., TDMA/FDMA) and treating
interference as noise (i.e. pseudo-random CDMA without multi-user detection). Notably, the problem treated here can
be seen as a variant of support recovery in compressed sensing for the unusual definition of sparsity with one non-zero
entry per each contiguous section of $2^k$ coordinates. This identifies our problem with that of the sparse regression
codes (SPARCs) and hence our results can be equivalently understood in the context of SPARCs with sections of 
length $2^{100}$. Finally, we discuss the relation of the almost perfect MUI cancellation property and the
replica-method predictions.
\end{abstract}


\section{Introduction}
\label{sec:intro}

We clearly witness two recent trends in the wireless communication technology: the increasing 
deployment density and miniaturization of radio-equipped sensors. The first trend results in progressively worsening
interference environment, while the second trend puts ever more stringent demands on communication energy efficiency.
This suggests a bleak picture for the future networks, where a chaos of 
packet collisions and interference contamination prevents reliable connectivity. 

This paper is part of a series aimed at elucidating the fundamental tradeoffs in this new ``dense-networks'' regime of communication, and on rigorously demonstrating suboptimality of state-of-the-art radio-access solutions (ALOHA,
orthogonalization, or FDMA/TDMA \footnote{In this paper we do not distinguish between TDMA, FDMA or any other orthogonalization strategy. From here on, all orthogonalization strategies will be referred to as ``TDMA''}, and treating interference as noise, or TIN). This suboptimality will eventually lead to dramatic consequences. For example, environmental impact of billions of toxic batteries getting depleted at 1/10 or 1/100 of the planned service time is easy to imagine. In order to future-proof our systems, we should avoid locking in on outdated and unfixable multiple-access architectures causing tens of dB losses in energy efficiency. The information-theoretic analysis in this paper demonstrates that the latter is indeed unavoidable (with state-of-the-art schemes). However, our message is in fact optimistic, as we also demonstrate existence of protocols which are partially immune to the increase of the sensor density.

%\IEEEpubidadjcol
Specifically, in this paper we consider a problem of $K$ nodes communicating over a
frame-synchronized multiple-access channel.  When $K$ is fixed and the frame size $n$ (which we will also call 
``blocklength'' or the ``number of degrees of freedom'') is taken to  infinity we get the classical
regime~\cite{cover2012elements}, in which the fundamental limits are given by well-known mutual information expressions. A
new regime, deemed \textit{many-access}, was put forward by Chen, Chen and Guo~\cite{chen2017capacity}. In this regime
the number of nodes $K$ grows with blocklength $n$. It is clear that the most natural scaling is linear: $K=\mu n,
n\to\infty$, corresponding to the fact that in time $n$ there are linearly many users that will have updates/traffic to
send~\cite{polyanskiy2017perspective}. That is, if each device wakes up once in every $T$ seconds and transmits over a frame of length $t$, then in time (proportional to) $t$ there are $K\approx t/T$ users where $t$ is large enough for this  approximation to hold but small that no device wakes up twice. Further, asymptotic results obtained from this linear scaling have been shown to approximately predict behavior of the fundamental limit at finite blocklength, e.g. at $n=30000$ and $K<=300$ \citep{polyanskiy2017perspective,ZPT-isit19}. The analysis of~\cite{chen2017capacity} focused on the regime of infinitely large
payloads (see
also~\cite{wei2018fundamental} for a related massive MIMO MAC analysis in this setting). In contrast~\cite{polyanskiy2017perspective} proposed to focus on a model where each of the $K=\mu n$ nodes
has only finitely many bits to send. In this regime, it turned out, one gets the relevant engineering trade-offs. Namely,
the communication with finite energy-per-bit is possible as $n\to\infty$ and the optimal energy-per-bit depends on the
user density $\mu$. For this to happen, however, a second crucial departure from the classical MAC model was needed:
the per-user probability of error, \textit{PUPE}, criterion~\cite{polyanskiy2017perspective}.

These two modifications (the scaling $K=\mu n$ and the PUPE) were investigated in the case of the AWGN channel 
in~\cite{polyanskiy2017perspective,ZPT-isit19,polyanskiy2018_course}. We next describe the main discovery of that work. The channel
model is:
\begin{equation}\label{eq:awgn}
	Y^n = \sum_{i=1}^K X_i + Z^n\,,\qquad Z^n \sim \mathcal{CN}(0,I_n)\,,
\end{equation}
and $X_i = f_i(W_i) \in \mathbb{C}^n$ is the codeword of $i$-th user corresponding to $W_i \in [2^k]$ chosen uniformly
at random. The
system is said to have PUPE $\epsilon$ if there exist decoders $\hat W_i = \hat W_i(Y^n)$ such that
\begin{\Ieee}{LLL}
\label{eq:pupe1}
P_{e,u}= \frac{1}{K} \sum_{i=1}^K \Pb{W_i \neq \hat W_i} \le \epsilon\,.\Ieeen
\end{\Ieee}
The energy-per-bit is defined as 
$$\frac{E_b}{N_0}=\frac{1}{k} \sup_{i\in [K], w\in [2^k]} \|f_i(w)\|^2.$$
The goal
in~\cite{polyanskiy2017perspective,polyanskiy2018_course} was to characterize the asymptotic limit
\begin{equation}
\label{eq:fund_lim}
\mate^*(\mu, k, \epsilon) \eqdef \limsup_{n\to\infty} \inf \frac{E_b}{N_0} 
\end{equation} 
where infimum is taken over all possible encoders $\{f_i\}$ and decoders $\{\hat W_i\}$ achieving the PUPE $\epsilon$
for $K=\mu n$ users. (Note that this problem may be recast in the language of compressed sensing and sparse regression
codes (SPARCs) -- see Section~\ref{sec:cs} below.)

To predict how $\mate^*(\mu, \epsilon)$ behaves, first consider a naive Shannon-theoretic
calculation~\cite{verdu2002spectral}: if $K$ users want to send $k$ bits in $n$ degrees of freedom, then their sum-power $P_{tot}$ should 
satisfy
$$ n\log(1+P_{tot}) = kK\,.$$
In turn, the sum-power $P_{tot} = \frac{kK}{n} \frac{E_b }{ N_0}$. Overall, we get 
$$ \mate^* \approx \frac{2^{\mu k} - 1}{ k\mu}\,. $$
This turns out to be a correct prediction, but only in the large-$\mu$ regime. The true behavior of the fundamental limit is roughly given by
\begin{equation}\label{eq:ebno_approx}
	\mate^*(\mu, k, \epsilon) \approx \max\left(\frac{2^{\mu k} - 1}{ k\mu}, \mate_{\mathrm{s.u.}}\right)\,, 
\end{equation}
where $\mate_{\mathrm{s.u.}}=\mate_{\mathrm{s.u.}}(k,\epsilon)$ does not depend on $\mu$ and corresponds to the single-user minimal energy-per-bit for sending $k$ bits with error $\epsilon$, for which a very tight characterization is given in \citep{polyanskiy2010minimum}. In particular, with good precision for $k\ge 10$ we have
\begin{equation}
\label{eq:E_su}
\mate_{\mathrm{s.u.}}(k,\epsilon)=\frac{1}{2k}\left(\mathcal{Q}^{-1}\left(2^{-k}\right)-\mathcal{Q}^{-1}\left(1-\epsilon\right)\right)^2
\end{equation}
where $\mathcal{Q}$ is the complementary CDF of the standard normal distribution:
$\mathcal{Q}(x)=\frac{1}{\sqrt{2\pi}}\int_{x}^{\infty}e^{-\frac{u^2}{2}}du$. 

 In all, results of~\cite{polyanskiy2017perspective,ZPT-isit19,polyanskiy2018_course} suggest that the minimal energy-per-bit has a certain
``inertia'': as the user density $\mu$ starts to climb from zero up, initially the energy-per-bit should stay the same
as in the single-user  $\mu=0$ limit. In other words, optimal multiple-access architectures should be able to
\textit{almost perfectly cancel all multi-user interference (MUI)}, achieving an essentially single-user performance for each user, \textit{provided the user density is below a critical threshold}. Note that this is much better than orthogonalization, which achieves the same effect at the expense of shortening the available (to each user) blocklength by a factor of $\frac{1}{ K}$. Quite surprisingly, standard approaches to multiple-access such as TDMA and TIN\footnote{Note that pseudo-random CDMA systems without multi-user detection and large load factor provide an efficient
implementation of TIN. So throughout our discussions, conclusions about TIN also pertain to CDMA systems of this kind.}, while having an
optimal performance at $\mu\to0$ demonstrated a significant suboptimality for $\mu>0$ regime. In particular, no
``inertia'' was observed and the energy-per-bit for those suboptimal architectures is always a monotonically increasing
function of the user density $\mu$. This opens the (so far open) quest for finding a future-proof MAC architecture that
would achieve $\mate_{\mathrm{s.u.}}$ energy-per-bit for a strictly-positive $\mu>0$. A thorough discussion of this curious behavior and its connections to replica-method predicted phase transitions is contained in
Section~\ref{sec:curious}.

(We note that in this short summary we omitted another important part of~\cite{polyanskiy2017perspective}: the issue of
random-access -- i.e. communicating when the 
identities/codebooks of active users are unknown a priori. We mention, however, that for the random-access version of the problem, there are a number of
low-complexity (and quite good performing) algorithms that are
available~\cite{kowshik2019energy,ordentlich2017low,vem2017user,amalladinne2018coupled,10.1007/978-3-030-01168-0_15,thompson2018compressed,calderbank2018chirrup}. See~\cite{pradhan2020polar, amalladinne2020unsourced,fengler2019sparcs,andreev2020polar} for more recent developments.)

The \textit{contribution} of this paper is in demonstrating the same almost perfect MUI cancellation effect in a much
more practically relevant communication model, in which the ideal unit power-gains of~\eqref{eq:awgn} are replaced by
random (but static) fading gain coefficients. We consider two cases of the channel state information: known at the
receiver (CSIR) and no channel state information (noCSI).

\textit{Key technical ideas:} For handling the noCSI case we employ the subspace projection decoder similar to the one proposed in \citep{yang2014quasi}, which can be seen as a version of the maximum-likelihood decoding (without prior on fading coefficients) -- an idea often used in support recovery literature~\citep{wainwright2009information,reeves2009note,reeves2012sampling}. Another key idea is to decode only a subset of users corresponding to the strongest channel gains -- a principle originating from
Shamai-Bettesh~\citep{bettesh2000outages}. While the randomness of channel gains increases the energy-per-bit requirements, in a related paper we find~\citep{kowshik2019energy} an unexpected advantage: the inherent randomization helps the decoder disambiguate different users and improves performance of the belief propagation decoder. Our second achievability bound improves projection decoder in the the low user density (low spectral efficiency) regime by applying the Approximate Message Passing (AMP) algorithm~\citep{donoho2009message}. The rigorous analysis of its performance is made possible by results in~\citep{bayati2011dynamics,reeves2012sampling}. On the converse side, we leverage the recent finite blocklength results for the noCSI channel from~\citep{polyanskiy2011minimum,yang2014quasi}.

The paper is organized
as follows. In Section~\ref{sec:model} we formally define the problem and the fundamental limits. In
Section~\ref{sec:cs} relation with compressed sensing is discussed. In
Section~\ref{sec:classical} as a warm-up we discuss the classical regime ($K$--fixed, $n\to\infty$) under the PUPE
criterion. We show that our projection decoder achieves the best known achievability bound in this
setting~\cite{bettesh2000outages}. (We also note that for the quasi-static fading channel model the idea of PUPE is very natural, and implicitly
appears in earlier works, e.g.~\cite{bettesh2000outages,tuninetti2002throughput}, where it is conflated with the outage
probability.) After this short warm-up we go to our
main Section~\ref{sec:manymac}, which contains rigorous achievability and converse bounds 
for the $K=\mu n, n\to\infty$ scaling regime. Some numerical evaluations are presented
Section~\ref{sec:numerical}, where we also compare our bounds with the TDMA and TIN. Finally, in
Section~\ref{sec:curious} we discuss the effect of
almost perfect MUI cancellation and its relation to other phase transitions in compressed sensing.

\subsection{Notations}
Let $\mathbb{N}$ denote the set of natural numbers. For $n\in\mathbb{N}$, let $\mathbb{C}^n$ denote the $n$--dimensional complex Euclidean space. Let $S\subset \mathbb{C}^n$. We denote the projection operator or matrix on to the subspace \textit{spanned} by $S$ as $P_{S}$ and its orthogonal complement as $P^{\perp}_{S}$. For $0\leq p\leq 1$, let $h_2(p)=-p\log_2(p)-(1-p)\log_2(1-p)$ and $h(p)=-p\ln(p)-(1-p)\ln(1-p)$, with $0\ln 0$ defined to be $0$. We denote by $\mathcal{N}(0,1)$ and $\cn(0,1)$ the standard normal and the standard circularly symmetric complex normal distributions, respectively. $\mathbb{P}$ and $\mathbb{E}$ denote probability measure and expectation operator respectively. For $n\in \mathbb{N}$, let $[n]=\{1,2,...,n\}$. $\log$ denotes logarithm to base $2$. Lastly, $\norm{\cdot}$ represents the standard euclidean norm.

%%%%%%%%%%%%%%%%%%%%%%%%%%%%%%%%%%%%%%%%%%%%%%%%%%%%%%%

\section{System Model}\label{sec:model}
Fix an integer $K\geq 1$ -- the number of users. Let $\{P_{Y^n|X^n}=P_{Y^n|X_1^n,X_2^n,...,X_K^n}:\prod_{i=1}^K \mathcal{X}_i^n\to \mathcal{Y}^n\}_{n=1}^{\infty}$ be a multiple access channel (MAC). In this work we consider only the quasi-static fading AWGN MAC: the channel law $P_{Y^n|X^n}$ is described by
\begin{equation}
\label{eq:sys1}
Y^n=\sum_{i=1}^{K}H_i X_i^n+Z^n
\end{equation}
where $Z^n\distas{}\cn(0,I_n)$, and $H_i\distas{iid}\cn(0,1)$ are the fading coefficients which are independent of $\{X^n_i\}$ and $Z^n$. Naturally, we assume that there is a maximum power constraint:
\begin{equation}
\label{eq:pow}
\norm{X^n_i}^2\leq nP. 
\end{equation}

We consider two cases: 1) no channel state information (no-CSI): neither the transmitters nor the receiver knows the realizations of channel fading coefficients, but they both know the law; 2) channel state information only at the receiver (CSIR): only the receiver knows the realization of channel fading coefficients. The special case of \eqref{eq:sys1} where $H_i=1,\forall i$ is called the Gaussian MAC (GMAC). 

In the rest of the paper we drop the superscript $n$ unless it is unclear.

\begin{define}
\label{def:1}
An $\left((M_1,M_2,...,M_K),n,\epsilon\right)_{U}$ code for the MAC $P_{Y^n|X^n}$ is a set of (possibly randomized) maps $\{f_i:[M_i]\to\mathcal{X}_i^n\}_{i=1}^{K}$ (the encoding functions) and $g:\mathcal{Y}^n\to \prod_{i=1}^{K}[M_i]$ (the decoder) such that if for $j\in [K]$, $X_j=f_j(W_j)$ constitute the input to the channel and $W_j$ is chosen uniformly (and independently of other $W_i$, $i\neq j$) from $[M_j]$ then the average (per-user) probability of error satisfies
\begin{\Ieee}{LLL}
\label{eq:def1}
P_{e,u}=\frac{1}{K}\sum_{j=1}^{K}\Pb{W_j\neq \left(g(Y)\right)_j}\leq \epsilon\Ieeen
\end{\Ieee}
where $Y$ is the channel output.
\end{define}

We define an $\left((M_1,M_2,...,M_K),n,\epsilon\right)_{J}$ code similarly, where $P_{e,u}$ is replaced by the usual joint error
\begin{equation}
\label{eq:def2}
P_{e,J}=\Pb{\bigcup_{j\in [K]}\left\{W_j\neq \left(g(Y)\right)_j\right\}}\leq \epsilon
\end{equation}

Further, if there are cost constraints, we naturally modify the above definitions such that the codewords satisfy the constraints. 

\begin{remark}
Note that in \eqref{eq:def1}, we only consider the \emph{average} per-user probability. But in some situations, it might be relevant to consider \emph{maximal} per-user error (of a codebook tuple) which is the maximum of the probability of error of each user. Formally, let $\cd_{[K]}=\{\cd_1,...,\cd_K\}$ denote the set of codebooks. Then
\begin{\Ieee}{LLL}
\label{eq:pupe_max}
P_{e,u}^{\max}=P_{e,u}^{\max}(\cd_{[K]})\\
=\max\left\{\Pb{W_1\neq\hat{W}_1 },...,\Pb{W_K\neq \hat{W}_K  }\right\}\Ieeen
\end{\Ieee}
where the probabilities are with respect to the channel and possibly random encoding and decoding functions. In this paper we only consider the fundamental limits with respect to $P_{e,u}$ and PUPE always refers to this unless otherwise noted. But we note here that for both asymptotics and FBL the difference is not important. See appendix \ref{app:pupe-max} for a discussion on this -- there we show that by random coding $\Ex{P_{e,u}^{\max}}$ is asymptotically equal to $\Ex{P_{e,u}}$ (expectations are over random codebooks).

\end{remark}

%{\color{blue}{
\subsection{Connection to compressed sensing and sparse regression codes}\label{sec:cs}
The system model and coding problem considered in this work (see eqn. \eqref{eq:sys1}) can be cast as a support recovery problem in compressed sensing. Suppose we have $K$ users each with a codebook of size $M$ and blocklength $n$. Let $A_i$ be the $n\times M$ matrix consisting of the codewords of user $i$ as columns. Then the codeword transmitted by the user can be represented as $X_i=A_i W_i$ where $W_i\in\{0,1\}^M$ with a single nonzero entry. Since each codeword is multiplied by a scalar random gain $H_i$, we let $U_i=H_i W_i$ which is again a $1$ sparse vector of length $M$. Finally the received vector $Y$ can be represented as 
\begin{equation}
\label{eq:cs1}
Y=\sum_{i\in[K]}H_i X_i +Z=AU+Z
\end{equation}
where $A=[A_1,\cdots,A_K]$ is $n\times KM$ matrix obtained by concatenating the codebooks, $U=[U_1^T,\cdots,U_K^T]^T$ is
the $MK\times 1$ length vector denoting the codewords (and fading gains) of each user. In our problem, the vector $U$ has a {\it
block-sparse} structure, 
namely $U$ has $K$ sections, each of length $M$, and there is only a single non-zero entry in each section. (Majority of
compressed sensing literature focuses on the {\it non-block-sparse} case, where $U$ has just $K$ non-zero entries, which
can be spread arbitrarily inside $KM$ positions.)
Decoding of the codewords, then, is equivalent to the support recovery problem under the block-sparse structure, a
problem considered in compressed sensing. In our setup, we keep $M$ fixed and let $K,n\to \infty$ with constant $\mu=K/n$. Hence $1/M$ is the sparsity rate and $M\mu$ is the measurement rate. 

This connection is not new and has been observed many times in the past \citep{aeron2010information,jin2011limits}. In \citep{jin2011limits} the authors consider a the exact support recovery problem in the case when the vector $U$ is just sparse (with or without random gains). This corresponds to the random access version of our model where the users share a same codebook \citep{kowshik2020energy}. They analyze the fundamental limits in terms of the rate (i.e., ratio of logarithm of signal size to number of measurements) necessary and sufficient to ensure exact recovery in both cases when sparsity is fixed and growing with the signal size. For the fixed sparsity case and $U$ having only ${0,1}$ entries, this fundamental limit is exactly the symmetric capacity of an AWGN multiple access channel with same codebook (with non colliding messages). With fading gains, they recover the outage capacity of quasi-static MAC \citep{biglieri1998fading,han1998information} (but with same codebook). 

In \citep{aeron2010information}, the authors discuss necessary and sufficient conditions for the
exact and approximate support recovery (in Hamming distortion), and $L_2$ signal recovery with various conditions on signal $X$ and matrix $A$ (deterministic versus random, discrete versus continuous support etc.). These results differ from ours in the sense that they are not for block sparse setting and more importantly, they do not consider approximate support recovery with Hamming distortion when the entries of the support of the signal are sampled from a continuous distribution, which is the case we analyze. Hence our results are not directly comparable.

Work~\citep{reeves2012sampling} comes closest to our work in terms of the flavor of results of achievability. As pointed out in~\citep{reeves2012sampling} itself, many other works like~\citep{aeron2010information} focus on the necessary and sufficient scalings (between sparsity, measurements and signal dimension) for various forms of support recovery. But the emphasis in~\citep{reeves2012sampling} and this work is on the precise constants associated with these scalings. In particular, the authors in~\citep{reeves2012sampling} consider the approximate support recovery (in Hamming distortion) problem when the entries in the support of the signal come from a variety of distributions. They analyze various algorithms, including matched filter and AMP, to find the minimum measurement rate required to attain desired support distortion error in terms of signal to noise ratio and other parameters. Furthermore, they compare these results to that of the optimal
decoder predicted by the replica method~\citep{guo2005randomly}.

The result on using replica method in~\citep{reeves2012sampling} is not directly applicable since our signal has block sparse (as opposed to i.i.d.) coordinates. But the AMP analysis presented there can be extended to our setting.
%Some of those results are not directly applicable since our signal has block sparse (as opposed to i.i.d.) coordinates. The notable exception are the results on the AMP. 
Because of the generality of the analysis
in~\cite{bayati2011dynamics}, it turns out to be possible to derive rigorous claims (and computable expressions) on the performance of the (scalar) AMP even in the block sparse setting. This is the content of Section~\ref{sec:amp_ach} below. Unlike the achievability side, for the converse we cannot rely on bounds in~\citep{reeves2012sampling} proven for the i.i.d. coordinates of $X$. Even ignoring the difference between the structural assumptions on $X$, we point out also that our converse bounds leverage finite-length results from~\citep{polyanskiy2011minimum}, which makes them tighter than the genie-based bounds in~\citep{reeves2013approximate}. 

The block-sparse assumption, however, comes very naturally in the area of
SPARCs~\citep{barron2012least,CIT-092,rush2018error}. The section error rate (SER) of a SPARC is precisely our PUPE. The vector AMP algorithm has been analyzed for SPARC with i.i.d Gaussian design matrix in \citep{rush2018error} and for the spatially-coupled matrix in~\citep{rush2020capacity} but for the AWGN
channel (i.e., when non-zero entries of $U$ in~\eqref{eq:cs1} are all 1). In\citep{barbier2017approximate}, heuristic derivation of state evolution of the vector-AMP decoder for spatially-coupled SPARCS was presented for various signal classes (this includes our fading scenario). However, the the resulting fixed point equations may not be possible to solve for our block
size as it amounts to computing $2^{100}$ dimensional integrals (and this also prevents evaluation of replica-method
predictions from \citep{barbier2017approximate}).

%}}

\section{Classical regime: $K$ fixed, $n\to\infty$}
\label{sec:classical}
In this section, we focus on the channel under classical asymptotics where $K$ is fixed (and large) and $n\to \infty$. Further, we consider two distinct cases of joint error and per-user error. We show that subspace projection decoder \eqref{eq:dec1} achieves a) $\epsilon$--capacity region ($C_{\epsilon,J}$) for the joint error and b) the best known bound for $\epsilon$--capacity region $C_{\epsilon,PU}$ under per-user error. This motivates using projection decoder in the many-user regime.

\subsection{Joint error}

A rate tuple $(R_1,...,R_K)$ is said to be $\epsilon$--achievable \citep{han1998information} for the MAC if there is a sequence of codes whose rates are asymptotically at least $R_i$ such that joint error is asymptotically smaller than $\epsilon$. Then the $\epsilon$--capacity region $C_{\epsilon,J}$ is the closure of the set of $\epsilon$--achievable rates. For our channel \eqref{eq:sys1}, the $C_{\epsilon,J}$ does not depend on whether or not the channel state information (CSI) is available at the receiver since the fading coefficients can be reliably estimated with negligible rate penalty as $n\to \infty$ \citep{biglieri1998fading}\citep{bettesh2000outages}. Hence from this fact and using \citep[Theorem 5]{han1998information} it is easy to see that, for $0\leq \epsilon< 1$, the $\epsilon$--capacity region is given by
\begin{equation}
\label{eq:cap}
C_{\epsilon,J} =\left\{R=(R_1,...,R_K):\forall i,R_i\geq 0 \text{ and } P_{0}(R)\leq \epsilon\right\}
\end{equation}
where the outage probability $P_{0}(R)$ is given by
\begin{IEEEeqnarray}{L}
\label{eq:outage}
P_{0}(R)= \IEEEnonumber\\
\Pb{\bigcup_{S\subset[K],S\neq \emptyset}\left\{\log\left(1+P\sum_{i\in S}|H_i|^2\right)\leq \sum_{i\in S}R_i\right\}}
\end{IEEEeqnarray}

Next, we define a subspace projection based decoder, inspired from \citep{yang2014quasi}. The idea is the following. Suppose there were no additive noise. Then the received vector will lie in the subspace spanned by the sent codewords no matter what the fading coefficients are. To formally define the decoder, let $C$ denote a set of vectors in $\mathbb{C}^n$. Denote $P_{C}$ as the orthogonal projection operator onto the subspace spanned by $C$. Let $P^\perp_{C}=I-P_{C}$ denote the projection operator onto the orthogonal complement of $span(C)$ in $\mathbb{C}^n$. 

Let $\cd_1,...,\cd_K$ denote the codebooks of the $K$ users respectively. Upon receiving $Y$ from the channel the decoder outputs $g(Y)$ which is given by
\begin{IEEEeqnarray}{C}
\label{eq:dec1}
g(Y)=\left(f_1^{-1}(\hat{c}_1),...,f_K^{-1}(\hat{c}_{K})\right)\IEEEnonumber\\
(\hat{c}_1,...\hat{c}_K)=\arg\max_{(c_i\in\cd_i)_{i=1}^{K}}\norm{P_{\{c_i:i\in[K]\}}Y}^2
\end{IEEEeqnarray}
where $f_i$ are the encoding functions.

In this section, we show that using spherical codebook with projection decoding, $C_{\epsilon,J}$ of the $K$--MAC is achievable. We prove the following theorem

\begin{theorem}[Projection decoding achieves $C_{\epsilon,J}$]
\label{th:cap_joint}
Let $R\in C_{\epsilon,J}$ of \eqref{eq:sys1}. Then $R$ is $\epsilon$--achievable through a sequence of codes with the decoder being the projection decoder~\eqref{eq:dec1}.

\begin{proof}
We generate codewords iid uniformly on the power sphere and show that \eqref{eq:dec1} yields a small $P_{e,J}$. See appendix \ref{app:1} for details.
\end{proof}

\end{theorem}
\begin{remark} Note that~\cite{jin2011limits} also analyzed capacity region of the quasi-static MAC, but under the same codebook requirement, for the joint error probability (as opposed to PUPE), and with a different decoder. 
\end{remark}

\subsection{Per-user error}

In this subsection, we consider the case of per-user error under the classical setting. Further, we assume availability of CSI at receiver (CSIR) which again can be estimated with little penalty. 

The $\epsilon$--capacity region for the channel under per-user error, $C_{\epsilon,PU}$ is defined similarly as $C_{\epsilon,J}$ but with per-user error instead of joint error. $C_{\epsilon,PU}$ is unknown, but the best lower bound is given by the Shamai-Bettesh capacity bound  \citep{bettesh2000outages}: given a rate tuple $R=(R_1,...,R_K)$, an upper bound on the per-user probability of error under the channel \eqref{eq:sys1}, as $n\to\infty$, is given by
%{\color{blue}{
\begin{IEEEeqnarray*}{LLL}
\label{eq:pu3}
P_{e,u}& \leq & P_e^S(R) \\
&\equiv &  1-\frac{1}{K}\mathbb{E}\sup\left\{\vphantom{\frac{P'\sum_{i\in S}|H_{i}|^2}{1+P'\sum_{i\in D^c}|H_i|^2}} |D|:D\subset [K], \forall S\subset D, S\neq \emptyset,\right.\\
&& \left. \sum_{i\in S}R_i < \log\left(1+\frac{P\sum_{i\in S}|H_{i}|^2}{1+P\sum_{i\in D^c}|H_i|^2}\right) \right\}\IEEEyesnumber
\end{IEEEeqnarray*}
%}}
where the maximizing set, among all those that achieve the maximum, is chosen to contain the users with largest fading coefficients. The corresponding achievability region is 
\begin{IEEEeqnarray*}{LLL}
\label{eq:pu3a}
C^{S.B}_{\epsilon,PU}=\left\{R: P_e^S(R)\leq \epsilon \right\} \IEEEyesnumber
\end{IEEEeqnarray*}
and hence it is an inner bound on $C_{\epsilon,PU}$.

 We note that, in \citep{bettesh2000outages}, only the symmetric rate case i.e, $R_i=R_j\, \forall i,j$ is considered. So \eqref{eq:pu3} is the extension of that result to the general non-symmetric case.

Here, we show that the projection decoding (suitably modified to use CSIR) achieves the same asymptotics as \eqref{eq:pu3} for per-user probability of error i.e., achieves the Shamai-Bettesh capacity bound.  Next we describe the modification to the projection decoder to use CSIR.

Let $\{\mathcal{C}_i\}_{i=1}^{K}$ denote the codebooks of the $K$ users with $|\mathcal{C}_i|=M_i$. We have a maximum power constraint given by \eqref{eq:pow}. Using the idea of joint decoder from \citep{bettesh2000outages}, our decoder works in  2 stages. The first stage finds the following set
\begin{IEEEeqnarray*}{LLL}
\label{eq:pu1}
D\in & \arg \max \left\{\vphantom{\frac{P\sum_{i\in S}|H_{i}|^2}{1+P\sum_{i\in D^c}|H_i|^2}} |D|:D\subset [K], \forall S\subset D, S\neq \emptyset, \right.\\
&\left.\sum_{i\in S}R_i < \log\left(1+\frac{P\sum_{i\in S}|H_{i}|^2}{1+P\sum_{i\in D^c}|H_i|^2}\right) \right\} \IEEEyesnumber
\end{IEEEeqnarray*}
where $D$ is chosen to contain users with largest fading coefficients. The second stage is similar to \eqref{eq:dec1} but decodes only those users in  D. Formally, let $?$ denote an error symbol. The decoder output $g_{D}(Y)\in\prod_{i=1}^{K} \mathcal{C}_i$ is given by
\begin{IEEEeqnarray}{C}
\label{eq:pu2}
(g_D (Y))_i= \begin{cases} 
      f_i^{-1}(\hat{c}_i) & i \in D \\
      ? & i\notin D 
   \end{cases}\IEEEnonumber\\
(\hat{c}_i)_{i\in D}=\arg\max_{(c_i\in\cd_i)_{i\in D}}\norm{P_{\{c_i:i\in D\}}Y}^2
\end{IEEEeqnarray}
where $f_i$ are the encoding functions. Our error metric is the average per-user probability of error \eqref{eq:def2}.

The following theorem is the main result of this section.

\begin{theorem}
\label{th:pu1}
For any $R\in C^{S.B}_{\epsilon,PU}$ there exists a sequence of codes with projection decoder~\eqref{eq:pu1}\eqref{eq:pu2} with asymptotic rate $R$ such that the per-user probability of error is asymptotically smaller than $\epsilon$  

\begin{proof}
We generate iid (complex) Gaussian codebooks $\cn(0,P'I_n)$ with $P'<P$ and show that for $R\in C^{S.B}_{\epsilon,PU}$, \eqref{eq:pu2} gives small $P_{e,u}$. See appendix~\ref{app:1} for details.
\end{proof}

\end{theorem}

In the case of symmetric rate, an outer bound on $C_{\epsilon,PU}$ can be given as follows.

\begin{prop}
\label{prop:conv}
If the symmetric rate $R$ is such that $P_{e,u}\leq \epsilon$, then 
\begin{\Ieee}{LLL}
\label{eq:pu_conv}
R\leq  & \min\biggr\{\frac{1}{K(\theta-\epsilon)} \Ex{\log_2\left(1+P\min_{\substack{S\subset[K]\\ |S|=\theta K}}\sum_{i\in S}|H_i|^2\right)},\\%\right.\\
& \log_2\left(1-P\ln(1-\epsilon)\right)\biggr\},\, \forall\theta\in (\epsilon,1]\Ieeen
\end{\Ieee}

\begin{proof}

The first of the two terms in the min in \eqref{eq:pu_conv} follows from Fano's inequality (see \eqref{eq:sc_conv_fan7}, with $\mu=K/n$, $M=2^{nR}$ and taking $n\to\infty$). The second is a single-user based converse using a genie argument. See appendix \ref{app:per_user_converse} for details.
\end{proof}

\begin{remark}
We note here that the second term inside the minimum in \eqref{eq:pu_conv} is the same as the one we would obtain if we used strong converse for the MAC. To be precise, let $\{|H_{(1)}|>|H_{(2)}|>...>|H_{(K)}|\}$ denote the order statistics of the fading coefficients. If $R>\log(1+P|H_{(t)}|^2)$ then, using a Genie that reveals the codewords (and fading gains) of $t-1$ users corresponding to $t-1$ largest fading coefficients, it can be seen that $P_{e,u}\geq \frac{K-t+1}{K}$. Setting $t=\theta K$ and considering the limit as $K\to\infty$ (with $P=\PT/K$) we obtain $S\leq -\PT\log_2(1-\epsilon)$ which is same as that obtained from the second term in \eqref{eq:pu_conv} under these limits.

\end{remark}

\end{prop}

\subsection{Numerical evaluation}
First notice that $C_{\epsilon,J}$ (under joint error) tends to $\{0\}$ as $K\to\infty$ because, it can be seen, for the symmetric rate, by considering that order statistics of the fading coefficients that $P_0(R)\to 1$ for $R_i=O(1/K)$. $C_{\epsilon,PU}$, however, is more interesting. We evaluate trade-off between system spectral efficiency and the minimum energy-per-bit required for a target per-user error for the symmetric rate, in the limit $K\to\infty$ and power scaling as $O(1/K)$. 

\begin{figure}[H]
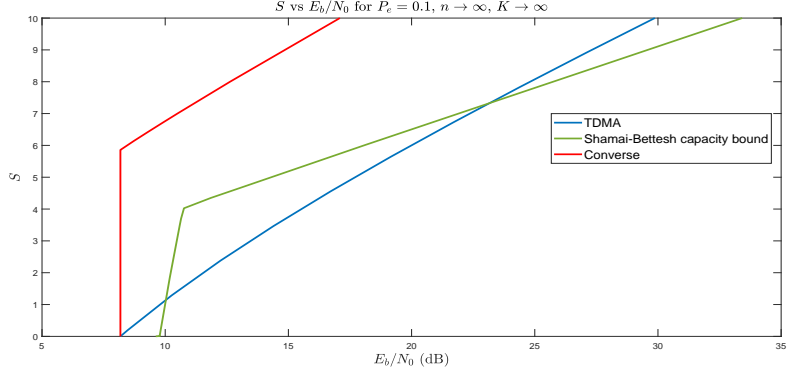

  \begin{center}
{\includegraphics *[height=5.2cm,width=0.7\columnwidth]{{ebn0_speff_shamai_new}}}
     \caption {$S$ vs $E_b/N_0$ for per-user error $\epsilon\leq 0.1$, $n\to\infty$ and then $K\to \infty$}
     \label{fig:0}
      \end{center} 
 \end{figure}

In the above figure we have also presented the performance of TDMA. That is, if we use orthogonalization then for any number of users $K$ (not necessarily large), we have
\begin{equation}
\label{eq:TDMA_classical_pupe}
 \epsilon = \Pb{R > 1/K \log (1+KP |H|^2)}
\end{equation} 
where $\epsilon$ is the PUPE. Thus the sum-rate vs $E_b/N_0$ formula for orthogonalization is

\begin{equation}
\label{eq:TDMA_ebn0}
E_b/N_0= \frac{2^S-1}{S}\frac{1}{-\ln (1-\epsilon)}
\end{equation}
where $S$ is the sum-rate or the spectral efficiency.

 We see that orthogonalization is suboptimal under the PUPE criterion. The reason is that it fails to exploit the multi-user diversity by allocating resources even to users in deep fades.  Indeed, under orthogonalized setting the resources allocated to a user that happens to experience a deep fade become completely wasted, while non-orthogonal schemes essentially adapt to the fading realization: the users in deep fades create very little interference for the problem of decoding strong users. This is the effect stemming from the PUPE criterion for error rate.

\section{Many user MAC: $K=\mu n$, $n\to\infty$}\label{sec:manymac}

This is our main section. We consider the linear scaling regime where the number of users $K$ scales with $n$, and $n\to\infty$. We are interested in the tradeoff of minimum $E_b/N_0$ required for the PUPE to be smaller than $\epsilon$, with the user density $\mu$ ($\mu<1$). So, we fix the message size $k$. Let $S=k\mu$ be the spectral efficiency.

We focus on the case of different codebooks, but under symmetric rate. So if $M$ denotes the size of the codebooks, then $S=\frac{K\log M}{n}=\mu\log M$. Hence, given $S$ and $\mu$, $M$ is fixed. Let $\PT=KP$ denote the total power. Therefore denoting by $\mathcal{E}$ the energy-per-bit, $\mathcal{E}=E_b/N_0=\frac{nP}{\log_2 M}=\frac{\PT}{S}$. For finite $E_b/N_0$, we need finite $\PT$, hence we consider the power $P$ decaying as $O(1/n)$. 

Let $\cd_j=\{c^j_1,...,c^j_M\}$ be the codebook of user $j$, of size $M$. The power constraint is given by $\norm{c^j_i}^2\leq nP=\mathcal{E}\log_2 M,\forall j\in[K],i\in[M]$. The collection of codebooks $\{\cd_j\}$ is called an $\left(n,M,\epsilon,\mathcal{E},K\right)$--code if it satisfies the power constraint described before, and the per-user probability of error is smaller than $\epsilon$. Then, we can define the following fundamental limit for the channel
\begin{\Ieee}{LLL}
\label{eq:sc_conv_1}
\mathcal{E}^*(M,\mu,\epsilon)=\lim_{n\to\infty} \inf\left\{\mathcal{E}:\exists (n,M,\epsilon,\mathcal{E},K=\mu n)-code\right\}.
\end{\Ieee}

We make an important remark here that all the following results also hold for maximal per-user error (PUPE-max) \eqref{eq:pupe_max} as discussed in appendix \ref{app:pupe-max}.

\subsection{No-CSI}

In this subsection, we focus on the no-CSI case. The difficulty here is that, a priori, we do not know which subset of the users to decode. We have the following theorem.

\begin{theorem}
\label{th:scaling_noCSI}
Consider the channel \eqref{eq:sys1} (no-CSI) with $K=\mu n$ where $\mu< 1$. Fix the spectral efficiency $S$ and target probability of error (per-user) $\epsilon$. Let $M=2^{S/\mu}$ denote the size of the codebooks and $\PT=KP$ be the total power. \iffalse Let $h(p)=-p\ln p-(1-p)\ln(1-p),p\in[0,1]$.\fi Fix $\nu \in (1-\epsilon,1]$. Let $\epsilon'=\epsilon-(1-\nu)$. Then if $\mathcal{E}>\mathcal{E}^*_{no-CSI}=\sup_{\frac{\epsilon'}{\nu}< \theta\leq 1}\sup_{\xi\in[0,\nu(1-\theta)]}\frac{\PTX(\theta,\xi)}{S}$, there exists a sequence of $\left(n,M,\epsilon_n,\mathcal{E},K=\mu n\right)$ codes such that $\limsup_{n\to\infty} \epsilon_n\leq \epsilon$, where, for $\frac{\epsilon'}{\nu}< \theta\leq 1$ and $\xi\in[0,\nu(1-\theta)]$, 
{\allowdisplaybreaks
\begin{IEEEeqnarray*}{RLL}
\PTX(\theta,\xi)&=&\frac{\hat{f}(\theta,\xi)}{1-\hat{f}(\theta,\xi)\alpha\left(\xi+\nu\theta,\xi+1-\nu(1-\theta)\right)}\label{eq:sc_nc1a}\Ieeen\\
\hat{f}(\theta,\xi)&=&\frac{f(\theta)}{\alpha(\xi,\xi+\nu\theta)}\label{eq:sc_nc1b}\Ieeen\\
f(\theta)&=& \frac{\frac{1+\delta_1^*(1-V_\theta)}{V_\theta}-1}{1-\delta_2^*}\label{eq:sc_nc1c}\Ieeen\\
V_{\theta}&=&e^{-\tilde{V}_{\theta}}\label{eq:sc_nc1d}\Ieeen\\
\tilde{V}_{\theta} &=&\delta^*+\frac{\theta\mu\nu\ln M}{1-\mu\nu}+\\
&&\frac{1-\mu\nu(1-\theta)}{1-\mu\nu}h\left(\frac{\theta\mu\nu}{1-\mu\nu(1-\theta)}\right) + \\
&&\frac{\mu(1-\nu(1-\theta))}{1-\mu\nu}h\left(\frac{\theta\nu}{1-\nu(1-\theta)}\right)\label{eq:sc_nc1e}\Ieeen\\
\delta^* &=&\frac{\mu h(1-\nu(1-\theta))}{1-\mu\nu}\label{eq:sc_nc1f}\Ieeen\\
c_\theta &=&\frac{2V_{\theta}}{1-V_{\theta}}\label{eq:sc_nc1g}\Ieeen\\
q_\theta &=&\frac{\mu h(1-\nu(1-\theta))}{1-\mu\nu(1-\theta)}\label{eq:sc_nc1h}\Ieeen\\
\delta_1^*&=& q_\theta(1+c_\theta)+\\
&& \sqrt{q_\theta^2 (c_\theta^2+2c_\theta)+2q_\theta(1+c_\theta)}\label{eq:sc_nc1i}\Ieeen\\
\delta_2^* &=& \inf \biggr\{x: 0<x<1, \\
&& -\ln(1-x)-x>\frac{\mu h(1-\nu(1-\theta))}{1-\mu\nu(1-\theta)}\biggr\}\label{eq:sc_nc1j}\Ieeen\\
%\sqrt{\frac{2\mu h(1-\nu(1-\theta))}{1-\mu\nu(1-\theta)}}
\alpha(a,b)&=& a\ln(a)-b\ln(b)+b-a. \label{eq:sc_nc1k}\Ieeen
\end{IEEEeqnarray*}

Hence $\mathcal{E}^*\leq \mathcal{E}^*_{no-CSI}$.
}
%{\color{blue}{
\begin{proof}[Proof Idea]
Before we present the full proof, the main ideas are presented here. Also, over the course, we explain the quantities that are present in the statement of the theorem. We start with choosing independent random Gaussian codebooks for all users. That is, for each message of each user there is an independent complex Gaussian $\cn(0,P'I_n)$ codeword where $P'<P$. The choice $P'<P$ is to ensure we can control the maximum power constraint violation events. 

For simplicity we will consider $\nu=1$. Here $\nu$ represents the fraction of users that the decoder can choose to decode. Due to random coding, we can assume that a particular tuple of codewords $(c_1,c_2,\cdots,c_K)$ were transmitted i.e., the received vector at the decoder is $Y=\sum_{i=1}^K H_i c_i +Z$. Then the decoder performs subspace projection decoding. The idea is that in the absence of noise, the received vector lies in the subspace spanned by the $K$ codewords. Since we assume $\mu=K/n<1$, and the $K$ codewords are linearly independent, we can uniquely decode them by projecting the received vector onto various $K$ dimensional subspaces formed by taking a codeword from each of the codebooks. Formally,

\begin{equation}
\label{eq:rev1}
\left\{\hat{c}_i:i\in[K]\right\}=\arg\max_{\left(c_i\in\cd_i:i\in[K]\right)}\norm{P_{\{c_i:i\in S\}}Y}^2
\end{equation}

Notice that the PUPE is given by 
$$P_e=\frac{1}{K}\sum_{i\in[K]}\Pb{c_i\neq \hat{c}_i}.$$
 We will bound this error with the probability of events $F_t$-- event that exactly $t$ users were misdecoded. That is 
\begin{equation}
\label{eq:rev2}
P_e\leq \epsilon +\Prb{\bigcup_{t>\epsilon K}F_t}
\end{equation}
Hence it is enough to find conditions under which the second term (call it $p_1$) in the above display goes to $0$ in our scaling. To analyze $F_t$, we consider subsets $S\subset[K]$ with $|S|=t$ and a choice of {\it incorrect} codewords $(c'_i\in\cd_i:i\in S)$ where $c'_i\neq c_i$, and bound $F_t$ as union (over $S$ and $(c'_i:i\in [S])$) of events $\left\{\norm{\Pd{S}{[K]\setminus S}Y}^2>\norm{\Pc{[K]}Y}^2\right\}$. With abuse of notation, denote this set as $F(S,t)$. 

Let $c_{[S]}=\{c_i:i\in S\}$, similarly we have $H_{[S]}$. We make a crucial observation that, conditioned on $c_{[K]}$, $H_{[K]}$ and $Z$, the random variable $\norm{\Pd{S}{[K]\setminus S}Y}^2$ can be written as $\norm{\Pc{K\setminus S}Y}^2 +\norm{\Pcp{K\setminus S}Y}^2\mathrm{Beta}(t,n-K)$ where $\mathrm{Beta}(a,b)$ is a beta distributed random variable with parameters $a$ and $b$.

Let $G_S=\frac{\norm{\Pcp{[K]}Y}^2}{\norm{\Pcp{K\setminus S}Y}^2}$. Then we show that 
\begin{\Ieee}{LLL}
\label{eq:rev3}
&&\Pb{F(S,t)|c_{[K]},H_{[K]},Z}\\
&=&\Pb{\mathrm{Beta}(n-K,t)<G_S | c_{[K]},H_{[K]}, Z}\\
&\leq & \binom{n-K+t-1}{t-1}\left(G_S\right)^{n-K}\Ieeen
\end{\Ieee}

Next, we use the idea of random coding union (RCU) bound \cite{polyanskiy2010channel_ieee} to get
\begin{equation}
\label{eq:rev4}
\Prb{\bigcup_t F_t}\leq \Exa{\min\biggr\{1,\sum_{t,S}\Pb{F(S,t)|c_{[K]},H_{[K]},Z}\biggr\}}
\end{equation}

Let $\theta=t/K$, which is the fraction of misdecoded users. Now, by thresholding the value of $G_S$ (this threshold is parameterized by a $\delta>0$) we get from \eqref{eq:rev4} a sum of an exponentially decaying term with combinatorial factors and the probability that $G_S$ violates this threshold for some $S$ and $t$ (call this probability $p_2$). Choosing the right threshold ($\delta^*$ and corresponding threshold value $V_{\theta}$ in the theorem) the first term vanishes (in the limit) and we are left with $p_2$.

This is analyzed by conditioning on $c_{[K]}$ and $H_{[K]}$ along with using concentration of non-central chi-squared distributed variables (see claim \ref{claim:sc_nc3}). We follow similar procedure to above (using RCU and thresholding) multiple times to obtain thresholds parameterized by $\delta_1^*$ and $\delta_2^*$ to vanish combinatorial factors (like $q_{\theta}$ in the theorem which is the exponent of a binomial coefficient) and finally we are left with the {\it bottleneck} term:
\begin{\Ieee}{LLL}
\label{eq:rev5}
\limsup_n P_e\leq \epsilon + \\
\limsup_n \Prb{\bigcup_{t,S}\biggr\{P'\sum_{i\in[S]}|H_i|^2<	g\left(\delta_1^*,\delta_2^*,\delta_3^*,M,\mu,\theta\right)\biggr\}}\\
\Ieeen
\end{\Ieee}
where $g$ is some specific function. In essence, this bottleneck term is precisely the event that $>\epsilon$ fraction of users are outside the Gaussian capacity region!

Next step is to replace $\cup_S$ with $\min_S$ and use the convergence of order statistics of fading coefficients i.e.,  $|H_{(1)}|>\cdots>|H_{(K)}|$:
\begin{\Ieee}{LLL}
\label{eq:rev6}
\limsup_n P_e &\leq &  \epsilon + \limsup_n \Prb{\bigcup_{t}\biggr\{P'\sum_{i=K-t+1}^{K}|H_{(i)}|^2<\\
&&	g\left(\delta_1^*,\delta_2^*,\delta_3^*,M,\mu,t/K\right)\biggr\}}\\
 \Ieeen
\end{\Ieee}

Then we show that, for $t=\theta K$,  $\frac{1}{K}\sum_{i=K-t+1}^{K}|H_{(i)}|^2\to \int_{1-\theta}^1 F^{-1}_{|H|^2}(1-\gamma)\mathop{d\gamma}\equiv \alpha(1-\theta,1)$ in probabilty as $n\to\infty$. Hence the bottleneck term becomes deterministic in the limit. The choice $\PT$ such that this terms vanishes is precisely the one given in the statement of the theorem. 
\end{proof}

%}}

\begin{proof}

The proof uses random coding. Let each user generate a Gaussian codebook of size $M$ and power $P'<P$ independently such that $KP'=\PT'<\PT$. Let $W_j$ denote the random (in $[M]$) message of user $j$. So, if $\cd_j=\{c^j_i:i\in[M]\}$ is the codebook of user $j$, he/she transmits $X_j=c^j_{W_j}1\left\{\norm{c^j_{W_j}}^2\leq nP\right\}$. For simplicity let $(c_1,c_2,...,c_K)$ be the sent codewords. Hence the received vector is $Y=\sum_{i\in[K]}H_i c_i +Z$ where $Z$ is the noise vector. Fix $\nu \in (1-\epsilon,1]$. Let $K_1=\nu K$ be the number of users that are decoded. Since there is no knowledge of CSIR, it is not possible to, a priori, decide what set to decode. Instead, the decoder searches of all $K_1$ sized subsets of $[M]$. Formally, let $?$ denote an error symbol. The decoder output $g_{D}(Y)\in\prod_{i=1}^{k} \mathcal{C}_i$ is given by
\begin{IEEEeqnarray*}{C}
\label{eq:sc_nc2}
\left[\hat{S},(\hat{c}_i)_{i\in \hat{S}}\right]=\arg\max_{\substack{S\subset [K]\\|S|=K_1}}\max_{(c_i\in\cd_i)_{i\in S}}\norm{P_{\{c_i:i\in S\}}Y}^2\\
(g_D (Y))_i= \begin{cases} 
      f_i^{-1}(\hat{c}_i) & i \in \hat{S}\\
      ? & i\notin \hat{S} 
   \end{cases}\Ieeen\\
\end{IEEEeqnarray*}
where $f_i$ are the encoding functions. The probability of error (averaged over random codebooks) is given by
\begin{IEEEeqnarray}{LLL}
\label{eq:sc_nc3}
P_e=\frac{1}{K} \sum_{j=1}^{K}\Pb{W_j\neq \hat{W}_j}
\end{IEEEeqnarray}
where $\hat{W}_j=(g(Y))_j$ is the decoded message of user $j$.

 We perform a change of measure to $X_j=c^j_{W_j}$. Since $P_e$ is the expectation of a non-negative random variable bounded by $1$, this measure change adds a total variation distance which can bounded by $p_0=K\Pb{\frac{\chi_2(2n)}{2n}>\frac{P}{P'}}\to 0$ as $n\to \infty$, where $\chi_2(d)$ is the distribution of sum of squares of $d$ iid standard normal random variables (the chi-square distribution). The reason is as follows. If we have two random vectors $U_1$ and $U_2$ on a the same probability space such that $U_1=U_2 1[U_2\in E]$, where $E$ is a Borel set, then for any Borel set A, we have 
\begin{\Ieee}{LLL}
\label{eq:sc_nc3a}
&&|\Pb{U_1\in A}-\Pb{U_2\in A}|\\
&=&|1[0\in A]\Pb{U_2\in E^c}-\Pb{U_2\in A\cap E^c}|\\
&\leq & \Pb{U_2\in E^c}.\Ieeen
\end{\Ieee}

 Henceforth we only consider the new measure.  
 
Let $\epsilon>1-\nu$ and $\epsilon'=\epsilon-(1-\nu)$. Now we have
\begin{\Ieee}{LLL}
\label{eq:sc_nc4}
P_e\leq \epsilon +\Prb{\frac{1}{K} \sum_{j=1}^{K}1[W_j\neq \hat{W}_j] > \epsilon}\\
=\epsilon +\Prb{ \sum_{j=1}^{K}1[W_j\neq \hat{W}_j] > K\epsilon'+K-K_1}\\
=\epsilon +p_1.\Ieeen
\end{\Ieee}

where 
$$p_1=\Prb{\bigcup_{t=\epsilon'K}^{\nu K}\biggr\{\sum_{j=1}^{K}1[W_j\neq \hat{W}_j] =K-K_1+t\biggr\}}.$$

Let $F_t=\left\{\sum_{j=1}^{K}1[W_j\neq \hat{W}_j] =K-K_1+t\right\}$. Let $c_{[S]}\equiv \{c_i:i\in S\}$ and $H_{[S]}\equiv \{H_i:i\in S\}$, where $S\subset [K]$. Conditioning on $c_{[K]}, H_{[K]}$ and $Z$, we have

\begin{\Ieee}{LLL}
\label{eq:sc_nc5}
&&\Pb{F_t|c_{[K]},H_{[K]},Z}\\
&\leq &  \mathbb{P}\biggr[\exists S\subset [K]: |S|=K-K_1+t, \exists S_1\subset S: |S_1|=t, \\
&& \exists\{c'_i\in \cd_i:i\in S_1, c'_i\neq c_i\} :  \norm{\Pd{S_1}{[K]\setminus S}Y}^2> \\
&&\max_{\substack{S_2\subset S\\|S_2|=t}}\norm{\Pdd{S_2}{[K]\setminus S}Y}^2  \biggr | c_{[K]},H_{[K]},Z \biggr]\\
 &\leq & \Prb{\bigcup_{\substack{S\subset [K]\\ |S|=K-K_1+t}}\bigcup_{\substack{S_1\subset S\\|S_1|=t}}\bigcup_{\substack{\{c'_i\in\cd_i:\\ i\in S_1, c'_i\neq c_i\}}} \\
&&  F(S,S_2^*,S_1,t) \biggr |c_{[K]},H_{[K]},Z } \Ieeen
\end{\Ieee}
where 
\begin{\Ieee}{LLL}
F(S,S_2^*,S_1,t)=\\
\biggr\{  \norm{\Pd{S_1}{[K]\setminus S}Y}^2>\norm{\Pdd{S_2^*}{[K]\setminus S}Y}^2\biggr\}
\end{\Ieee}
 and $S_2^*\subset S$ is a possibly random (depending only on $H_{[K]}$) subset of size $t$, to be chosen later. Next we will bound $\Pb{F(S,S_2^*,S_1,t)|c_{[K]},H_{[K]},Z}$.\\

For the sake of brevity, let $A_0=c_{[S_2^*]}\cup c_{[[K]\setminus S]}$, $A_1=c_{[[K]\setminus S]}$ and $B_1=c'_{[S_1]}$. We have the following claim.

\begin{claim}
\label{claim:sc_nc1}
For any $S_1\subset S$ with $|S_1|=t$, conditioned on $c_{[K]}$, $H_{[K]}$ and $Z$, the law of $\norm{\Pd{S_1}{[K]\setminus S}Y}^2$ is same as the law of $\norm{P_{A_1}Y}^2+\norm{(I-P_{A_1})Y}^2 \emph{Beta}(t,n-K_1)$ where $\emph{Beta}(a,b)$ is a beta distributed random variable with parameters $a$ and $b$.
\begin{proof}
 Let us write $V=span\{A_1,B_1\}=A\oplus B$ where $A\perp B$ are subspaces of dimension $K_1-t$ and $t$ respectively, with $A=span(A_1)$ and $B$ is the orthogonal complement of $A_1$ in $V$. Hence $\norm{P_V Y}^2=\norm{P_A Y}^2+\norm{P_B Y}^2$ (by definition, $P_A=P_{A_1}$). Now we analyze $\norm{P_B Y}^2$. We can further write $P_B Y=P_B \PP_A Y$. Observe that the subspace $B$ is the span of $\PP_A B_1$, and, conditionally, $\PP_A B_1\distas{} \cn^{\otimes |S|}(0,P'\PP_A)$ which is the product measure of $|S|$ complex normal vectors in a subspace of dimension $n-K_1+t$. Hence, the conditional law of $\norm{P_B \PP_A Y}^2$ is the law of squared length of projection of a fixed $n-K_1+t$ dimensional vector of length $\norm{(I-P_A)Y}^2$ onto a (uniformly) random $t$ dimensional subspace.

Further, the law of the squared length  of the orthogonal projection of a fixed unit vector in $\mathbb{C}^d$ onto a random $t$--dimensional subspace is same as the law of the squared length  of the orthogonal projection of a random unit vector in $\mathbb{C}^d$ onto a fixed $t$--dimensional subspace, which is $\mathrm{Beta}(t,d-t)$ (see for e.g. \citep[Eq. 79]{yang2013quasi}):  that is, if $u$ is a unit random vector in $\mathbb{C}^d$ and $L$ is a fixed $t$ dimensional subspace, then $\Pb{\norm{P_L u}^2\leq x}=\Pb{\frac{\sum_{i=1}^t |Z_i|^2}{\sum_{i=1}^d |Z_i|^2}\leq x}= F_{\beta}(x;t,n-K_1)$ where $Z_i\distas{iid}\cn(0,1)$ and $F_\beta (x;a,b)=\frac{\Gamma(a+b)}{\Gamma(a)\Gamma(b)}\int_{0}^{x}w^{a-1}(1-w)^{b-1}dw$ denotes the CDF of the beta distribution with parameters $a$ and $b$. Hence the conditional law of $\norm{P_B \PP_A Y}^2$ is  $\norm{(I-P_A)Y}^2 \mathrm{Beta}(t,n-K_1)$.

\end{proof}
\end{claim}

Therefore we have,
\begin{\Ieee}{LLL}
\label{eq:sc_nc6}
\Pb{F(S,S_2^*,S_1,t)|c_{[K]},H_{[K]},Z} \\
= \Pb{Beta(n-K_1,t)<G_S|c_{[K]},H_{[K]},Z }\\
=F_{\beta}\left(G_S;n-K_1,t\right)\Ieeen
\end{\Ieee}
where 
\begin{\Ieee}{LLL}
\label{eq:sc_nc7}
G_S=\frac{\norm{Y}^2-\norm{P_{A_0}Y}^2}{\norm{Y}^2-\norm{P_{A_1}Y}^2}.\Ieeen
\end{\Ieee}
Since $t\geq 1$, we have 
$$F_\beta \left(G_S;n-K_1,t\right)\leq \binom{n-K_1+t-1}{t-1}G_S^{n-K_1}.$$

Let us denote $\bigcup_{t=\epsilon'K}^{\nu K}$ as $\bigcup_t$, $\bigcup_{\substack{S\subset[K]\\|S|=K-K_1+t}}$ as $\bigcup_{S,K_1}$, and $\bigcup_{t}\bigcup_{\substack{S\subset[K]\\|S|=K-K_1+t}}$ as $\bigcup_{t,S,K_1}$; similarly for $\sum$ and $\bigcap$ for the ease of notation. Using the above claim, we get,

\begin{\Ieee}{LLL}
\label{eq:sc_nc8}
\Pb{F_t|c_{[K]},H_{[K]},Z}\leq \\
\sum_{S,K_1}\binom{K-K_1+t}{t}M^t\binom{n-K_1+t-1}{t-1}G_S^{n-K_1}.\Ieeen
\end{\Ieee}

Therefore $p_1$ can be bounded as
{\allowdisplaybreaks
\begin{\Ieee}{LLL}
\label{eq:sc_nc9}
p_1 & = &\Prb{\bigcup_t F_t} \\
&\leq & \Exa{\min\biggr\{1,\sum_{t,S,K_1}\binom{K-K_1+t}{t}M^t\cdot\\
&&\binom{n-K_1+t-1}{t-1}G_S^{n-K_1} \biggr\}}\\
&=&  \Exa{\min\biggr\{1,\sum_{t,S,K_1} e^{(n-K_1)s_t} M^t G_S^{n-K_1} \biggr\}}\Ieeen
\end{\Ieee}
}
where $s_t=\frac{\ln\left(\binom{K-K_1+t}{t}\binom{n-K_1+t-1}{t-1}\right)}{n-K_1}$.

Now we can bound the binomial coefficient~\citep[Ex. 5.8]{Gallager:1968:ITR:578869} as

\begin{IEEEeqnarray*}{LLL}
\label{eq:sc_nc10}
&&\binom{n-K_1+t-1}{t-1} \\
&\leq & \sqrt{\frac{n-K_1+t-1}{2\pi (t-1)(n-K_1)}}e^{(n-K_1+t-1)h(\frac{t-1}{n-K_1+t-1})}\\
& =& O\left(\frac{1}{\sqrt{n}}\right)e^{n(1-\mu\nu(1-\theta))h(\frac{\theta\mu\nu}{1-\mu\nu(1-\theta)})}.\IEEEyesnumber
\end{IEEEeqnarray*}

Similarly, 
\begin{\Ieee}{LLL}
\label{eq:sc_nc11}
&&\binom{K-K_1+t}{t}\\
&\leq & O\left(\frac{1}{\sqrt{n}}\right)e^{n\mu(1-\nu(1-\theta))h\left(\frac{\theta \nu}{1-\nu(1-\theta)}\right) }\Ieeen
\end{\Ieee}

Let $r_t=s_t+\frac{t\ln M}{n-K_1}$. For $\delta>0$, define $\tilde{V}_{n,t}=r_t+\delta$ and $V_{n,t}=e^{-\tilde{V}_{n,t}}$. Let $E_1$ be the event
\begin{\Ieee}{LLL}
\label{eq:sc_nc12}
E_1 &=& \bigcap_{t,S,K_1}\left\{-\ln G_S - r_t>\delta\right\}\\
&=& \bigcap_{t,S,K_1}\left\{G_S<V_{n,t}\right\}.\Ieeen
\end{\Ieee}

Let $p_2=\Pb{\bigcup_{t,S,K_1}\left\{G_S>V_{n,t}\right\}}$. Then

{\allowdisplaybreaks
\begin{\Ieee}{LLL}
\label{eq:sc_nc13}
&& p_1 \\
&\leq & \Exa{\min\biggr\{1,\sum_{t,S,K_1}e^{(n-K_1)r_t}G_S^{n-K_1}\biggr\}\left(1[E_1]+1[E_1^c]\right)}\\
&\leq &  \Ex{\sum_{t,S,K_1} e^{-(n-K_1)\delta}}+p_2\\
&=& \sum_{t} \binom{K}{K-K_1+t} e^{-(n-K_1)\delta}+p_2.\IEEEyesnumber
\end{\Ieee}
}

Observe that, for $t=\theta K_1=\theta\nu K$, 
\begin{\Ieee}{LLL}
s_t &=&\frac{1-\mu\nu(1-\theta)}{1-\mu\nu}h\left(\frac{\theta\mu\nu}{1-\mu\nu(1-\theta)}\right)+\\
&&\frac{\mu(1-\nu(1-\theta))}{1-\mu\nu}h\left(\frac{\theta \nu}{1-\nu(1-\theta)}\right)-O\left(\frac{\ln(n)}{n}\right)
\end{\Ieee}
 and $r_t=s_t+\frac{\theta\mu\nu}{1-\mu\nu}\ln M$. Therefore $n\to\infty$ with $\theta$ fixed, we have

\begin{IEEEeqnarray*}{LLL}
\label{eq:sc_nc14}
 \lim_{n\to\infty} \tilde{V}_{n,\theta\nu\mu}=\tilde{V}_{\theta}\IEEEyesnumber
\end{IEEEeqnarray*}
where $\tilde{V}_{\theta}$ is given in \eqref{eq:sc_nc1e}.

Now, note that, for $1<t<K_1$, 
\begin{\Ieee}{LLL}
\label{eq:sc_nc15}
&&\binom{K}{K-K_1+t} \\
&\leq & \sqrt{\frac{K}{2\pi (K-K_1+t) (K_1-t)}} e^{K h(\frac{K-K_1+t}{K})}.\Ieeen
\end{\Ieee}
Hence choosing $\delta>\frac{K h(\frac{K-K_1+t}{K})}{n-K_1} $ will ensure that the first term in \eqref{eq:sc_nc13} goes to $0$ as $n\to\infty$. So for $t=\theta K_1=\theta\nu K$, we need to have

\begin{equation}
\label{eq:sc_nc16}
\delta >  \delta^*.
\end{equation}
where $\delta^*$ is given in \eqref{eq:sc_nc1f}.

Let us bound $p_2$. Let $\hat{Z}=Z+\sum_{i\in S\setminus S^*_2}H_i c_i$. We have

\begin{claim}
\label{claim:sc_nc2}
\begin{IEEEeqnarray*}{LLL}
\label{eq:sc_nc17}
 p_2 &=&\Prb{\bigcup_{t,S,K_1}\left\{G_S>V_{n,t}\right\}}\\
& \leq &  \Prb{\bigcup_{t,S,K_1}\biggr\{\norma{(1-V_{n,t})\PPA \hat{Z}-V_{n,t}\PPA \sum_{i\in S_2^*}H_i c_i}^2  \\
&&\geq V_{n,t}\norma{\PPA \sum_{i\in S^*_2}H_i c_i}^2\biggr\}}.\IEEEyesnumber
\end{IEEEeqnarray*}

\begin{proof}
See appendix \ref{app:2}.
 
\end{proof}
\end{claim}

Let $\chi'_2(\lambda,d)$ denote the non-central chi-squared distributed random variable with non-centrality $\lambda$ and degrees of freedom $d$. That is, if $W_i\distas{}\mathcal{N}(\mu_i,1),i\in[d]$ and $\lambda=\sum_{i\in[d]}\mu_i^2$, then $\chi'_2(\lambda,d)$ has the same distribution as that of $\sum_{i\in[d]}W_i^2$. We have the following claim.

\begin{claim}
\label{claim:sc_nc3}
Conditional on $H_{[K]}$ and $A_0$,
 \begin{IEEEeqnarray*}{LLL}
\label{eq:sc_nc18a}
&&\norma{\PPA\biggr(\hat{Z}-\frac{V_{n,t}}{1-V_{n,t}}\hcc\biggr)}^2 \\
&\distas{}&\phsa\frac{1}{2}\chi'_2\left(2F,2 n'\right)\IEEEyesnumber
\end{IEEEeqnarray*}
where 
\begin{IEEEeqnarray}{LLL}
F=\frac{\norm{\frac{V_{n,t}}{1-V_{n,t}}\PPA\hcc}^2}{\phs}\label{eq:sc_nc18ba}\\
n'=n-K_1+t\label{eq:sc_nc18bb}.
\end{IEEEeqnarray}

Hence its conditional expectation is
\begin{equation}
\label{eq:eq:sc_nc18c}
\mu=n'+F.
\end{equation}
\begin{proof}
See appendix \ref{app:2}.
\end{proof}
\end{claim}

Now let
\begin{IEEEeqnarray}{LLL}
T=\frac{1}{2}\chi'_2(2F,2n')-\mu\label{eq:sc_nc19}\\
U=\frac{V_{n,t}}{(1-V_{n,t})}\frac{\norm{\PPA\hcc}^2}{\phs}-n' \label{eq:sc_nc20}\\
U^1=\frac{1}{1-V_{n,t}}\left(V_{n,t}W_S-1\right)\label{eq:sc_nc21}\IEEEeqnarraynumspace
\end{IEEEeqnarray}
where 
$$W_S= \biggr(1+\frac{\norm{\PPA\hcc}^2}{n'\phs}\biggr).$$
 Notice that $U=n'U^1$ and $F=\frac{V_{n,t}}{1-V_{n,t}}n'(1+U^1)$.

Then we have

\begin{\Ieee}{LLL}
\label{eq:sc_nc22}
%&\Pb{\bigcup_{t,S,D}\left\{\norm{\PPA Z_D-\frac{V_{n,t}}{(1-V_{n,t})}\PPA \sum_{i\in S}H_i c_i}^2 \right.\right.\\
%& \geq \left.\left. \frac{V_{n,t}}{(1-V_{n,t})^2}\norm{\PPA \sum_{i\in S}H_i c_i}^2\right\}} \\
&&\mathrm{RHS\ of\ \eqref{eq:sc_nc17}}\\
&=& \Prb{\bigcup_{t,S,K_1}\biggr\{\norma{\PPA \hat{Z}-\frac{V_{n,t}}{(1-V_{n,t})}\PPA \hcc}^2 \\
&& -\mu \geq U\biggr\}} \\
&=&\Prb{ \bigcup_{t,S,K_1}\left\{T\geq U\right\} }. \Ieeen
\end{\Ieee}

Now, let $\delta_1>0$, and $E_2=\cap_{t,S,K_1}\left\{U^1>\delta_1\right\}$. Taking expectations over $E_1$ and its complement, we have
\begin{\Ieee}{LLL}
\label{eq:sc_nc23}
&&\Prb{ \bigcup_{t,S,K_1}\left\{T\geq U\right\} } \\
&\leq &\sum_{t,S,K_1} \Pb{T>U, U^1>\delta_1}+\Pb{E_2^c}\\
&= &\sum_{t,S,K_1}\Ex{\Pb{\left.T>U\right| H_{[K]},A_0}1[U^1>\delta_1]}\\
&&+\Pb{E_2^c}\Ieeen
\end{\Ieee}
which follows from the fact that $\{U^1> \delta_1\}\in \sigma(H_{[K]},A_0)$. To bound this term, we use the following concentration result from \citep[Lemma 8.1]{birge2001alternative}.

\begin{lemma}[\citep{birge2001alternative}]
\label{lem:chi2}
Let $\chi=\chi_2'(\lambda,d)$ be a non-central chi-squared distributed variable with $d$ degrees of freedom and non-centrality parameter $\lambda$. Then $\forall x>0$
\begin{equation}
\label{eq:sc_nc24}
\begin{split}
& \Pb{\chi-(d+\lambda)\geq 2\sqrt{(d+2\lambda)x}+2x}\leq e^{-x}\\
& \Pb{\chi-(d+\lambda)\leq  -2\sqrt{(d+2\lambda)x}} \leq e^{-x}
\end{split}
\end{equation}
\end{lemma}

Hence, for $x>0$, we have
\begin{equation}
\label{eq:sc_nc25a}
\Pb{\chi-(d+\lambda)\geq x}\leq e^{-\frac{1}{2}\left(x+d+2\lambda-\sqrt{d+2\lambda}\sqrt{2x+d+2\lambda}\right)}.
\end{equation}
and for $x<(d+\lambda)$, we have
\begin{equation}
\label{eq:sc_nc25b}
\Pb{\chi\leq x}\leq e^{-\frac{1}{4}\frac{(d+\lambda-x)^2}{d+2\lambda}}.
\end{equation}

Observe that, in \eqref{eq:sc_nc25a}, the exponent is always negative for $x>0$ and finite $\lambda$ due to AM-GM inequality. When $\lambda=0$, we can get a better bound for the lower tail in \eqref{eq:sc_nc25b} by using \citep[Lemma 25]{reeves2012sampling}.
\begin{lemma}[\citep{reeves2012sampling}]
Let $\chi=\chi_2(d)$ be a chi-squared distributed variable with $d$ degrees of freedom. Then $\forall x>1$
\begin{\Ieee}{LLL}
\label{eq:sc_nc25c}
\Pb{\chi\leq \frac{d}{x}}\leq e^{-\frac{d}{2}\left(\ln x +\frac{1}{x}-1\right)}\Ieeen
\end{\Ieee}
\end{lemma}

Therefore, from \eqref{eq:sc_nc17}, \eqref{eq:sc_nc22}, \eqref{eq:sc_nc23} and \eqref{eq:sc_nc25a},  we have

\begin{IEEEeqnarray*}{LLL}
\label{eq:sc_nc26}
p_2 &\leq & \sum_{t,S,K_1}\Ex{e^{-n'f_n(U^1)}1[U^1>\delta_1]}+\\
&&\Prb{\bigcup_{t,S,K_1}\biggr\{U^1\leq \delta_1\biggr\}}\Ieeen
\end{IEEEeqnarray*}
where $f_n$ is given by
\begin{IEEEeqnarray*}{LL}
\label{eq:sc_nc27}
f_n(x)=  x+1+\frac{2V_{n,t}}{1-V_{n,t}}(1+x)\\
 -\sqrt{1+\frac{2V_{n,t}}{1-V_{n,t}}(1+x)}\sqrt{2x+1+\frac{2V_{n,t}}{1-V_{n,t}}(1+x)}.\IEEEyesnumber \IEEEeqnarraynumspace
\end{IEEEeqnarray*}

Next, we have the following claim.
\begin{claim}
\label{claim:sc_nc4}
For $0<V_{n,t}<1$ and $x>0$, $f_n(x)$ is a monotonically increasing function of $x$.
\begin{proof}
See appendix \ref{app:2}.
\end{proof}
\end{claim}

From this claim, we get
\begin{IEEEeqnarray*}{LLL}
\label{eq:sc_nc28}
p_2\leq \sum_{t,S,K_1}e^{-n'f_n(\delta_1)}+p_3\Ieeen
\end{IEEEeqnarray*}
where $p_3=\Pb{E_2^c}$.

Now, if, for each $t$, $\delta_1$ is chosen such that $f_n(\delta_1)>\frac{K h(\frac{K-K_1+t}{K})}{n-K_1+t}$, then the first term in \eqref{eq:sc_csir11} goes to $0$ as $n\to \infty$. Therefore, for $t=\theta K_1$, setting $c_\theta$ and $q_\theta$ as in \eqref{eq:sc_nc1g} and \eqref{eq:sc_nc1h} respectively, and choosing $\delta_1$ such that

\begin{IEEEeqnarray}{LLL}
\label{eq:sc_nc29}
\delta_1 >\delta_1^*
\end{IEEEeqnarray}
with $\delta_1^*$ given by \eqref{eq:sc_nc1i}, will ensure that the first term in \eqref{eq:sc_nc28} goes to 0 as $n\to \infty$.

Note that
 \begin{IEEEeqnarray}{LLL}
\label{eq:sc_nc30}
p_3 &=&\Pb{E_2^c}\\
&=& \Prb{\bigcup_{t,S,K_1}\biggr\{V_{n,t}W_S-1\leq \delta_1(1-V_{n,t})\biggr\}}.
 \end{IEEEeqnarray}
 Conditional on $H_{[K]}$, 
 $$\norma{\PPA\hcc}^2\distas{}\frac{1}{2}P'\sum_{i\in S_2^*}|H_i|^2\chi_2^{S_2^*}(2n')$$
  where $\chi_2(2n')$ is a chi-squared distributed random variable with $2n'$ degrees of freedom (here the superscript $S_2^*$ denotes the fact that this random variable depends on the codewords corresponding to $S_2^*$). For $1>\delta_2>0$, consider the event $E_4=\bigcap_{t,S,K_1} \left\{\frac{\chi_2^{S_2^*}(2n')}{2n'}> 1-\delta_2\right\}$. Using \eqref{eq:sc_nc25c} , we can bound $p_3$ as 
 
 \begin{IEEEeqnarray*}{LLL}
 \label{eq:sc_nc31}
 p_3\leq \sum_t\binom{K}{K-K_1+t} e^{-n'\left(-\ln (1-\delta_2)-\delta_2\right)}+p_4 \IEEEyesnumber
 \end{IEEEeqnarray*}
 where 
 \begin{IEEEeqnarray*}{LLL}
 \label{eq:sc_nc32}
 p_4 &=& \Pb{E_4^c}\\
&=& \Prb{\bigcup_{t,S,K_1}\biggr\{V_{n,t}\biggr(1+\frac{P'\sum_{i\in S_2^*}|H_i|^2(1-\delta_2)}{\phs}\biggr)\\
&& \leq  1+\delta_1(1-V_{n,t})\biggr\}}.\IEEEyesnumber
\end{IEEEeqnarray*}

 Again, it is enough to choose $\delta_2$ such that
 \begin{equation}
 \label{eq:sc_nc33}
 \delta_2>\delta_2^*
 \end{equation}
with $\delta_2^*$ given by \eqref{eq:sc_nc1j}, to make sure that the first term in \eqref{eq:sc_nc31} goes to $0$ as $n\to\infty$.

Note that the union bound over $S$ is the minimum over $S$, and this minimizing $S$ should be contiguous amongst the indices arranged according the decreasing order of fading powers. Further, $S_2^*$ is chosen to be corresponding to the top $t$ fading powers in $S$. Hence, we get

\begin{\Ieee}{LLL}
\label{eq:sc_nc34}
p_4 &=& \Prb{\bigcup_t\biggr\{\min_{0\leq j\leq K_1-t}\biggr(\frac{P'\sum_{i=j+1}^{j+t}|H_{(i)}|^2(1-\delta_2)}{1+P'\sum_{i=j+t+1}^{j+t+K-K_1}|H_{(i)}|^2}\biggr)\\
&& \leq \frac{ 1+\delta_1(1-V_{n,t})}{V_{n,t}}-1 \biggr\}}.\Ieeen
\end{\Ieee}

We make the following claim
\begin{claim}
\label{claim:sc_nc5}

\begin{IEEEeqnarray*}{LLL}
\label{eq:sc_nc35}
\limsup_{n\to\infty} p_4 \leq 1\biggr[\bigcup_{\theta\in (\frac{\epsilon'}{\nu},1]\cap \mathbb{Q}}\\
\biggr\{ \inf_{\xi\in [0,\nu(1-\theta)]} \biggr(\frac{(1-\delta_2)\PT'\alpha(\xi,\xi+\nu\theta)}{1+\PT'\alpha(\xi+\nu\theta,\xi+1-\nu(1-\theta))}\biggr)\\
 \leq \frac{1+\delta_1(1-V_{\theta})}{V_\theta}-1\biggr\}\biggr]\IEEEyesnumber
\end{IEEEeqnarray*}
where $\alpha(a,b)$ is given by \eqref{eq:sc_nc1k}.

\begin{proof}

We have $|H_1|^2,...,|H_{K}|^2$ with CDF $F(x)=(1-e^{-x})1[x>=0]$. Let $\tilde{F}_K (x)=\frac{1}{K}\sum_{i=1}^{K}1[|H_i|^2\leq x]$ be the empirical CDF (ECDF). Then standard Chernoff bound gives, for $0<r<1$,
\begin{\Ieee}{LLL}
\label{eq:sc_nc36}
\Pb{|\tilde{F}_K(x)-F(x)|> r F(x)}\leq 2e^{-KcF(x)r^2}\Ieeen
\end{\Ieee}
where $c$ is some constant. 

From \citep{bahadur1966note}, we have the following representation. Let $0<\gamma<1$. Then
\begin{\Ieee}{LLL}
\label{eq:sc_nc37}
|H_{\left(\lceil n\gamma \rceil\right)}|^2 =\\
 F^{-1}(1-\gamma) - \frac{\tilde{F}_K(F^{-1}(1-\gamma))-(1-\gamma)}{f\left(F^{-1}(1-\gamma)\right)}+R_K\Ieeen
\end{\Ieee}
where $f$ is the pdf corresponding to F, and with probability 1, we have $R_K=O(n^{-3/4}\log(n))$ as $n\to \infty$. 

Let $\tau>0$. Then using \eqref{eq:sc_nc36} and \eqref{eq:sc_nc37}, we have 
\begin{\Ieee}{LLL}
\label{eq:sc_nc38}
\left||H_{\left(\lceil n\gamma \rceil\right)}|^2 -F^{-1}(1-\gamma)\right| \leq O\left(\frac{1}{n^{\frac{1-\tau}{2}}}\right)\Ieeen
\end{\Ieee}
with probability atleast $1-e^{-O(n^\tau)}$.

Hence, for $0<\xi<\zeta<1$, we have, with probability $1-e^{-O(n^\tau)}$, 

\begin{\Ieee}{LLL}
\label{eq:sc_nc39}
\frac{1}{K}\sum_{i=\lceil\alpha K\rceil}^{\lceil\beta K\rceil}|H_{\left(i\right)}|^2= \\
\left[\frac{1}{K}\sum_{i=1}^{K}|H_i|^2 1\left[b\leq |H_i|^2\leq a\right]\right] +o(1)\Ieeen
\end{\Ieee}
where $a=F^{-1}(1-\xi)$ and $b=F^{-1}(1-\zeta)$. Now, by law of large numbers (and Bernstein's inequality \citep{boucheron2013concentration}), with overwhelming probability (exponentially close to 1), we have 

\begin{\Ieee}{LLL}
\label{eq:sc_nc40}
\frac{1}{K}\sum_{i=1}^{K}|H_i|^2 1\left[b\leq |H_i|^2\leq a\right] = \int_{b}^a x dF(x)+ o(1)\Ieeen
\end{\Ieee}

and $\int_{b}^a x dF(x)=\int_{\xi}^{\zeta}F^{-1}(1-\gamma)d\gamma=\alpha (\xi,\zeta)$.

Define the events
\begin{\Ieee}{RLL}
J_{n,\theta,\xi} &=&\biggr\{\biggr(\frac{P'\sum_{i=\lceil\xi K\rceil +1}^{\lceil (\xi+\nu\theta)K\rceil } |H_{(i)}|^2(1-\delta_2)}{1+P'\sum_{i=\lceil (\xi+\nu\theta)K\rceil+1}^{\lceil (\xi+1-\nu(1-\theta))K\rceil}|H_{(i)}|^2}\biggr)\\
&&\leq \frac{ 1+\delta_1(1-V_{n,\lceil\theta\nu K\rceil})}{V_{n,\lceil\theta\nu K\rceil}}-1\biggr\} \label{eq:sc_nc41a}\Ieeen\\
I_{n,\theta,\xi} &= &\biggr\{\biggr(\frac{(1-\delta_2)\PT'\alpha(\xi,\xi+\nu\theta)}{1+\PT'\alpha(\xi+\nu\theta,\xi+1-\nu(1-\theta))}\biggr) \\
&&\leq \frac{1+\delta_1(1-V_{n,\lceil\theta\nu K\rceil})}{V_{n,\lceil\theta\nu K\rceil}}-1\biggr\}\label{eq:sc_nc41b}\Ieeen\\
I_{\theta,\xi} &= &\biggr\{\biggr(\frac{(1-\delta_2)\PT'\alpha(\xi,\xi+\nu\theta)}{1+\PT'\alpha(\xi+\nu\theta,\xi+1-\nu(1-\theta))}\biggr) \\
&&\leq \frac{1+\delta_1(1-V_{\theta})}{V_\theta}-1\biggr\}\label{eq:sc_nc41c}\Ieeen\\
E_{n,\theta,\xi}&= &\biggr\{\biggr| \frac{1}{K}\sum_{i=\lceil\xi K\rceil +1}^{\lceil (\xi+\nu\theta)K\rceil }|H_{(i)}|^2-\alpha\left(\xi,\xi+\nu\theta\right)\leq o(1)\biggr|\biggr\}\\
&& \bigcap\biggr\{\biggr| \frac{1}{K}\sum_{i=\lceil (\xi+\nu\theta)K\rceil+1}^{\lceil (\xi+1-\nu(1-\theta))K\rceil}|H_{(i)}|^2-\\
&&\alpha\left(\xi+\nu\theta,\xi+1-\nu(1-\theta)\right)\leq o(1)\biggr|\biggr\}\label{eq:sc_nc41d}\Ieeen\\
E_n &= &\left(\bigcap_{\theta\in A_n}\bigcap_{\xi \in B_{K,\theta}} E_{n,\theta,\xi} \right)\label{eq:sc_nc41e}\Ieeen
\end{\Ieee}

where $A_n=\left(\frac{\epsilon'}{\nu},1\right]\cap \left\{\frac{i}{K_1}: i\in [K_1]\right\}$ and $B_{K,\theta}=[0,\nu(1-\theta)]\cap \left\{\frac{i}{K}:i\in[K]\right\}$. Note that, from \eqref{eq:sc_nc39} and \eqref{eq:sc_nc40}, $\Pb{E^c_{n,\theta,\xi}}$ is exponentially small in $n$.

Then we have

\begin{\Ieee}{LLL}
\label{eq:sc_nc42}
p_4 &=& \Pb{\bigcup_{\theta\in A_n}\bigcup_{\xi \in B_{K,\theta}}J_{n,\theta,\xi}}\\
&\leq &\Prb{\bigcup_{\theta\in A_n} \bigcup_{\xi \in B_{K,\theta}} J_{n,\theta,\xi} \cap E_{n,\theta,\xi}}\\
&& + \sum_{\theta\in A_n}\sum_{\xi\in B_{K,\theta}}\Pb{E^c_{n,\theta,\xi}}\\
&\leq & 1\biggr[\bigcup_{\theta\in A_n}\bigcup_{\xi\in B_{K,\theta}}I_{n,\theta,xi}\biggr]+ o(1)\\
&\leq & 1\biggr[\bigcup_{\theta \in (\frac{\epsilon'}{\nu},1] } \bigcup_{\xi\in [0,\nu(1-\theta)] }I_{n,\theta,\xi}\biggr]+o(1).\Ieeen
\end{\Ieee}

Therefore
\begin{\Ieee}{LLL}
\limsup_{n\to\infty} p_4\leq 1\biggr[\bigcup_{\theta \in (\frac{\epsilon'}{\nu},1] } \bigcup_{\xi\in [0,\nu(1-\theta)] }I_{\theta,\xi}\biggr] \Ieeen
\end{\Ieee}

This concludes the proof of claim \ref{claim:sc_nc5}.

\end{proof}
\end{claim}

The statement of the theorem follows by choosing $\PT'$ to make sure that $\limsup_{n\to\infty}p_4=0$.

\begin{remark}
In retrospect, our analysis is rather similar to the one in \citep{reeves2012sampling}. We remind that the problem considered there can be seen (as argued in \citep{polyanskiy2017perspective}) as a version of the many-MAC problem with random-access, cf. Section~\ref{sec:cs} for more.
\end{remark}

\end{proof}

\end{theorem}

\subsection{CSIR}

In this subsection, we focus on the CSIR scenario. We could use projection decoding to decode a fraction of users where decoding set is a function of CSIR. But a better bound is obtained by directly using euclidean metric to decode, similar to \citep{polyanskiy2017perspective}. Then have the following theorem.

\begin{theorem}
\label{th:scaling_CSIR}
Consider the channel \eqref{eq:sys1} (with CSIR) with $K=\mu n$ where $\mu<1$. Fix the spectral efficiency $S$ and target probability of error (per-user) $\epsilon$. Let $M=2^{S/\mu}$ denote the size of the codebooks and $\PT=KP$ be the total power. Fix $\nu \in (1-\epsilon,1]$. Let $\epsilon'=\epsilon-(1-\nu)$. Then if $\mathcal{E}>\mathcal{E}^*_{CSIR}=\sup_{\frac{\epsilon'}{\nu}< \theta\leq 1}\inf_{0\leq \rho\leq 1}\frac{\PTX(\theta,\rho)}{S}$, there exists a sequence of $\left(n,M,\epsilon_n,\mathcal{E},K=\mu n\right)$ codes such that $\limsup_{n\to\infty} \epsilon_n\leq \epsilon$, where, for $\frac{\epsilon'}{\nu}< \theta\leq 1$, 

{\allowdisplaybreaks
\begin{IEEEeqnarray*}{LLL}
\PTX(\theta,\rho) \\
= \frac{(1+\rho)\left(e^{\mu\nu\left(\frac{ h(\theta)}{\rho}+\theta \ln M\right)}-1\right)}{\alpha(\nu(1-\theta),\nu)-\left(e^{\mu\nu\left(\frac{ h(\theta)}{\rho}+\theta \ln M\right)}-1\right)\alpha(\nu,1)(1+\rho)}\label{eq:sc_csir1a}\\
\Ieeen\\
\alpha(a,b) = a\ln(a)-b\ln(b)+b-a.\label{eq:sc_csir1b}\Ieeen
\end{IEEEeqnarray*}
Hence $\mathcal{E}^*\leq \mathcal{E}^*_{CSIR}$.
}

The proof idea is a combination of techniques similar to \citep{polyanskiy2017perspective} and theorem \ref{th:scaling_noCSI}

\begin{proof}

Let each user generate a Gaussian codebook of size $M$ and power $P'<P$ independently such that $KP'=\PT'<\PT$. Let $W_j$ denote the random (in $[M]$) message of user $j$. So, if $\cd_j=\{c^j_i:i\in[M]\}$ is the codebook of user $j$, he/she transmits $X_j=c^j_{W_j}1\left\{\norm{c^j_{W_j}}^2\leq nP\right\}$. For simplicity let $(c_1,c_2,...,c_K)$ be the sent codewords. Fix $\nu \in (1-\epsilon,1]$. Let $K_1=\nu K$ be the number of users that are decoded. Fix a decoding set $D\subset[K]$, possibly depending on $H_{[K]}$ such that $|D|=K_1, a.s$. Since the receives knows $H_{[K]}$, we can use the euclidean distance used in \citep{polyanskiy2017perspective} as the decoding metric. Formally, the decoder output $g_{D}(Y)\in\prod_{i=1}^{K} \mathcal{C}_i$ is given by

\begin{\Ieee}{LLL}
\label{eq:sc_csir2}
(g_D (Y))_i= \begin{cases} 
      f_i^{-1}(\hat{c}_i) & i \in D \\
      ? & i\notin D 
   \end{cases}\IEEEnonumber\\
(\hat{c}_i)_{i\in D}=\arg\min_{(c_i\in\cd_i)_{i\in D}}\norm{Y-\sum_{i\in D}H_i c_i}^2.
\end{\Ieee}

The probability of error is given by
\begin{IEEEeqnarray}{LLL}
\label{eq:sc_csir3}
P_e=\frac{1}{K} \sum_{j=1}^{K}\Pb{W_j\neq \hat{W}_j}
\end{IEEEeqnarray}
where $\hat{W}_j=(g(Y))_j$ is the decoded message of user $j$.
 Similar to the no-CSI case, we perform a change of measure to $X_j=c^j_{W_j}$ by adding a total variation distance bounded by $p_0=K\Pb{\frac{\chi_2(2n)}{2n}>\frac{P}{P'}}\to 0$ as $n\to \infty$.
  
Let $\epsilon'=\epsilon-(1-\nu)$. Now we have
{\allowdisplaybreaks
\begin{IEEEeqnarray*}{LLL}
\label{eq:sc_csir4a}
P_e &=&\Ex{\frac{1}{K} \sum_{j=1}^{K}1\{W_j\neq \hat{W}_j\}} \\
&=&\frac{K-K_1}{K}+\Ex{\frac{1}{K} \sum_{j\in D}1\{W_j\neq \hat{W}_j\}}\\
& \leq & (1-\nu) + \epsilon'+\nu\Pb{\frac{1}{K} \sum_{j\in D}1\{W_j\neq \hat{W}_j\} \geq \epsilon'}\\
%& =& \epsilon +\nu\Pb{\bigcup_{t=\epsilon' K}^{\nu K}\left\{\sum_{j\in D}1\{W_j\neq \hat{W}_j\} =t\right\}}\\
&=& \epsilon + \nu p_1\IEEEyesnumber
\end{IEEEeqnarray*} 
}
where $p_1=\Pb{\bigcup_{t=\epsilon' K}^{\nu K}\left\{\sum_{j\in D}1\{W_j\neq \hat{W}_j\} =t\right\}}$.

From now on, we just write $\bigcup_{t}$ to denote $\bigcup_{t=\epsilon'K}^{\nu K}$, $\sum_t$ for $\sum_{t=\epsilon'K}^{\nu K}$, and $\sum_S$ for $\sum_{\substack{S\subset D\\ |S|=t}}$. Let $c_{[S]}\equiv \{c_i:i\in[S]\}$ and $H_{[K]}=\{H_i:i\in[K]\}$.

Let $F_t=\left\{\sum_{j\in D}1\{W_j\neq \hat{W}_j\} =t\right\}$. Let $\rho\in[0,1]$. We bound $\Pb{F_t}$ using Gallager's rho trick similar to \citep{polyanskiy2017perspective} as

\begin{\Ieee}{LLL}
\label{eq:sc_csir4b}
&&\Pb{F_t| Z,c_{[K]},H_{[K]}}\\
&\leq & \Prb{\exists S\subset D: |S|=t, \exists \{c'_i\in \cd_i:i\in S, c_i'\neq c_i\}:\\
&&  \norma{Y-\sum_{i\in S}H_i c'_i-\sum_{i\in D\setminus S }H_ic_i}^2<\\
&&\norma{ Y-\sum_{i\in D }H_ic_i}^2 \biggr| Z,c_{[K]},H_{[K]}}\\
& \leq & \sum_{S}\Prb{\bigcup_{\substack{c'_i\in \cd_i:i\in S\\ c_i'\neq c_i}} \biggr\{\norma{Z_D+\hc -\sum_{i\in S}H_i c'_i}^2<\\
&&\norma{Z_D}^2\biggr\} \biggr| Z,c_{[K]},H_{[K]}}\\
& \leq & \sum_S M^{\rho t} \Prb{\norma{Z_D+\hc -\sum_{i\in S}H_i c'_i}^2<\\
&&\norma{Z_D}^2 \biggr| Z,c_{[K]},H_{[K]}}^\rho\Ieeen
\end{\Ieee}
where $Z_D=Z+\sum_{i\in [K]\setminus D}H_i c_i$ and $c'_{[S]}$ in the last display denotes a generic set of unsent codewords corresponding to codebooks of users in set $S$.

We use the following simple lemma which is a trivial extension of a similar result used in \citep{polyanskiy2017perspective} to compute the above probability.

\begin{lemma}
\label{lem:sc_csir1}
Let $Z\distas{}\cn(0,I_n)$ and $u\in \mathbb{C}^n$. Let $D=diag(d_1,...,d_n)\in \mathbb{C}^{n\times n}$ be a diagonal matrix. If $\gamma>\sup_{j\in[n]}-\frac{1}{|d_j|^2}$, then
\begin{\Ieee}{LLL}
\label{eq:sc_csir5}
\Ex{e^{-\gamma\norm{DZ+u}^2}}=\frac{1}{\prod_{j\in[n]}\left(1+\gamma |d_j|^2\right)}e^{-\gamma\sum_{j\in[n]}\frac{|u_j|^2}{1+\gamma |d_j|^2}}
\end{\Ieee}

\begin{proof}
Omitted.

\end{proof}
\end{lemma}

So, using the above lemma, we have, for $\lambda_1>0$, 

\begin{\Ieee}{LLL}
\label{eq:sc_csir6}
&& \Exxa{\{c'_{S}\}}{\Prb{ \norma{Z_D+\hc -\sum_{i\in S}H_i c'_i}^2<\\
&& \norm{Z_D}^2\biggr| Z,c_{[K]},H_{[K]}}^\rho}\\
&=&\Exxa{\{c'_{S}\}}{\Prb{ \Exp{-\lambda_1\norma{Z_D+\hc -\sum_{i\in S}H_i c'_i}^2}>\\
&&\Exp{-\lambda_1\norm{Z_D}^2}\biggr| Z,c_{[K]},H_{[K]}}^\rho}\\
& \leq & \frac{e^{\rho\lambda_1 \norm{Z_D}^2}}{\left(1+\lambda_1 \phss\right)^{\rho n}}e^{\frac{-\rho\lambda_1\norm{Z_D+\hc}^2}{1+\lambda_1 \phss}}\Ieeen
\end{\Ieee}
where $\mathbb{E}_{c'_{S}}$ denotes taking expectation with respect to $\{c'_i:i\in S\}$ alone, and $1+\lambda_1 \phss>0$. 

Let $\lambda_2=\frac{\rho \lambda_1}{1+\lambda_1 \phss}$. Note that $\lambda_2$ is a function of $H_{S}$. Now using lemma \ref{lem:sc_csir1} again to take expectation over $c_{S}$, we get

{\allowdisplaybreaks
\begin{IEEEeqnarray*}{LLL}
\label{eq:sc_csir7}
&&\Exxa{c_{S}}{\frac{e^{\rho\lambda_1 \norm{Z_D}^2}}{\left(1+\lambda_1 \phss\right)^{\rho n}}e^{\frac{-\rho\lambda_1\norm{Z_D+\hc}^2}{1+\lambda_1 \phss}}} \\
&\leq &\frac{1}{\left(1+\lambda_1 \phss\right)^{\rho n}}\frac{1}{\left(1+\lambda_2 \phss\right)^{ n}}\cdot\\
&&e^{\left(\rho\lambda_1-\frac{\lambda_2}{1+\lambda_2\phss}\right)\norm{Z_D}^2}\Ieeen
\end{IEEEeqnarray*}
}
with $1+\lambda_2\phss>0$. Finally, taking expectation over $Z$, we get

\begin{\Ieee}{LLL}
\label{eq:sc_csir8}
\Pb{F_t|H_{[K]}}\leq \sum_S M^{\rho t}e^{-nE_0(\lambda_1;\rho,H_{[K]},S)}\Ieeen
\end{\Ieee}
 
where  
\begin{\Ieee}{LLL}
\label{eq:sc_csir9}
&&E_0(\lambda_1;\rho,H_{[K]},S)\\
&=&\rho \ln\biggr(1+\lambda_1\phss\biggr)+\ln\biggr(1+\lambda_2\phss\biggr)+\\
&& \ln\biggr(1-\phdd\cdot\\
&&\biggr(\rho\lambda_1-\frac{\lambda_2}{1+\lambda_2\phss}\biggr)\biggr)\Ieeen
\end{\Ieee}
with 
\begin{\Ieee}{LLL}
1>\phdd\left(\rho\lambda_1-\frac{\lambda_2}{1+\lambda_2\phss}\right).
\end{\Ieee}

It is easy to see that the optimum value of $\lambda_1$ that maximizes $E_0$ is given by

\begin{\Ieee}{LLL}
\label{eq:sc_csir10}
\lambda_1^*=\frac{1}{\phdd (1+\rho)}\Ieeen
\end{\Ieee}
and hence the maximum value of the exponent 
$$E_0(\rho,H_{[K]},S)=E_0(\lambda_1^*;\rho,H_{[K]},S)$$
 is given by
\begin{\Ieee}{LLL}
\label{eq:sc_csir11}
E_0(\rho, H_{[K]},S)\\
=\rho \ln\left(1+\frac{\phss}{(1+\rho)\phdd}\right).
\end{\Ieee}

Therefore, we have

\begin{\Ieee}{LLL}
\label{eq:sc_csir12}
p_1\leq \Ex{\sum_t\sum_S e^{\rho t \ln M} e^{-n E_0(\rho,H_{[K]},S)}}.\Ieeen
\end{\Ieee}

Since we want an upper bound for \eqref{eq:sc_csir12}, we would like to take minimum over $S\subset D: |S|=t$. For a given choice of $D$, this corresponds to minimizing $\phss$ which mean we take $S$ to contain indices in $D$ which correspond to $t$ smallest fading coefficients (within $D$). Then, the best such bound is obtained by choosing $D$ that maximizes $\frac{\phss}{\phdd}$. Clearly this corresponds to choosing $D$ to contain indices corresponding to top $K_1$ fading coefficients.

Therefore, we get
\begin{\Ieee}{LLL}
\label{eq:sc_csir13}
p_1\leq\\
 \Exa{\sum_t \binom{K_1}{t}e^{\rho t \ln M} e^{-n\rho\ln\left(1+\frac{P'\sum_{i=K_1-t+1}^{K_1}|H_{(i)}|^2}{(1+\rho)\left(1+P'\sum_{i=K_1+1}^{K}|H_{(i)}|^2\right)}\right)}}.
\end{\Ieee}

Let $A_n=[\frac{\epsilon'}{\nu},1]\cap \left\{\frac{i}{K_1}:i\in[K_1]\right\}$. For $\theta\in A_n$ and $t=\theta K_1$, using \citep[Ex. 5.8]{Gallager:1968:ITR:578869} again, we have

\begin{IEEEeqnarray*}{LLL}
\label{eq:sc_csir14}
\binom{K_1}{t}\leq  \sqrt{\frac{K_1}{2\pi t (K_1-t)}}e^{K_1 h(\frac{t}{K_1})} = O\left(\frac{1}{\sqrt{n}}\right)e^{n\mu\nu h(\theta)}.\\
\IEEEyesnumber
\end{IEEEeqnarray*}

The choice of $\rho$ was arbitrary, and hence,

\begin{\Ieee}{LLL}
\label{eq:sc_csir15}
p_1 &\leq &  \Exa{\min\biggr\{1,\sum_{\theta\in A_n} \exp\biggr(-n\sup_{\rho\in[0,1]}\biggr( \\
&&\rho\ln\biggr(1+\frac{P'\sum_{i=\nu(1-\theta)K+1}^{\nu K}|H_{(i)}|^2}{(1+\rho)(1+P'\sum_{i=\nu K+1}^{K}|H_{(i)}|^2)}\biggr) \\
&&-\mu\nu h(\theta)-\mu\nu\theta\ln M \biggr)\biggr)\biggr\}}\\
&\leq &  \Exa{\min\biggr\{1,|A_n| \exp\biggr(-n\inf_{\theta\in A_n}\sup_{\rho\in[0,1]}\\
&&\biggr( \rho\ln\biggr(1+\frac{P'\sum_{i=\nu(1-\theta)K+1}^{\nu K}|H_{(i)}|^2}{(1+\rho)(1+P'\sum_{i=\nu K+1}^{K}|H_{(i)}|^2)}\biggr) \\
&&-\mu\nu h(\theta)-\mu\nu\theta\ln M \biggr)\biggr)\biggr\}}\Ieeen
\end{\Ieee}

where we have used min since $p_1\leq 1$.
Now, using similar arguments as in the proof of claim \ref{claim:sc_nc5} and taking limits, we can see that

\begin{\Ieee}{LLL}
\label{eq:sc_csir16}
\inf_{\theta\in A_n}\sup_{\rho\in[0,1]} \biggr( \rho\ln\biggr(1+\frac{P'\sum_{i=\nu(1-\theta)K+1}^{\nu K}|H_{(i)}|^2}{(1+\rho)(1+P'\sum_{i=\nu K+1}^{K}|H_{(i)}|^2)}\biggr)\\
 -\mu\nu h(\theta)-\mu\nu\theta\ln M \biggr) =\\
 \inf_{\theta\in A_n}\sup_{\rho\in[0,1]}\biggr( \rho\ln\biggr(1+\frac{\PT' \alpha(\nu(1-\theta),\nu)}{(1+\rho)(1+\PT'\alpha(\nu,1))}\biggr)\\
 -\mu\nu h(\theta)-\mu\nu\theta\ln M \biggr) + o(1)\Ieeen
\end{\Ieee}
with exponentially high probability. Hence, 

\begin{\Ieee}{LLL}
\label{eq:sc_csir17}
p_1 &\leq &\Exa{|A_n|\exp\biggr(o(n)-n\inf_{\theta\in A_n}\sup_{\rho\in[0,1]}\\
&&\biggr( \rho\ln\biggr(1+\frac{\PT' \alpha(\nu(1-\theta),\nu)}{(1+\rho)(1+\PT'\alpha(\nu,1))}\biggr) -\mu\nu h(\theta)\\
&&-\mu\nu\theta\ln M \biggr) \biggr)} +o(1)\\
&\leq &  \Exa{|A_n|\exp\biggr(o(n)-n\inf_{\theta\in A}\sup_{\rho\in[0,1]}\\
&&\biggr( \rho\ln\biggr(1+\frac{\PT' \alpha(\nu(1-\theta),\nu)}{(1+\rho)(1+\PT'\alpha(\nu,1))}\biggr) -\mu\nu h(\theta)\\
&&-\mu\nu\theta\ln M \biggr) \biggr)} +o(1) \Ieeen
\end{\Ieee}
where $A=[\frac{\epsilon'}{\nu},1]$.

Therefore, choosing $\PT'>\sup_{\theta\in A}\inf_{\rho\in[0,1]}\PT(\theta,\rho)$ will ensure that $\limsup_{n\to\infty} p_1 =0$.

\end{proof}

\end{theorem}

%{\color{blue}{
\begin{remark} 
Note that the analysis of the CSIR case in this paper and the AWGN case in~\cite{polyanskiy2017perspective} are similar, in particular both analyze a (suboptimal for PUPE) maximum likelihood decoder. However, there are two new subtleties, compared to~\cite{polyanskiy2017perspective}.  First,~\cite{polyanskiy2017perspective} applies Gallager's  $\rho$-trick twice, where the second application (with parameter $\rho_1$ in the notation of~\cite{polyanskiy2017perspective}) is applied just before taking the expectation over $Z$ in~\eqref{eq:sc_csir8}. In the CSIR case, the summands of $\sum_S$ actually depend on the subset $S$ through the fading gains, which makes the $\rho$-trick less appealing, and that is why we omitted it here. Secondly, because the summands depend on $S$, we upper bound each by taking the maximum over $S$, and this requires analysis of order statistics which is, of course, not present in the AWGN case.
\end{remark}
%}}

%\txcol{blue}{
\subsection{Achievability bound via scalar AMP}
\label{sec:amp_ach}
In this section, we will given an achievability bound on $E_b/N_0$ for the no-CSI case by the asymptotic analysis of the
scalar AMP algorithm \citep{donoho2009message,bayati2011dynamics,barbier2017approximate,reeves2012sampling}. Here, we
recall the compressed sensing view of our model \eqref{eq:cs1} where $U$ is block sparse. As discussed in section
\ref{sec:cs}, a better algorithm to use in this case would be the vector or block version of AMP, whose analysis is also well studied, e.g.~\citep{barbier2017approximate}. However, as we discussed in Section~\ref{sec:cs} evaluation of performance of this block-AMP requires computing $M=2^{100}$ dimensional integrals, and thus does not result in computable bounds. 
Instead, here we take a different approach by analyzing the scalar AMP algorithm, whose asymptotic analysis
in~\cite{bayati2011dynamics} in fact only requires that the empirical distribution of entries of $U$ be convergent -- a
fact emphasized in~\citep{reeves2012sampling}. Let us restate the signal model we have:
\begin{equation}\label{eq:amp1}
	Y=AU+Z, \qquad A_{i,j} \stackrel{iid}{\sim}\cn(0,\PT/\mu), i\in[n], j\in[KM]\,,
\end{equation}
where $U\in\mathbb{C}^{KM}$ is block sparse with
$K=\mu n$ blocks each of length $M$, with a single non-zero entry $U_j$ in each block with $U_j\sim \cn(0,1)$ (Rayleigh
fading), and $Z\distas{}\cn(0,I_n)$.
The support of $U$, denoted by $S\in \{0,1\}^{KM}$, is sampled uniformly from all such block sparse supports (there are
$M^K$ of them). The goal is to get an estimate $\hat S=\hat S(Y,A)$ of $S$ where our figure of merit is the following:
\begin{equation}\label{eq:amp_pupe}
	\text{PUPE}(\hat S) = \frac{1}{K} \sum_{k=1}^K \Pb{S_{1+(k-1)M}^{kM} \neq \hat S_{1+(k-1)M}^{kM}}\,,
\end{equation}
which is also known as section error rate (SER) in the SPARC literature~\cite{CIT-092}.

The AMP-based algorithm operates as follows. First we estimate $U$ iteratively, then after estimating $U$ we threshold
its values to obtain an estimator for $S$. 

To describe scalar AMP we first introduce the following scalar problem. For each $\sigma>0$ define
$\mu^{(\sigma)} = P_{X,V}$ to be the joint distribution of variables $X$ and $V$:
\begin{equation}\label{eq:amp_scalar}
	V = X + \sigma W, \qquad X \perp\!\! W \sim \mathcal{CN}(0,1)
\end{equation}
and
\begin{equation}
\label{eq:ber_gauss}
X\distas{}\mathrm{BG}(1,1/M)=\begin{cases}
      \cn(0,1) & \mathrm{w.p.}\, \frac{1}{M}\\
      0 & \mathrm{w.p.}\, 1-\frac{1}{M}
\end{cases} 
\end{equation}
We also define
\begin{equation}
\label{eq:amp_scalar_mmse}
 \eta(z,\sigma^2) \eqdef \EE[X | V = z]\,, \mmse(\sigma^2) \eqdef \EE[(X-\EE[X|V])^2]\,.
 \end{equation}

Next, start with $U^{(0)}=0\in\mathbb{C}^{KM}$, $R^{(0)}=Y$, $\hat\sigma^2_0=\frac{\mu}{\PT}+\mu$. Then for $t=1,2,\cdots $ we have the following iterations
\begin{\Ieee}{LLL}
U^{(t)}= \eta\left(A^* R^{(t-1)}+U^{(t-1)},\hat\sigma^2_{t-1}\right)\label{eq:amp_iter1}\Ieeen\\
R^{(t)} = Y-AU^{(t)}+\\
\mu M R^{(t-1)}\frac{1}{KM}\sum_{i=1}^{KM}\eta'\left(\left(A^*R^{(t-1)}+U^{(t-1)}\right)_i,\hat\sigma_{t-1}^2\right)\label{eq:amp_iter2}\\
\Ieeen\\
\hat \sigma_t^2 =\frac{1}{n}\norm{R^{(t)}}^2\label{eq:amp_iter3}\Ieeen
\end{\Ieee}
where $\eta'(x+iy,\sigma^2)$ denotes $\frac{1}{2}\left(\frac{\partial \eta(x+iy,\sigma^2)}{\partial x}-i\frac{\partial \eta(x+iy,\sigma^2)}{\partial y}\right)$ and $i=\sqrt{-1}$ is the imaginary unit (see \citep{meng2015concise,zou2018concise} for a more general derivation of complex AMP). The estimate of $U$ after $t$ steps is 
given by (see \citep{reeves2012sampling} for more details) 
\begin{equation}
\label{eq:amp_iter4}
\hat U^{(t)}=A^* R^{(t)}+U^{(t)}
\end{equation}
To convert $\hat U^{(t)}$ into $\hat S^{(t)}$ we perform a simple thresholding:
\begin{equation}\label{eq:amp_decoder}
	\hat S^{(t)}(\theta) = \{i\in[KM]:|\hat U_i^{(t)}|^2>\theta\}\,.
\end{equation}

\begin{theorem}[Scalar AMP achievability]
\label{th:amp_ach} Fix any $\mu>0$, $P_{tot}>0$ and $M\ge 1$. Then for every $\mathcal{E} > \frac{P_{tot}}{\mu \log_2 M}$
there exist a sequence of
$\left(n,M,\epsilon_n,\mathcal{E},K=\mu n\right)$ codes (noCSI) such that AMP decoder~\eqref{eq:amp_decoder} (with
a carefully chosen $\theta = \theta(\mathcal{E},M,\mu)$ and sufficiently large $t$) achieves
$$ \limsup_{n\to\infty} \epsilon_n \le \pi^*(\sigma^2_{\infty}, M)\,,$$
where $\pi^*(\tau, M) = 1-\frac{1}{1+\tau} \left((M-1)\left(\frac{1}{\tau}+1\right)\right)^{-\tau}$ 
and $\sigma^2_\infty$ is found from 
 \begin{\Ieee}{LLL}
 \label{eq:amp_fp}
 \sigma_{\infty}^2 &\equiv & \sigma_{\infty}^2(\mu,\PT,M)\\
 &=&\sup\left\{\tau\geq 0:\tau=\frac{\mu}{\PT}+\mu M\mmse(\tau) \right\}\Ieeen
 \end{\Ieee}
\end{theorem}

\begin{proof}
Denote the Hamming distance 
 \begin{equation}
 \label{eq:amp_ham}
 d_{H}(S,\hat S)=\frac{1}{KM}\sum_{i=1}^{KM}1[S_i\neq \hat S_i]
 \end{equation}
 Note that according to the definition~\eqref{eq:amp_pupe} we have a bound

 \begin{equation}
 \label{eq:amp_pupe_ham}
 \text{PUPE}(\hat S^{(t)}(\theta))\leq M\Ex{d_H(S,\hat S^{(t)}(\theta)))}
 \end{equation}
 Indeed, this is a simple consequence of upper bounding each probability in~\eqref{eq:amp_pupe} by the union bound. 
 
 The key result of~\cite{bayati2011dynamics} shows the following. Let the empirical joint
distribution of entries in $(U,\hat{U}^{(t)})$ be denoted by
	$$\hat{\mu}_{U,\hat U^{(t)}} \eqdef \frac{1}{KM}\sum_{i=1}^{KM}\delta_{(U_i, \hat U_i^{(t)})}\,,$$
where $\delta_{x}$ is the Dirac measure at $x$. Then as $n\to\infty$ this (random) distribution on $\mathbb{C}^2$
converges weakly to a deterministic limit $\mu_{X,V}$ almost surely. 
More precisely, from~\cite[Lemma 1(b)]{bayati2011dynamics} and proof of~\cite[Theorem 5]{reeves2012sampling} for any bounded Lipschitz continuous function $f:\mathbb{C}^2 \to \mathbb{R}$ we have  

\begin{equation}
\label{eq:amp_a.s_convergence}
\lim_{n\to\infty}\int f\,d\hat{\mu}_{U,\hat U^{(t)}}=\int f\,d\mu^{(\sigma_t)}\, \mathrm{a.s.}
\end{equation}
%$$\Pb{\lvert \int f\,d\hat{\mu}_{U,\hat U}-\int f\,d\mu_{X,V} \rvert>\epsilon}\to 0 \qquad \mbox{as~} n\to \infty\,.$$
where $\mu^{(\sigma_t)}$ is the joint distribution of $(X,V=X+\sigma_t W)$ defined in \eqref{eq:amp_scalar}, and $\sigma_{t}$ can further be determined from the so called state evolution sequence: Set $\sigma_0^2=\frac{\mu}{\PT}+\mu$
and then
\begin{equation}
\label{eq:amp_se}
\sigma_t^2=\frac{\mu}{\PT}+\mu M\mmse(\sigma^2_{t-1})
\end{equation}
where $\mmse$ is defined in \eqref{eq:amp_scalar_mmse}.

 Note that the assumptions on $U$ $A$ and $Z$ in \cite[Lemma 1(b)]{bayati2011dynamics} hold in our case. In particular, since the support of $U$ is sampled uniformly from all block sparse supports of size $K$ and the entries in the support are iid $\cn(0,1)$ random variables, we have that the empirical distribution of entries of $U$ converge weakly almost surely to the distribution $P_X$ of $X$ defined in \eqref{eq:ber_gauss}. Further the moment conditions in \citep[Theorem 2]{bayati2011dynamics} are also satisfied. 
%The measure $\mu_{X,V}$ turns out to coincide with $\mu^{(\sigma_t)}$ implicitly defined in~\eqref{eq:amp_scalar}.
%
We note here that although \cite{bayati2011dynamics,reeves2012sampling} consider only real valued signals, the results there also hold for the complex case (see \citep[Theorem III.15]{maleki2013asymptotic}, \citep[Chapter 7]{barbier2015statistical}).
 
We next consider the support recovery in the scalar model \eqref{eq:amp_scalar}. Let $S_0=1[X\neq 0]$ denote the
indicator of the event when $X$ is non-zero. Let $\hat{S}_0\equiv\hat S_0(\theta)=1[|V|^2>\theta]$ denote an estimator
of $S_0$ using the observation $V=X+\sigma W$ in~\eqref{eq:amp_scalar}. Let 
\begin{equation}
\label{eq:amp_scalar_pe}
\psi(\sigma^2,\theta,M)=\Pb{S_0\neq \hat{S}_0}
\end{equation} 
 denote the probability of error in the scalar model \eqref{eq:amp_scalar} with $\sigma$ dependence made explicit as an argument of $\psi$. The from the convergence of $\hat \mu_{U,\hat U^{(t)}}$ we conclude as in~\citep{reeves2012sampling} that for any number $t$ of 
steps of the AMP algorithm $\hat S^{(t)}(\theta)$ achieves 
\begin{equation}\label{eq:amp_pupe2}
	\lim_{n\to\infty} \text{PUPE}(\hat S^{(t)}(\theta)) \le M\psi(\sigma^2_t, \theta, M)\,.
\end{equation}
Since this holds for any $t$ and any $\theta$ we can optimize both by taking $t\to\infty$ and $\inf_{\theta>0}$. 
From the proof of \citep[Theorem 6]{reeves2012sampling} it follows that $\lim_{t\to\infty}\sigma_t^2=\sigma^2_{\infty}$
exists and $\sigma_{\infty}$ satisfies~\eqref{eq:amp_fp}. 
The
proof is completed by the application of the following Claim, which allows us to compute infimum over $\theta$ in closed
form.
\end{proof}
\begin{claim}
\label{claim:amp_psi_opt}
\begin{equation}
\label{eq:amp_psi_opt}
M\inf_{\theta}\psi(\tau,\theta,M)=1-\frac{1}{1+\tau} \frac{1}{\left((M-1)\left(\frac{1}{\tau}+1\right)\right)^{\tau}}
\end{equation}
\end{claim}
\begin{proof}
Let us define $\tau=\sigma^2$. We have
\begin{\Ieee}{LLL}
\psi(\tau,\theta,M)&=&\Pb{S_0\neq \hat S_0}\\
&=& \frac{1}{M}\Pb{\hat S_0=0|S_0=1}+\\
&&\left(1-\frac{1}{M}\right)\Pb{\hat S_0=1|S_0=0}\\
\end{\Ieee}

Now conditioned on $S_0=1$, $|V|^2\distas{}(1+\tau)\mathrm{Exp}(1)$ and conditioned on $S_0=0$, $|V|^2\distas{}\tau \mathrm{Exp}(1)$ where $\mathrm{Exp}(1)$ is the Exponential distribution with density function $p(x)=e^{-x}1[x\geq 0]$.
Hence
\begin{equation}
\label{eq:amp_psi}
\psi(\tau,\theta,M)=\frac{1}{M}\left(1-e^{-\frac{\theta}{1+\tau}}\right)+\left(1-\frac{1}{M}\right)e^{-\frac{\theta}{\tau}}
\end{equation}
The claim follows by optimizing \eqref{eq:amp_psi} over $\theta$. The optimum occurs at 
$$\theta^*=\tau(1+\tau)\ln\left(\frac{1+\tau}{\tau}(M-1)\right)$$
 Substituting $\theta^*$ in \eqref{eq:amp_psi} proves the claim.
\end{proof}
% }
\subsection{Converse}

In this section we derive a converse for $\mathcal{E}^*$, based on the Fano inequality and the results from \citep{polyanskiy2011minimum}.

\begin{theorem}
Let $M$ be the codebook size. Given $\epsilon$ and $\mu$, let $S=\mu \log M$. Then assuming that the distribution of $|H|^2$ has a density with $\Ex{|H|^2}=1$ and $\Ex{|H|^4}< \infty$, $\mathcal{E}^*(M,\mu,\epsilon)$ satisfies the following two bounds
\begin{enumerate}
\item \begin{\Ieee}{LLL}
\label{eq:sc_conv_2}
\mathcal{E}^*(M,\mu,\epsilon)\geq \inf \frac{\PT}{S}\Ieeen
\end{\Ieee}
where infimum is taken over all $\PT>0$ that satisfies 
\begin{\Ieee}{LLL}
\label{eq:sc_conv_3}
\theta S -\epsilon\mu \log\left( 2^{S/\mu}-1\right) -\mu h_2(\epsilon)\leq\\
 \log\left(1+\PT\alpha\left(1-\theta,1\right)\right),\,\forall \theta\in [0,1]\Ieeen
\end{\Ieee}
where $\alpha(a,b)=\int_{a}^{b}F_{|H|^2}^{-1}(1-\gamma)d\gamma$, and $F_{|H|^2}$ is the CDF of squared absolute value of the fading coefficients.

\item \begin{\Ieee}{LLL}
\label{eq:sc_conv_4}
\mathcal{E}^*(M,\mu,\epsilon)\geq \inf \frac{\PT}{S}\Ieeen
\end{\Ieee}
where infimum is taken over all $\PT>0$ that satisfies 
\begin{\Ieee}{LLL}
\label{eq:sc_conv_5}
\epsilon \geq 1-\Ex{ Q\left(Q^{-1}\left(\frac{1}{M}\right)-\sqrt{\frac{2\PT}{\mu}|H|^2}\right)}\\
\Ieeen
\end{\Ieee}
where $\mathcal{Q}$ is the complementary CDF function of the standard normal distribution.
\end{enumerate}

\begin{proof}

First, we use the Fano inequality.

Let $W=(W_1,...,W_K)$, where $W_i\distas{iid}Unif[M]$ denote the sent messages of $K$ users. Let $X=(X_1,...,X_K)$ where $X_i\in\mathbb{C}^n$ be the corresponding codewords, $Y\in\mathbb{C}^n$ be the received vector. Let $\hat{W}=\left(\hat{W}_1,...\hat{W}_K\right)$ be the decoded messages. Then $W\to X\to Y\to\hat{W}$ forms a Markov chain. Then $\epsilon=P_e=\frac{1}{K}\sum_{i\in[K]}\Pb{W_i\neq \hat{W}_i}$.

Suppose a genie $G$ reveal a set $S_1\subset [K]$ for transmitted messages $W_{S_1}=\{W_i:i\in S_1\}$ and the corresponding fading coefficients $H_{S_1}$ to the decoder. So, a converse bound in the Genie case is a converse bound for our problem (when there is no Genie). Further, the equivalent channel at the receiver is 
\begin{equation}
\label{eq:sc_conv_fan1}
Y_G=\sum_{i\in S_2}H_i X_i +Z
\end{equation}
where $S_2=[K]\setminus S_1$, and the decoder outputs a $[K]$ sized tuple. So, PUPE with Genie is given by 
\begin{equation}
\label{eq:sc_conv_fan2}
P_e^G=\frac{1}{K}\sum_{i\in[K]}\Pb{W_i\neq \hat{W}^G_i}.
\end{equation}

Now, it can be seen that the optimal decoder must have the codewords revealed by the Genie in the corresponding locations in the output tuple, i.e., if $\hat{W}^G$ denotes the output tuple (in the Genie case), for $i\in S_1$, we must have that $W_i=\hat{W}^G_i$. Otherwise, PUPE can be strictly decreased by including these Genie revealed codewords. 

So, letting $E_i=1[W_i\neq \hat{W}^G_i]$ and $\epsilon_i^G=\Ex{E_i}$, we have that $\epsilon_i^G=0$ for $i\in S_1$. For $i\in S_2$, a Fano type argument gives
\begin{\Ieee}{LLL}
\label{eq:sc_conv_fan3}
I(W_i;\hat{W}^G_i)\geq \log M -\epsilon_i^G\log(M-1)-h_2(\epsilon_i^G).\Ieeen
\end{\Ieee}

So, using the fact that 
\begin{\Ieee}{LLL}
\sum_{i\in S_2}I(W_i;\hat{W}^G_i)&\leq & I(W_{S_2};\hat{W}^G_{S_2})\\
&\leq & n\Ex{\log(1+P\sum_{i\in S_2}|H_i|^2)}
\end{\Ieee}
 we have
\begin{\Ieee}{LLL}
\label{eq:sc_conv_fan4}
|S_2|\log M -\sum_{i\in S_2}\epsilon_i^G \log(M-1)-\sum_{i\in S_2}h_2(\epsilon_i^G)\\
\leq  n\Ex{\log(1+P\sum_{i\in S_2}|H_i|^2)}.\Ieeen
\end{\Ieee}

By concavity of $h_2$, we have 
\begin{equation}
\label{eq:sc_conv_fan5}
\frac{1}{K}\sum_{i\in S_2}h_2(\epsilon_i^G)=\frac{1}{K}\sum_{i\in [K]}h_2(\epsilon_i^G)\leq h_2(P_e^G).
\end{equation}

Hence we get 
\begin{\Ieee}{LLL}
\label{eq:sc_conv_fan6}
\frac{|S_2|}{K}\log M -P_e^G \log(M-1)-h_2(P_e^G)\\
\leq  \frac{n}{K}\Ex{\log(1+P\sum_{i\in S_2}|H_i|^2)}.\Ieeen
\end{\Ieee}

Next, notice that $P_e^G\leq P_e\leq 1-\frac{1}{M}$ and hence $P_e^G \log(M-1)+h_2(P_e^G)\leq P_e \log(M-1)+h_2(P_e)$. Further the inequality above hols for all $S_2\subset [K]$ (which can depend of $H_{[K]}$ as well). Hence, letting $|S_2|=\theta K$
\begin{\Ieee}{LLL}
\label{eq:sc_conv_fan7}
\theta\log M -P_e \log(M-1)-h_2(P_e)\\
\leq \frac{1}{\mu}\Ex{\log\left(1+\inf_{S_2:|S_2|=\theta K} \frac{\PT}{K}\sum_{i\in S_2}|H_i|^2\right)}.\Ieeen
\end{\Ieee}

Now, taking limit as $K\to \infty$ and using results on strong laws of order statistics \citep[Theorem 2.1]{van1980strong}, we get that 
\begin{\Ieee}{LLL}
\label{eq:sc_conv_fan8}
\log\left(1+\inf_{S_2:|S_2|=\theta K} \frac{\PT}{K}\sum_{i\in S_2}|H_i|^2\right)\\
\to \log\left(1+\PT\alpha(1-\theta,1)\right).\Ieeen
\end{\Ieee}

For any $a,b\in [0,1]$ with $a<b$, let $S_K\equiv S_K(a,b)=\frac{1}{K}\sum_{i=aK}^{bK}|H_{(i)}|^2$. Note that $S_K\to \alpha(a,b)$ as $K\to\infty$. Then 
\begin{\Ieee}{LLL}
\label{eq:sc_conv_fan9}
\Ex{S_K^2}&\leq & \Ex{\left(\frac{1}{K}\sum_{i=1}^{K}|H_i|^2\right)^2}\\
&=& 1+\frac{\Ex{|H|^4}-1}{K}\leq \Ex{|H|^4}.\Ieeen
\end{\Ieee}

Hence the family of random variables $\{S_K:K\in \mathbb{N}\}$ is \textit{uniformly integrable}. Further 
$$0\leq \log(1+\PT S_K)\leq \PT S_K.$$
 Hence the family $\{\log(1+\PT S_K):K\geq 1\}$ is also uniformly integrable. Then from theorem~\citep[Theorem 9.1.6]{rosenthal2006first}, 
 $$\Ex{\log(1+\PT S_K)}\to \log(1+\PT\alpha(a,b)).$$
  Using this in \eqref{eq:sc_conv_fan7} with $a=1-\theta$ and $b=1$, we obtain \eqref{eq:sc_conv_3}.

Next we use the result from \citep{polyanskiy2011minimum} to get another bound.

Using the fact that $S/\mu$ bits are needed to be transmitted under a per-user error of $\epsilon$, we can get a converse on the minimum $E_b/N_0$ required by deriving the corresponding results for a single user quasi-static fading MAC. In \citep{polyanskiy2011minimum}, the authors gave the following non-asymptotic converse bound on the minimum energy required to send $k$ bits for an AWGN channel. Consider the single user AWGN channel $Y=X+Z$, $Y,X\in\mathbb{R}^\infty$, $Z_i\distas{iid}\mathcal{N}(0,1)$. Let $M^*(E,\epsilon)$ denote the largest $M$ such that there exists a $(E,M,\epsilon)$ code for this channel: codewords $(c_1,...,c_M)$ with $\norm{c_i}^2\leq E$ and a decoder such that probability of error is smaller than $\epsilon$. The following is a converse bound from \citep{polyanskiy2011minimum}.

\begin{lemma}[\citep{polyanskiy2011minimum}]
\label{lem:conv_awgn}
Any $(E,M,\epsilon)$ code satisfies 
\begin{\Ieee}{LLL}
\label{eq:sc_con1}
\frac{1}{M}\geq Q\left(\sqrt{2E}+Q^{-1}\left(1-\epsilon\right)\right)\\\Ieeen
\end{\Ieee}
\end{lemma}

Translating to our notations, for the channel $Y=HX+Z$, conditioned on $H$, if $\epsilon(H)$ denotes the probability of error for each realization of $H$, then we have 

\begin{\Ieee}{LLL}
\label{eq:sc_con2}
\frac{1}{M}\geq Q\left(\sqrt{\frac{2\PT}{\mu}|H|^2}+Q^{-1}\left(1-\epsilon(H)\right)\right).\Ieeen
\end{\Ieee}

Further $\Ex{\epsilon(H)}=\epsilon$. Therefore we have
\begin{\Ieee}{LLL}
\label{eq:sc_con3}
\epsilon \geq 1-\Ex{ Q\left(Q^{-1}\left(\frac{1}{M}\right)-\sqrt{\frac{2\PT}{\mu}|H|^2}\right)}.\Ieeen
\end{\Ieee}

Hence we have the required converse bound.

\begin{remark}
We also get the following converse from \citep[theorem 7]{yang2014quasi} by taking the appropriate limits $P=\frac{\PT}{\mu n}$ and $n\to \infty$.
\begin{\Ieee}{LLL}
\label{eq:sc_con4}
\log M\leq -\log\left(\Ex{\mathcal{Q}\left(\frac{c+\frac{\PT |H|^2}{\mu}}{\sqrt{\frac{2\PT |H|^2}{\mu}}}\right)}\right)\Ieeen
\end{\Ieee}
where $c$ satisfies
\begin{\Ieee}{LLL}
\label{eq:sc_con5}
\Ex{\mathcal{Q}\left(\frac{c-\frac{\PT |H|^2}{\mu}}{\sqrt{\frac{2\PT |H|^2}{\mu}}}\right)}=1-\epsilon.\Ieeen
\end{\Ieee}

But this is strictly weaker than \eqref{eq:sc_con3}. This is because, using lemma \ref{lem:conv_awgn}, we perform hypothesis testing (in the meta-converse) for each realization of $H$ but in the bound used in \citep{yang2014quasi}, hypothesis testing is performed over the joint distribution (including the distribution of $H$). This is to say that if $H$ is presumed to be constant (and known), then in \eqref{eq:sc_con4} and \eqref{eq:sc_con5} we can remove the expectation over $H$ and this gives precisely the same bound as \eqref{eq:sc_con2}.
\end{remark}

\end{proof}

\end{theorem}

Bounds tighter than \eqref{eq:sc_conv_3} can be obtained if further assumptions are made on the codebook. For instance, if we assume that each codebook consists of iid entries of the form $\frac{C}{K}$ where $C$ is sampled from a distribution with zero mean and finite variance, then using ideas similar to \citep[Theorem 3]{reeves2013approximate} we have the following converse bound.

\begin{theorem}
\label{Th:conv2}
Let $M$ be the codebook size, and let $\mu n$ users ($\mu<1$) generate their codebooks independently with each code symbol iid of the form $\frac{C}{K}$ where $C$ is of zero mean and variance $\PT$. Then in order for the iid codebook to achieve PUPE $\epsilon$ with high probability, the energy-per-bit $\mathcal{E}$ should satisfy
\begin{\Ieee}{LLL}
\label{eq:sc_conv_new1}
\mathcal{E}\geq \inf \frac{\PT}{\mu\log M}\Ieeen
\end{\Ieee}
where infimum is taken over all $\PT>0$ that satisfies 
\begin{\Ieee}{LLL}
\label{eq:sc_conv_new2}
\ln M - \epsilon \ln(M-1)-h(\epsilon)\\
\leq  \left(M\mathcal{V}\left(\frac{1}{\mu M},\PT\right)-\mathcal{V}\left(\frac{1}{\mu},\PT\right)\right) \Ieeen
\end{\Ieee}
where $\mathcal{V}$ is given by \citep{reeves2013approximate}
\begin{\Ieee}{LLL}
\mathcal{V}(r,\gamma)&=& r\ln\left(1+\gamma-\mathcal{F}(r,\gamma)\right)+\ln\left(1+r\gamma-\mathcal{F}(r,\gamma)\right)\\
&&-\frac{\mathcal{F}(r,\gamma)}{\gamma}\label{eq:sc_conv_new3a}\Ieeen\\
\mathcal{F}(r,\gamma)&=&\frac{1}{4}\left(\sqrt{\gamma\left(\sqrt{r}+1\right)^2+1}-\sqrt{\gamma\left(\sqrt{r}-1\right)^2+1}\right)^2\label{eq:sc_conv_new3b}\\
\Ieeen\\
%\mathcal{V}_{LB}(r,\gamma)=\ln\left(1+\gamma r \left(\frac{r}{r-1}\right)^{r-1}\frac{1}{e}\right)\label{eq:sc_conv_new3c}\Ieeen
\end{\Ieee}

\end{theorem}
\begin{proof}[Proof sketch]
The proof is almost the same as in \citep[Theorem 3]{reeves2013approximate} (see \citep[Remark 3]{reeves2013approximate} as well). We will highlight the major differences here. First, our communication system can be modeled as a support recovery problem as follows. Let $A$ be the $n\times KM$ matrix consisting of $n\times M$ blocks of codewords of users. Let $\boldsymbol{H}$ be the $KM\times KM$ block diagonal matrix with block $i$ being a diagonal $M\times M$ matrix with all diagonal entries being equal to $H_i$. Finally let $W\in\{0,1\}^{KM}$ with $K$ blocks of size $M$ each and within each $M$ sized block, there is exactly one $1$. So the position of $1$ in block $i$ of $W$ denotes the message or codeword corresponding to the user $i$ which is the corresponding column in block $i$ of matrix $A$. Hence our channel can be represented as 
\begin{equation}
\label{eq:sys_comp_sens}
Y=A\boldsymbol{H}W+Z
\end{equation}
with the goal of recovering $W$.

Next the crucial step is bound $R^K (\epsilon,M)$ in \eqref{eq:sc_conv_fan3} as 
\begin{\Ieee}{LLL}
\label{eq:sc_conv_new4}
R^K(\epsilon,M)&\leq & I(W;Y|A)\\
&=& I(\boldsymbol{H}W;Y|A)-I(\boldsymbol{H}W;Y|A,W)\Ieeen
\end{\Ieee}
where the equality in the above display follows from \citep[equation (78)]{reeves2013approximate}. The first term in above display is bounded as
\begin{\Ieee}{LLL}
\label{eq:sc_conv_new5}
 I(\boldsymbol{H}W;Y|A=A_1) &=& I(\boldsymbol{H}W;A_1 HW+Z ) \\
 &\leq & \sup_{U} I(U;A_1 U+Z) \Ieeen
\end{\Ieee}
where $A_1$ is a realization of $A$ and supremum is over random vectors $U\in \mathbb{C}^{KM}$ such that $\Ex{U}=0$ and $\Ex{UU^*}=\Ex{(\boldsymbol{H}W)(\boldsymbol{H}W)^* }=\frac{\Ex{|H_1|^2}}{M}I_{KM\times KM}$. Now similar to \citep{reeves2013approximate}, the supremum is achieved when 
$$U\distas{}\cn\left(0,\frac{\Ex{|H_1|^2}}{M}I_{KM\times KM}\right).$$
 Hence
\begin{\Ieee}{LLL}
\label{eq:sc_conv_new6}
I(\boldsymbol{H}W;Y|A=A_1)\leq \log \det \left(I_{n\times n}+\frac{1}{M}AA^*\right).\Ieeen
\end{\Ieee}

Next, for any realization $A_1$ and $W_1$ of $A$ and $W$ respectively, we have
\begin{\Ieee}{LLL}
\label{eq:sc_conv_new7}
&&I(\boldsymbol{H}W;Y|A=A_1,W=W_1)\\
&=&I(\boldsymbol{H}W_1;A_1\boldsymbol{H}W_1+Z)\\
&=& I(\tilde{\boldsymbol{H}};(A_1)_{W_1}\tilde{\boldsymbol{H}}+Z)\\
&\geq & I(\tilde{\boldsymbol{H}};\tilde{\boldsymbol{H}}+(A_1)_{W_1}^{\dagger}Z)\Ieeen
\end{\Ieee}
where $\tilde{\boldsymbol{H}}=[H_1,...,H_{K}]^T$ and $(A_1)_{W_1}$ is the $n\times K$ submatrix of $A_1$ formed by columns corresponding to the support of $W_1$ and $\dagger$ denotes the Moore-Penrose inverse (pseudoinverse). The last equality in the above follows from the data processing inequality. Now, by standard mutual information of Gaussians, we have 
\begin{\Ieee}{LLL}
\label{eq:sc_conv_new7a}
I(\tilde{\boldsymbol{H}};\tilde{\boldsymbol{H}}+(A_1)_{W_1}^{\dagger}Z)\\
=\log\det\left(I_{K\times K}+\left((A_1)_{W_1}\right)^*(A_1)_{W_1}\right)\Ieeen.
\end{\Ieee}

Hence
\begin{\Ieee}{LLL}
\label{eq:sc_conv_new7b}
I(\boldsymbol{H}W;Y|A,W)= \Ex{\log\det \left(I_{K\times K}+A^{*}_{W}A_{W}\right)}.\Ieeen
\end{\Ieee}

Hereafter, the we can proceed similarly to the proof of \citep[Theorem 3]{reeves2013approximate} using results from random-matrix theory \citep{verdu1999spectral,salo2006asymptotic} to finish the proof.
\end{proof}

We remark here that for a general fading distribution, the term $I(\tilde{\boldsymbol{H}};\tilde{\boldsymbol{H}}+(A_1)_{W_1}^{\dagger}Z)$ can be lower bounded similar to the proof of \citep[Theorem 3]{reeves2013approximate} using EPI (and its generalization \citep{zamir1993generalization}) to get
\begin{\Ieee}{LLL}
\label{eq:sc_conv_new8}
I(\tilde{\boldsymbol{H}};\tilde{\boldsymbol{H}}+\left((A_1)_{W_1}\right)^{\dagger}Z)\\
\geq K\log \left(1+N_{H}\left(\det\left(\left((A_1)_{W_1}\right)^*(A_1)_{W_1}\right)\right)^{\frac{1}{K}}\right)\Ieeen
\end{\Ieee}
where $N_H=\frac{1}{\pi e}\exp(h(H))$ is the entropy power of fading distribution. Hence
\begin{\Ieee}{LLL}
\label{eq:sc_conv_new9}
I(\boldsymbol{H}W;Y|A,W)\\
\geq K\Ex{\log \left(1+N_{H}\left(\det\left(A^{*}_{W}A_{W}\right)\right)^{\frac{1}{K}}\right)}.\Ieeen
\end{\Ieee}

Again, we can use results from random-matrix theory  \citep{salo2006asymptotic} and  proceed similarly to the proof of \citep[Theorem 3]{reeves2013approximate} to get a converse bound with the second term in \eqref{eq:sc_conv_new2} replaced by $\mathcal{V}_{LB}\left(\frac{1}{\mu},\PT\right)$ and 
\begin{equation}
\label{eq:sc_conv_new10}
\mathcal{V}_{LB}(r,\gamma)=\ln\left(1+\gamma r \left(\frac{r}{r-1}\right)^{r-1}\frac{1}{e}\right)
\end{equation}

We make a few observations regarding the preceding theorem. First and foremost, this hold only for the case of no-CSI because the term analogous to $I(\boldsymbol{H}W;Y|A,W)$ in the case of CSIR is $ I(\boldsymbol{H}W;Y|A,\boldsymbol{H},W)$ which is zero. Next, it assumes that the codebooks have iid entries with variance scaling $\Theta(1/n)$. This point is crucial to lower bounding $I(\boldsymbol{H}W;Y|A,W)$, and this is where a significant improvement comes when compared to \eqref{eq:sc_conv_3}. Indeed, EPI and results from random matrix theory give $O(n)$ lower bound for $I(\boldsymbol{H}W;Y|A,W)$. This once again brings to focus the the difference between classical regime and the scaling regime, where in the former, this term is negligible. Further this leaves open the question of whether we could improve performance in the high-density of users case by using non-iid codebooks.

 %{\color{blue}{
 Now, as to what types of codebooks give a $\Theta(n)$ lower bound for $I(\boldsymbol{H}W;Y|A,W)$, a partial answer can be given by carefully analyzing the full proof of the theorem. In particular, if $\mathcal{S}=\mathrm{supp}W$ i.e, the support of $W$, then as seen from \citep[equation (85)]{reeves2013approximate}, any non zero lower bound on $\det(A_{\mathcal{S}}^{*}A_{\mathcal{S}})^{1/K}$ in the limit is enough. So if the matrix $A_{\mathcal{S}}^{*}A_{\mathcal{S}}$ possesses strong diagonal dominance then it is possible to have such a non zero lower bound on $\det(A_{\mathcal{S}}^{*}A_{\mathcal{S}})^{1/K}$ for every $S$ \citep{li2012some}. These could be ensured by having codewords that are overwhelmingly close to orthogonal.
 %}}

\section{Numerical evaluation and discussion}\label{sec:numerical}
In this section, we provide the results of numerical evaluation of the bounds in the paper. We focus on the trade-off of user density $\mu$ with the minimum energy-per-bit $\mathcal{E}^*$ for a given message size $k$ and target probability of error $P_e$. 

For $k=100$ bits, we evaluate the trade-off from the bounds in this paper for $P_e=0.1$ and $P_e=0.001$ in figures
\ref{fig:3} and \ref{fig:4} respectively. For TDMA, we split the frame of length $n$ equally among $K$ users, and
compute the smallest $P_{tot}$ the ensures the existence of a single user quasi-static AWGN code of rate $S$,
blocklength $\frac{1}{\mu}$ and probability of error $\epsilon$ using the bound from \citep{yang2014quasi}. The
simulations of the single user bound is performed using codes from \citep{Spectre}. TIN is computed using a method
similar to theorem \ref{th:scaling_CSIR}. In particular, the codeword of user $i$ is decoded as $\hat c_i=\arg\min_{c'\in \cd_i}\norm{Y-H_i c_i}^2$ where we assume that the decoder has the knowledge of CSI. The analysis proceeds in a similar way as theorem \ref{th:scaling_CSIR}.

\textit{Achievability bounds.} It can be seen that for small $\mu$ the scalar-AMP bound of Theorem~\ref{th:amp_ach} is
better than the projection decoder bounds of Theorems~\ref{th:scaling_CSIR} and~\ref{th:scaling_noCSI}. The latter
bounds have another artifact. For example, the no-CSI bound on $\mathcal{E}^*$ from Theorem~\ref{th:scaling_noCSI}
increases sharply as $\mu\downarrow 0$, in fact one can show that the said bound behaves as
$\mathcal{E}=\Omega(\sqrt{-\ln \mu})$.

\textit{Engineering insights.}
From these figures, we clearly observe the \textit{perfect MUI cancellation} effect mentioned in the introduction and
previously observed for the non-fading model~\citep{polyanskiy2017perspective,polyanskiy2018_course}. Namely, as
$\mu$ increases from $0$, the $\mathcal{E}^*$ is almost a constant, $\mathcal{E}^*(\mu,\epsilon,k) \approx
\mathcal{E}_{\mathrm{s.u.}}(\epsilon,k)$ for $0<\mu < \mu_{\mathrm{s.u.}}$. As $\mu$ increases beyond
$\mu_{\mathrm{s.u.}}$ the tradeoff undergoes a ``phase transition'' and the
energy-per-bit $\mathcal{E}^*$ exhibits a more familiar increase with $\mu$. Further,
standard schemes for multiple-access like TDMA and TIN do not have this behavior. Moreover, although these suboptimal
schemes have an optimal trade-off at $\mu\to 0$ they show a significant suboptimality at higher $\mu$. We note again
that this perfect MUI cancellation which was observed in standard GMAC
\citep{polyanskiy2017perspective,polyanskiy2018_course} is also present in the more practically relevant quasi-static
fading model. So, we suspect that this effect is a general characteristic of the many-user MAC. 

\textit{Suboptimality of orthogonalization.}
The fact that orthogonalization is not optimal is one of the key practical implications of our work. It
was observed before in the GMAC and here we again witness it in the more relevant QS-MAC. How to understand this suboptimality? First, in the fading case we have already seen this effect even in the classical regime (but under PUPE) -- see \eqref{eq:TDMA_ebn0}. To give another intuition we consider a $K=\mu n$ user binary adder MAC

\begin{equation}
\label{eq:adder_MAC}
Y=\sum_{i=1}^{K}X_i
\end{equation}
where $X_i\in \{0,1\}$ and addition is over $\mathbb{Z}$. Now, using TDMA on this channel, each user can send at most $n/K=1/\mu$ bits. Hence the message size is bounded by

\begin{equation}
\label{eq:msg_TDMA}
\log M\leq \frac{1}{\mu}.
\end{equation}

Next, let us consider TIN. Assume $X_i\distas{}\mathrm{Ber}(1/2)$. For user 1, we can treat $\sum_{i=2}^{\mu n}X_i$  as noise. By central limit theorem, this noise can be approximated as $\sqrt{\frac{1}{4}\mu n}Z$ where $Z\distas{}N(0,1)$. Thus we have a binary input AWGN (BIAWGN) channel
\begin{equation}
\label{eq:BIAWGN}
Y=X_1+\sqrt{\frac{1}{4}\mu n}Z.
\end{equation}
Therefore, the message size is bounded as
\begin{\Ieee}{LLL}
\label{eq:msg_TIN}
\log M &\leq n C_{\mathrm{BIAWGN}}\left(1+\frac{4}{\mu n}\right)\\
& \leq \frac{n}{2}\log \left(1+\frac{4}{\mu n}\right) \to \frac{2}{\mu \ln 2}\Ieeen
\end{\Ieee}
where $C_{\mathrm{BIAWGN}}$ is the capacity of the BIAWGN channel. Note that in both the above schemes the achievable message size is a constant as $n\to\infty$.

On the other hand, the true sum-capacity of the $K$-user adder MAC is given by
%	$$ C_{sum} = \max_{X_1 \dperp \cdots \dperp X_K} H(X_1 + \cdots +X_K)\,.$$
$$ C_{sum} = \max_{X_1 , \cdots , X_K} H(X_1 + \cdots +X_K)\,.$$
As shown in~\citep{shepp1981entropy} this maximum as achieved at $X_i \stackrel{iid}{\sim} \mathrm{Ber}(1/2)$. Since the 
the entropy of binomial distributions~\citep{knessl1998integral} can be computed easily, we obtain
	$$ C_{sum} = \frac{1}{2} \log K + o(\log K)\,.$$
In particular, for our many-user MAC setting we obtain from the Fano inequality (and assuming PUPE is small)
%	$$ \log M \simleq \frac{\log (\mu n)}{2\mu }\,.$$ 
$$ \log M \lessapprox \frac{\log (\mu n)}{2\mu }\,.$$ 
Surprisingly, there exist explicit codes that achieve this limit and with a very low-complexity (each message bit is
sent separately),-- a construction rediscovered several times~\citep{lindstrom1965combinatorial,cantor1966determination,khachatrian1995new}. 
Hence the optimal achievable message size is
\begin{equation}
\label{eq:msg_lindstrom}
\log M \approx \frac{\log n}{2\mu}\to\infty
\end{equation}
as $n\to\infty$. And again, we see that TDMA and TIN are severely suboptimal for the many-user adder MAC as well.

 \begin{figure}[h]
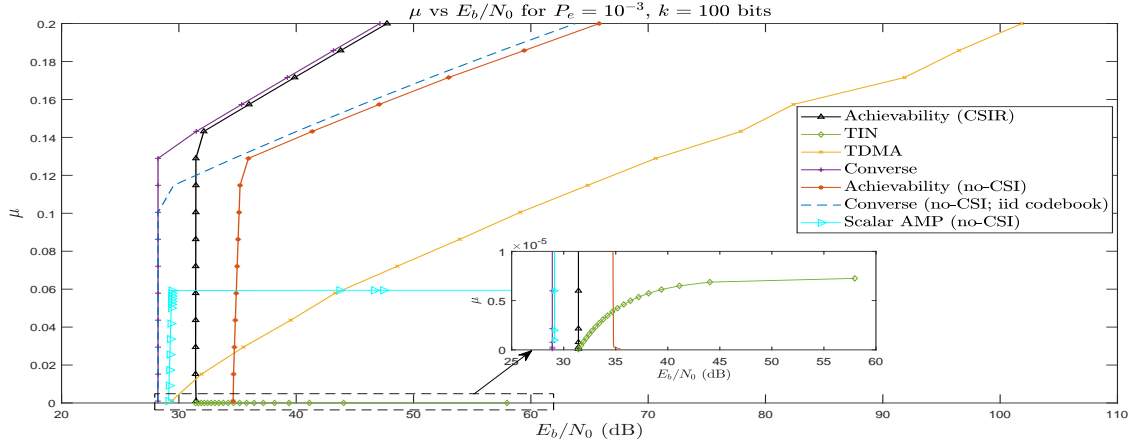

  \begin{center}
     %\resizebox{11cm}{5.2cm} {\includegraphics *[width=\linewidth]{{{isit_ebn0_mu_pe_0.001_k_100_new}}}}
     \includegraphics *[height=5.5cm,width=\columnwidth]{{TIT_ebn0_mu_pe_0.001_k_100_new2020}}
     \caption {$\mu$ vs $E_b/N_0$ for $\epsilon\leq 10^{-3}$, $k=100$}
     \label{fig:4}
      \end{center} 
 \end{figure}

 \begin{figure}[h]
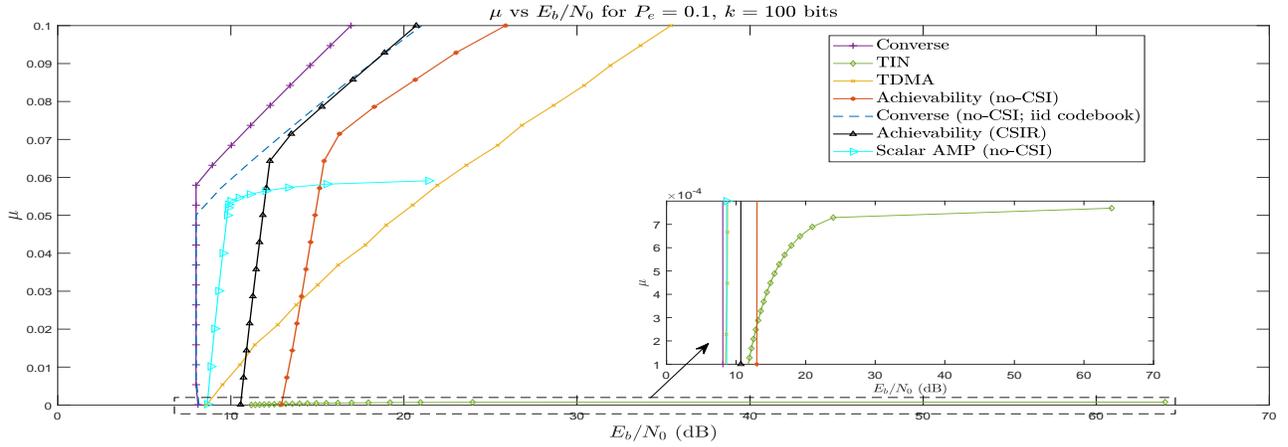

  \begin{center}
    % \resizebox{10cm}{5.2cm} {\includegraphics *[width=\linewidth]{{{isit_ebn0_mu_pe_0.1_k_100_new}}}}
    \includegraphics *[height=5.5cm,width=\columnwidth]{{TIT_ebn0_mu_pe_0.1_k_100_new2020_1}}
     \caption {$\mu$ vs $E_b/N_0$ for $\epsilon\leq 10^{-1}$, $k=100$}
     \label{fig:3}
      \end{center} 
 \end{figure}

\subsection{The ``curious behavior'' in phase transition}\label{sec:curious}

%{\color{blue}
As we emphasized in the Introduction, the most exciting conclusion of our work is the existence of the almost  vertical
part on the $\mu$ vs $\frac{E_b}{N_0}$ plots of Fig.~\ref{fig:3} and~\ref{fig:4}. In this section we want to explain how this effect arises, why it can be
called the ``almost perfect MUI cancellation'' and how it relates (but is not equivalent) to well-known phase transitions in compressed sensing.

To make things easier to evaluate, however, we depart from the model in the previous sections and do two relaxations.
First, we consider a non-fading AWGN. Second, we endow all users with the same codebook. The second assumption simply
means that the decoding from now on is only considered up to permutation of messages,
see~\cite{polyanskiy2017perspective} for more on this. Technically, these two assumptions mean that we are considering a
model~\eqref{eq:cs1} with $U$ vector that is $K$-sparse (as opposed to block-sparse) and that all non-zero entries of
$U$ are equal to 1. Finally, we will consider the real-valued channel. In all, we get the following signal
model~\cite[Section IV]{ZPT-isit19}:
\begin{equation}
\label{eq:awgn_cs_model}
Y^n=A U^p+Z^n\,, \qquad Z^n\distas{}\mathcal{N}(0,I_n)\,,
\end{equation}
with $A_{i,j} \stackrel{iid}{\sim}\mathcal{N}(0,b^2/n)$, $(i,j) \in {[n]\times [p]}$, $U_i \stackrel{iid}{\sim}\text{Ber}(K/p)$, so that $\EE[\|U^p\|_0] = K$. We take the proportional scaling limit with $K=\mu n$ and $p=KM$. 
Interpretation of these parameters in the context of communication problem are:
\begin{itemize}
\item $M$ as the number of messages that each user wants to communicate 
\item $\mu$ is the user density per (real) degree of freedom
\item $\PT \eqdef b^2\mu$ is the total received power from all $K$ users at the receiver. 
\item We also introduce \textit{the spectral efficiency (in nats)} as \begin{equation}
\label{eq:spec_eff_true}
S_e \eqdef \mu Mh\left(\frac{1}{M}\right) = \mu (\ln M - 1) + O(\frac{\mu}{M})
\end{equation}
where $h(x)=-x\ln x-(1-x)\ln(1-x)$ is the binary entropy function in nats. Note that $S_e = \frac{1}{n} H(U^p)$.
\item Consequently, we may define energy-per-bit as\footnote{\txb{This definition is slightly different, but appears more reasonable than the one in \citep{ZPT-isit19}}} \txb{$$\frac{E_b}{N_0}\triangleq \frac{\PT\ln 2}{2S_e}\,.$$}
\end{itemize}

Given $(Y^n,A)$, the decoder outputs an estimate $\hat{U}^p\in\{0,1\}^p$ with $\Ex{\|\hat{U}^p\|_0}=K$ and we are
interested in the minimal achievable PUPE, or
\begin{\Ieee}{LLL}
\label{eq:opt_pupe0}
P_e^*(\mu,M,b) \\
\eqdef \limsup_{n\to\infty} \min_{\hat U^p: \EE[\|\hat U\|_0] = K} \frac{1}{K}\sum_{i\in
[p]}\Pb{U_i=1, \hat{U}_i=0} \Ieeen
\end{\Ieee}
To discuss performance of the optimal decoder, we need to return to the scalar
channel~\eqref{eq:amp_scalar} with the following modifications: $X\distas{}\mathrm{Ber}(1/M)$,
$W\distas{}\mathcal{N}(0,1)$. Now, for every value of $\sigma$ in~\eqref{eq:amp_scalar} we may ask for the smallest
possible error $\epsilon^*(\sigma,M) = \min \Pb{\hat{X}\neq X}$ where minimum is taken over all estimators
$\hat{X}=\hat{X}(V)$ such that $\Pb{\hat{X}=1}=\frac{1}{M}$. As discussed in \citep[Section IV.B]{ZPT-isit19}, this
minimal $\epsilon^*(\sigma,M)$ satisfies \citep{ZPT-isit19} is found from solving: 
 \begin{equation}
 \label{eq:opt_pupe1}
\frac{1}{\sigma}=\cQ^{-1}\left(\frac{\epsilon^*}{M-1}\right)+\cQ^{-1}\left(\epsilon^*\right)
 \end{equation}
% which is also equivalent to
%  \begin{equation}
% \label{eq:opt_pupe2}
%{1\over \sigma}=Q^{-1}\left(\frac{\epsilon^*}{M-1}\right)-Q^{-1}\left(1-\epsilon^*\right)
% \end{equation}
where $\cQ(\cdot)$ is the complementary CDF of the standard normal distribution.

Now the limit $P_e^*$ in~\eqref{eq:opt_pupe0} can be computed via the replica method.\footnote{Note that in~\cite{reeves2019replica,barbier2016mutual} it was shown that the replica-method
prediction is correct for estimating $I(U_i; Y^n, A)$ and $\Var[U_i|Y^n,A]$, but what we need for computing the
$P_e$ is asymptotic distribution of a random variable $\Pb{U_i=1|Y^n,A}$. 
First, it is known that AMP initialized at the true value $U$ converges to an asymptotically MMSE-optimal estimate.
Second, distribution of the AMP estimates are known to belong to a $P_{X|V}$ in~\eqref{eq:amp_scalar} (with $\sigma$
identified from the replica method). Finally, 
any asymptotically MMSE-optimal estimator $\hat U$ should satisfy $\hat U_i \stackrel{(d)}{\to} \EE[U_i|Y^n,A] =
\Pb{U_i=1|Y^n,A}$, and thus $\Pb{U_i=1|Y^n,A}$ should match the replica-method predicted one.}
 Namely, replica predictions tell us that 
$$ P_e^*(b,\mu) = \epsilon^*(\sigma, M)\,,$$
where $\sigma^2=\frac{1}{\eta^* b^2}$ and the {\it multiuser efficiency} 
$$\eta^* = \eta^*(M,b,\mu)$$ is given by~\citep{guo2005randomly,guo2009single,ZPT-isit19} \txb{
 \begin{equation}
 \label{eq:MU_eff}
 \eta^*\equiv \eta_M^*\in \arg\min_{\eta\in[0,1]}\tilde{\cF}_M(\eta; b^2, \mu)
 \end{equation}
 where 
 \begin{equation}
 \tilde{\cF}_M(\eta; b^2, \mu)=\frac{I_M\left(\frac{\mu}{\eta \PT}\right)}{h(1/M)}+\frac{1}{2S_e}\phi(\eta)
 \end{equation}
where $\phi(x)=x-1-\ln x$ and} $I_M(\sigma^2)=I(X;X+\sigma W)$ is the mutual information between the signal $X\distas{}\mathrm{Ber}(1/M)$ and observation in the scalar channel~\eqref{eq:amp_scalar} \txb{in nats}.

 In the figure \ref{fig:5} we have shown the plots of optimal PUPE $P_e$ for the model \eqref{eq:awgn_cs_model} versus $E_b/N_0$ for various values of $\mu$ when $M=2^{100}$, computed via replica predictions. What is traditionally referred to as the phase transition in compressed sensing is the step-function drop from $P_e \approx 1$ to $P_e \ll 1$. However, there is a \textit{second effect} here as well. Namely that all the curves with different $\mu$ seem to have a \textit{common envelope}. The former has not only been observed in compressed sensing, e.g.~\cite[Fig.1]{reeves2012sampling,reeves2013approximate} and~\cite[Fig.4]{barbier2012compressed} among others, but also in a number of other inference problems: randomly-spread CDMA~\cite{guo2005randomly}, LDPC
 codes~\cite{vicente2003low} and random SAT~\cite{kirkpatrick1994critical}. However, the second effect appears to be a rather different phenomenon, and in fact it is exactly the one that corresponds to the existence of the vertical part of the curves on Fig.~\ref{fig:4}-\ref{fig:3}.

 \begin{figure}[h]
  \begin{center}
     %\resizebox{11cm}{5.2cm} {\includegraphics *[width=\linewidth]{{{isit_ebn0_mu_pe_0.001_k_100_new}}}}
    % \includegraphics *[height=5.5cm,width=\columnwidth]{TIT_pupe_ebn0_k100_new}
	 \includegraphics *[height=5.5cm,width=\columnwidth]{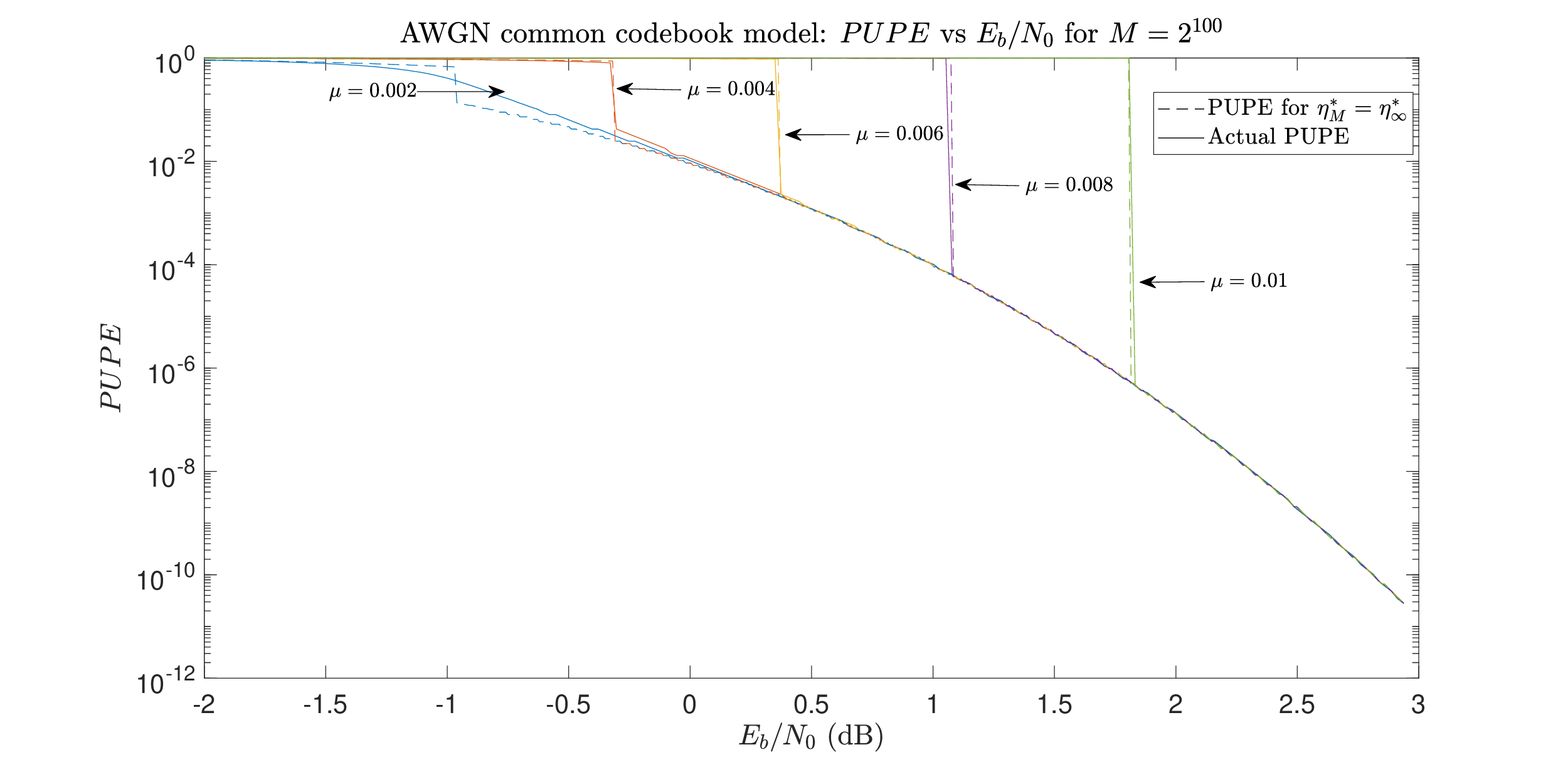}    
     %\caption {AWGN same codebook model: $P_e$ vs $E_b/N_0$ for $M=2^{100}$}
     \caption[]{\txb{AWGN same codebook model: $P_e$ vs $E_b/N_0$ for $M=2^{100}$. The solid lines correspond to using $\eta^*\equiv \eta_M^*$ from \eqref{eq:MU_eff} in~\eqref{eq:opt_pupe1}. The dashed lines correspond to  $\eta^*_{\infty}$ defined in \eqref{eq:eta_inf}.}}
     \label{fig:5}
      \end{center} 
 \end{figure}

Let us, for the moment, assume that the envelope is actually exactly the same for all $\mu$. Fix a value of PUPE $P_e=10^{-3}$ (say) and consider how the intercept of the horizontal line at $P_e=10^{-3}$ on Fig.~\ref{fig:5} changes with $\mu$. It is easy to see that as long as the value of $\mu$ is small enough the intercept will not be moving (corresponding to constancy of the $E_b/N_0$ as a function of $\mu$). However, once the value of $\mu$ exceeds a value (dependent on the fixed value of $P_e$) the intercept starts moving to the right together with the step-drop portion of the curves. From this we conclude that indeed, existence of the (almost) common envelope on Fig.~\ref{fig:5} results in the (almost) vertical part on Fig.~\ref{fig:4}-\ref{fig:3}. (As a side note, we also note that since the slanted portions of the tradeoff curves on those figures correspond to the vertical drop on the Fig.~\ref{fig:5} and hence the slanted portion is virtually independent of the fixed value of $P_e$ -- as predicted by~\eqref{eq:ebno_approx}.) 

How can the curves have common envelope? Notice that in the expression for $P_e^*$ only $\eta^*$ is a function of $\mu$. Thus, we conclude that for small $\mu$ we must have $\eta^*(\mu)\approx \mathrm{const}$. But as $\mu\to 0$ we should get a $\eta^*\to 1$. Thus we see that common envelopes are only possible if to the right of the step-drops on Fig.~\ref{fig:5} we get $\eta^*\approx 1$. Is that indeed so? Fig.~\ref{fig:6} provides an affirmative answer.

In reality, the ``vertical part'' is not truly vertical and the common envelope is not exactly common. In truth the right portions of the curves on Fig.~\ref{fig:5} (following the drop) are all very slightly different, but this difference is imperciptible to the eye (and irrelevant to an engineer). What makes them so close is the incredible degree of sparsity $\frac{1}{M} = 2^{-100}$. Indeed, as \txb{Fig.~\ref{fig:7}} demonstrates that as $M\to\infty$ the value of $\eta^*$ to the right of the step transition approaches 1.

To summarize, we conclude that what determines our ``curious behavior'' is not a sudden change in the estimation performance (typically credited as ``phase transition'' in compressed sensing), but rather a more subtle effect arising in the super-low sparsity limit: the step-transition of the parameter $\eta^*$ from a moderate value in the interior of $(0,1)$ to a value close to $1$. The fact that only the incredibly low sparsity values $\frac{1}{ M}$ are relevant for the many-MAC problems makes this new effect practically interesting.

\txb{To study $\eta^*$ in the limit of large $M$, we apply a clever observation of~\cite{reeves2019allB,reeves2019allB_arxiv}: in the limit of $M\to \infty$ (low sparsity), instead of using parameters $(M, \mu, b^2)$ it is better to use $(M, \PT, S_e)$. Holding $\PT$ and $S_e$ constant implies that 
\begin{align*}
\mu&=\frac{S_e}{Mh(1/M)} = \frac{S_e}{\ln M}\cdot (1+o(1))\\
   b^2 &= \frac{\PT {Mh(1/M)}}{S_e} = \frac{\PT}{S_e}\ln M \cdot (1+o(1))
\end{align*}
as $M\to \infty$. Under this scaling of $\mu$ and $b^2$ we get~\citep[Section 3]{reeves2019allB_arxiv}, for every $\eta\in(0,1]$ we have
\begin{align}
\label{eq:rescaled_replica_cgt}
\lim_{M\to\infty}\tilde{\cF}_M(\eta)&=\tilde{\cF}_{\infty}(\eta)\\
\label{eq:rescaled_replica_infty}
\tilde{\cF}_{\infty}(\eta)&\eqdef \min\left(1,\frac{\PT}{2S_e}\eta\right)+\frac{1}{2S_e}\phi(\eta)
\end{align}
Furthermore, for $S_e \neq \frac{1}{2}\ln(1+P_{tot})$ we have 
\begin{equation}
\label{eq:eta_inf}
\eta^*_M \to \eta_{\infty}^*=\begin{cases}
\frac{1}{1+\PT}, & \frac{1}{2}\ln(1+\PT)<S_e \\
1, & \frac{1}{2}\ln(1+\PT)>S_e
\end{cases}
\end{equation}

Thus, we see that (at least in this regime of $\mu$ and $b^2$) indeed the $\eta_M^*$ jumps from $<1$ to $\approx 1$ the energy-per-bit crosses over a threshold. Furthermore, in figures~\ref{fig:5},~\ref{fig:6} and \ref{fig:7} we show that $\eta_{\infty}^*(\PT,S_e)$ gives an excellent agreement with $\eta^*(M,\mu,b^2)$. 
}

\txb{
Incidentally, under this scaling of $\mu\to 0$ and $b^2\to \infty$ with a fixed $\PT$ and $S_e$, it turns out that PUPE experiences a step-transition from $0$ to $1$ upon crossing the ``slanted part'' of the $S_e$ vs $E_b/N_0$ diagram. This step transition is dubbed all-or-nothing phenomenon and was indeed the subject of~\cite{reeves2019allB}. 

To see that notice that~\eqref{eq:opt_pupe1} corresponds to
\begin{\Ieee}{LLL}
\label{eq:opt_pupe2}
\sqrt{2\eta_M^* Mh(1/M) \frac{E_b/N_0}{\ln 2}}&=&\cQ^{-1}\left(\frac{\epsilon^*}{M-1}\right)+\cQ^{-1}\left(\epsilon^*\right)\\
\Ieeen
\end{\Ieee}
Next, using the approximation 
\begin{equation}
\label{eq:qinv_appr}
\cQ^{-1}(\delta)\approx \sqrt{2\ln\frac{1}{\delta}-\ln\left(4\pi\ln\frac{1}{\delta}\right)}
\end{equation}
one can check that 
$$ \lim_{M\to \infty} \epsilon^*\to \begin{cases} 0, & \eta_{\infty}^* E_b/N_0>\ln 2\\
1, & \eta_{\infty}^* E_b/N_0<\ln 2
\end{cases}\,,$$
which incidentally corresponds to the two cases $\eta_{\infty}^*\frac{E_b}{N_0} \lessgtr -1.59~dB$. %Now we will use~\eqref{eq:eta_lim}. If $\frac{1}{2}\ln(1+\PT)>S$, then clearly $\eta_{\infty}^* E_b/N_0>\ln 2$. But if $\frac{1}{2}\ln(1+\PT)<S$ then using the inequality $\frac{x}{1+x}\leq \ln(1+x)$ for $x>-1$ we obtain
%$$\eta^*_{\infty}\frac{E_b/N_0}{\ln 2}=\frac{1}{1+\PT}\frac{\PT}{2S}< \frac{\PT}{(1+\PT)\ln(1+\PT)}\leq 1$$ 
From~\eqref{eq:eta_inf} we get
$$ \lim_{M\to \infty} \epsilon^*\to \begin{cases} 0, & \frac{1}{2}\ln(1+\PT)>S_e\\
							1, & \frac{1}{2}\ln(1+\PT)<S_e
							\end{cases}\,,$$

%The above shown $0$-$1$ jump in $\epsilon^*$ is same as the so-called ``All-or-Nothing'' (AoN) phenomenon in sparse linear regression~\citep{reeves2019allB,reeves2019allB_arxiv}. 
In all, we see that the ``mother'' effect is that of $\eta^*$ jumping from $<1$ to $\approx 1$ in the regime of low sparsity (large $M$). It causes both the ``curious behavior'' for us, and the all-or-nothing for sparse regression.

Finally we note that when $\eta_M^*\approx 1$, then from~\eqref{eq:opt_pupe2} the asymptotic expansion of $\ln(M-1)$ as $\frac{\PT}{\mu}\to\infty$ (and fixed $\epsilon$) matches with~\citep[Theorem 3]{polyanskiy2011minimum} up to the $O(1)$ term, thus showing that the replica prediction for the location of the vertical part matches with the finite-blocklength prediction in~\eqref{eq:E_su}.
}

 \begin{figure}[h]
  \begin{center}
     %\resizebox{11cm}{5.2cm} {\includegraphics *[width=\linewidth]{{{isit_ebn0_mu_pe_0.001_k_100_new}}}}
     \includegraphics *[height=5.5cm,width=\columnwidth]{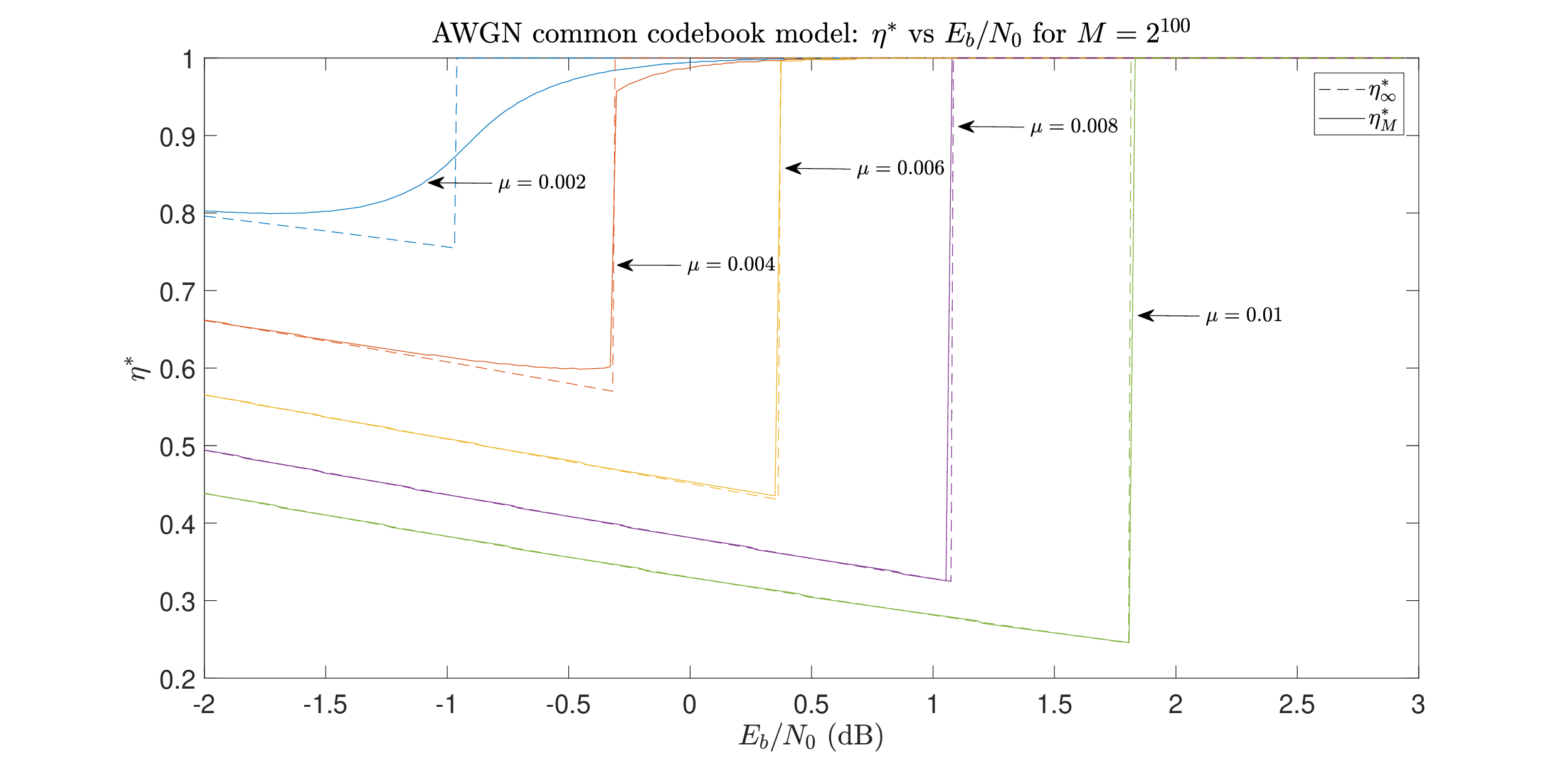}
     \caption[]{\txb{AWGN same codebook model: $\eta^*$ vs $E_b/N_0$ for $M=2^{100}$. The solid lines correspond to $\eta^*\equiv \eta_M^*$ from \eqref{eq:MU_eff}. The dashed lines correspond to  $\eta^*_{\infty}$ defined in \eqref{eq:eta_inf}.}}
     \label{fig:6}
      \end{center} 
 \end{figure}

 \begin{figure}[h]
  \begin{center}
     %\resizebox{11cm}{5.2cm} {\includegraphics *[width=\linewidth]{{{isit_ebn0_mu_pe_0.001_k_100_new}}}}
     \includegraphics *[height=5.5cm,width=\columnwidth]{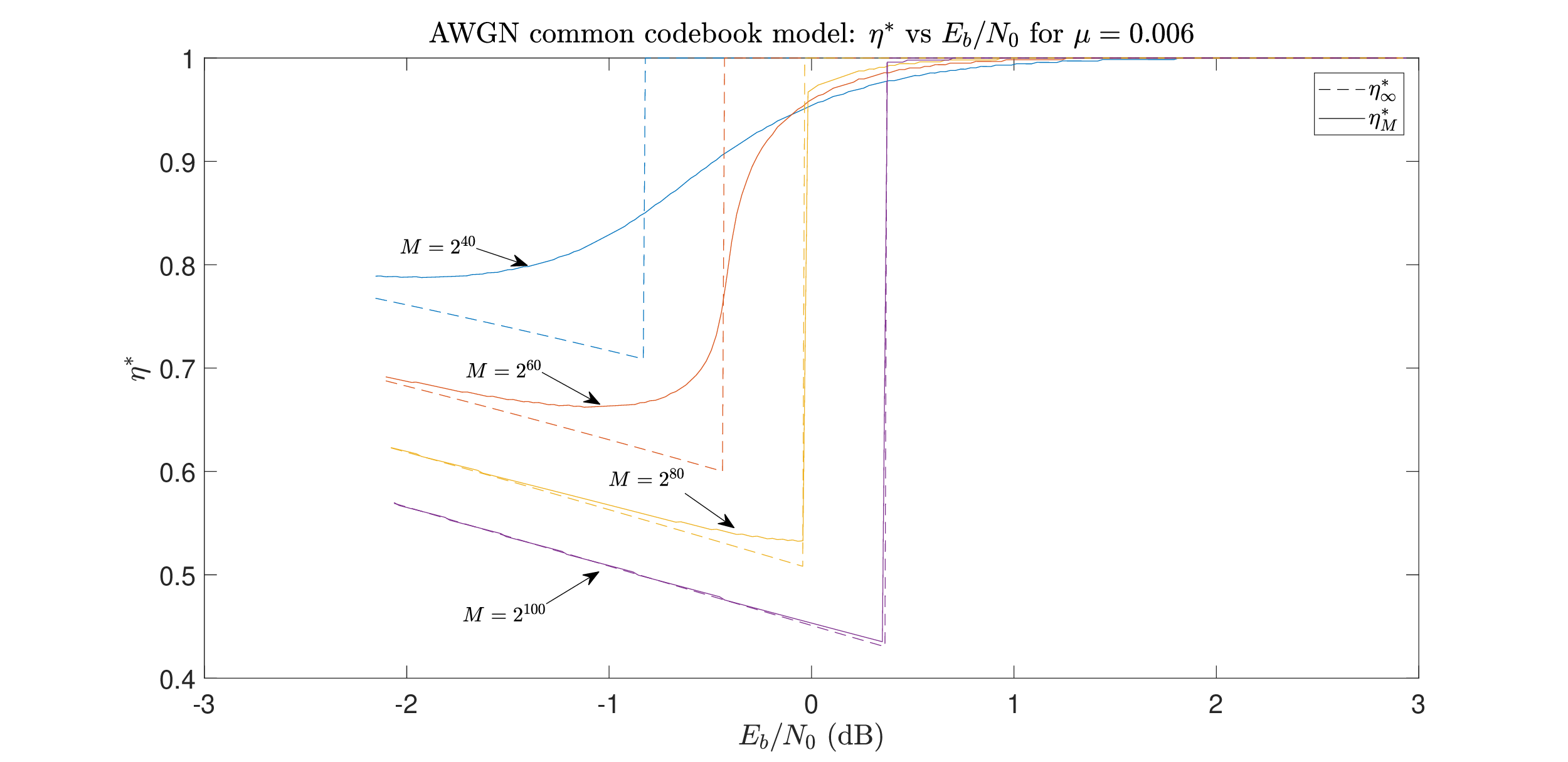}
     \caption[]{\txb{AWGN same codebook model: $\eta^*$ vs $E_b/N_0$ for $\mu=0.006$. The solid lines correspond to $\eta^*\equiv \eta_M^*$ from \eqref{eq:MU_eff}. The dashed lines correspond to  $\eta^*_{\infty}$ defined in \eqref{eq:eta_inf}.}}
     \label{fig:7}
      \end{center} 
 \end{figure}

\subsection{Future work}
There are a lot of interesting directions for future work. A natural extension, already undertaken for the setting of
random access~\cite{fengler2019massive}, would be to analyze the many-user massive
MIMO fading channel with receiver having $N>1$ antennas under different fading scenarios (block-fading and quasi-static
fading would probably be most relevant). Further, different asymptotic scalings of $N$ and $n$ may lead to radically
different tradeoffs. Another interesting direction is to see how much improvements would the (so far incomputable)
vector-AMP (and replica-method) yield over the bounds presented in this work. Note that for practical systems the
asymptotic limit $n\to\infty$ is less relevant than finite-$n$ bounds. However, AMP based bounds are inherently
asymptotic. In this regard, it would be interesting to derive finite-blocklength versions of achievability bounds.
Theorems~\ref{th:scaling_noCSI} and~\ref{th:scaling_CSIR} should be possible to ``de-asymptotize''. Similarly, recent
bounds based on Gaussian processes~\cite{ZPT-isit19} should be extendable to both the quasi-static fading and finite
blocklength. From a practical standpoint, the most pressing issue is finding \textit{any} architecture for the many-MAC
setting that would have an approximately perfect MUI cancellation (i.e. have vertical part in the $\mathcal{E}$ vs $\mu$
tradeoff). One imagines such a system should be comprised of a message-passing decoding alternating between interference
cancellation and signal re-estimation -- as for example proposed in~\cite{caire2001optimal}. 

%}

\pagebreak
\appendices
\section{Proofs of section \ref{sec:classical}}
\label{app:1}
\subsection{Joint error}
\begin{proof}[Proof of theorem \ref{th:cap_joint}]
Let 
$$R=(R_1,...,R_K)\in C_{\epsilon,J}.$$
 We need to show that there exists a sequence of $$\left((M^{(n)}_1,M^{(n)}_2,...,M^{(n)}_K),n,\epsilon_n\right)_J$$ codes with projection decoding, such that
\begin{IEEEeqnarray}{LL}
\liminf\limits_{n\to\infty}\frac{1}{n}\log\left(M^{(n)}_i\right)\geq R_i, \forall i\in[K] \label{eq:1a}\\
\limsup\limits_{n\to\infty} \epsilon_n\leq \epsilon\label{eq:1b}.
\end{IEEEeqnarray}

Let $\eta_i>0,i\in[K]$. Choose $$M_i^{(n)}=\lfloor 2^{n(R_i-\eta_i)}\rfloor, \forall i\in [K].$$ We use random coding: user $j$, independently generates $M_j^{(n)}$ vectors, each independently and uniformly distributed on the $\sqrt{nP}$--complex sphere. That is $X_i\distas{iid}\mathrm{Unif}\left(\sqrt{nP}\cs^{n-1}\right)$. Hence the channel inputs  are given by $$X^{(n)}_i\distas{iid}\left(\sqrt{nP}\cs^{n-1}\right).$$ We will drop the superscript $n$ for brevity.

Suppose the codewords $(c_1,c_2,...,c_K)\in \prod_{i=1}^K\cd_i$ were actually sent. Then by \eqref{eq:dec1}, error occurs iff $\exists (c'_1,c'_2,...,c'_K)\in \cd_1 \times \cd_2 ... \times \cd_K $ such that $(c'_1,c'_2,...,c'_K)\neq (c_1,c_2,...,c_K)$ and 
\begin{equation}
\label{eq:4}
\norm{P_{c'_1,...,c'_K}Y}^2\geq\norm{P_{c_1,...,c_K}Y}^2.
\end{equation}
This can be equivalently written as follows. Let $S\subset [K]$ be such that 
\begin{IEEEeqnarray}{C}
\label{eq:5}
\IEEEyesnumber
%\begin{split}
i\in[S] \iff \hat{c}_i\neq c_i
%\end{split}
\end{IEEEeqnarray}
where $(\hat{c_i})_{i=1}^K$ denote the decoded codewords.

Let $c_{[S]}\equiv \{c_i:i\in[S]\}$. Then, error occurs iff $\exists S\subset [K]$ and $S\neq \emptyset$, and $\exists \{c'_i:i\in [S], c'_i\neq c_i\}$ such that 
\begin{equation}
\label{eq:6}
\norm{\Pd{S}{S^c}Y}^2\geq\norm{\Pc{[K]}Y}^2.
\end{equation}

Let $B_S=\left\{\norm{\Pd{S}{S^c}Y}^2\geq\norm{\Pc{[K]}Y}^2\right\}$ (here primes denote unsent codewords i.e., $c'_i$ here means that it is independent of the channel inputs/output and distributed with the same law as $c_i$). Note that, for the sake of brevity, we are suppressing the dependence on $c'$.

So, the average probability of error is given by

\begin{IEEEeqnarray*}{LLL}
\label{eq:7}
\epsilon_n& =&\Prb{\bigcup_{\substack{ S\subset [K]\\ S\neq \emptyset}}\bigcup_{\substack{\{c'_i\in \cd_i:\\ i\in S, c'_i\neq c_i\}}} B_S}\\
&=&\Prb{\bigcup_{t\in[K]}\bigcup_{ \substack{S\subset [K]\\ |S|=t}}\bigcup_{\substack{\{c'_i\in \cd_i:\\ i\in S, c'_i\neq c_i\}}}B_S}\IEEEyesnumber
\end{IEEEeqnarray*}

Using ideas similar to the Random Coding Union (RCU) bound \citep{polyanskiy2010channel}, we have
\begin{\Ieee}{LLL}
\label{eq:8}
\epsilon_n &\leq & \Exa{\min\biggr\{1,\sum_{t\in[K]}\sum_{S\subset[K]:|S|=t}\left(\prod_{j\in S}(M_j-1)\right)\cdot\\
&&\Prb{B_S\left. \right|c_{[K]},H_{[K]},Z}\biggr\}}\Ieeen
\end{\Ieee}
where $H_{[K]}=\{H_i:i\in[K]\}$. 

From now on we denote 
$$\bigcup_{t\in[K]}\bigcup_{ \substack{S\subset [K]\\ |S|=t}}\equiv \bigcup_{t,S}$$
$$\sum_{t\in[K]}\sum_{S\subset[K]:|S|=t}\equiv \sum_{t,S}$$
$$\bigcap_{t\in[K]}\bigcap_{ \substack{S\subset [K]\\ |S|=t}}\equiv\bigcap_{t,S}.$$

\begin{claim}
\label{claim:1}
For $t\in[K]$ and $S\subset[K]$ with $|S|=t$, 
\begin{\Ieee}{LLL}
\label{eq:9}
\Prb{\norma{\Pd{S}{S^c}Y}^2>\norm{\Pc{[K]}Y}^2\biggr|c_{[K]},H_{[K]},Z}\\
=F\biggr(\frac{\norm{Y}^2-\norm{\Pc{[K]}Y}^2}{\norm{Y}^2-\norm{\Pc{S^c}Y}^2};n-K,t\biggr)\Ieeen
\end{\Ieee}
where $F(x;a,b)$ is the cdf of beta distribution $Beta(a,b)$. Further, from \citep{yang2014quasi}, we have 
\begin{equation}
\label{eq:10}
F(x;n-K,t)\leq (n-K+t-1)^{t-1}x^{n-K}
\end{equation}
\begin{proof}
 See proof of claim \ref{claim:sc_nc1}.
\end{proof}
\end{claim}

Letting $G_S\equiv g(Y,c_{[K]},S)=\frac{\norm{Y}^2-\norm{\Pc{[K]}Y}^2}{\norm{Y}^2-\norm{\Pc{S^c}Y}^2}$, $M_{S}=\prod_{j\in S}(M_j-1)$, $s_t= (t-1)\frac{\ln(n-K+t-1)}{(n-K)}$ and $r_t=s_t + \frac{\ln M_S}{ (n-K) }$, we have the following from \eqref{eq:8}, \eqref{eq:9} and \eqref{eq:10}
%\begin{lemma}
%\label{lem:K1}
%For the $K$-user MAC defined above, with the projection decoder, the average probability of error is upper bounded as
\begin{IEEEeqnarray*}{LLL}
\label{eq:11}
\epsilon_n\leq & \mathbb{E}\left[\min \left\{1,  \sum_{t,S}\exp\left(-(n-K)\left[-r_t-\ln G_S\right]\right)\right\}\right]\\
 \IEEEyesnumber
\end{IEEEeqnarray*}
%\end{lemma}
\vspace{0.5cm}

Let $\delta>0$ and  let $E_1$ be the following event
\begin{IEEEeqnarray*}{LLLL}
E_1 &=& \bigcap_{t,S} & \left\{-\ln G_S -r_t>\delta\right\}\label{eq:12a}\IEEEyesnumber\\
&=&\bigcap_{t,S}&\left\{-\ln G_S >\tilde{V}_{n,S}\right\}\\
&=&\bigcap_{t,S}&\left\{G_S<V_{n,S}\right\}\label{eq:12b}\IEEEyesnumber
\end{IEEEeqnarray*}
where $\tilde{V}_{n,S}=r_t+\delta$ and $V_{n,S}=e^{-\tilde{V}_{n,S}}$. Note that $V_{n,S}$ depends on $S$ and $t$.

Then, from \eqref{eq:11} we have the following

\begin{lemma}
\label{lem:1a}
For the $K$-user MAC defined above, with the projection decoder, the average probability of error is upper bounded as
\begin{IEEEeqnarray*}{LLL}
\label{eq:13}
\epsilon_n %& \leq &\sum_{t,S}e^{-(n-K)\delta} + \mathbb{P}[E_1^c]\\
& \leq &\sum_{t,S}e^{-(n-K)\delta} +\mathbb{P}\left[\bigcup_{t,S} G_S\geq V_{n,S}\right] \IEEEyesnumber
\end{IEEEeqnarray*}
\begin{proof}
By \eqref{eq:11}, 
\begin{IEEEeqnarray*}{LLL}
&&\epsilon_n\\
&\leq & \mathbb{E}\left[\min \left\{1,  \sum_{t,S}e^{\left(-(n-K)\left[-r_t-\ln G_S \right]\right)}\right\}(1[E_1]+1[E_1^c])\right]\\
&\leq & \sum_{t,S}e^{-(n-K)\delta} + \mathbb{P}[E_1^c] \IEEEyesnumber
\end{IEEEeqnarray*}
\end{proof}
\end{lemma}

Hence, as $n\to\infty$, it is the second term in the above expression that potentially dominates.% Next we analyze this second term.

\begin{claim}
\label{claim:2}
For $t\in[K]$, $S\subset[K]$ with $|S|=t$, we have
\begin{IEEEeqnarray*}{LLL}
\label{eq:14}
&&\Pb{G_S\geq V_{n,S}}\\
&\leq & \Prb{\norma{(1-V_{n,S})\Pcp{S^c}Z-V_{n,S}\Pcp{S^c}\hc}^2 \\
&& \geq V_{n,S}\norma{\Pcp{S^c}\hc}^2}\IEEEyesnumber
\end{IEEEeqnarray*}
where $\Pcp{S^c}$ represents the orthogonal projection onto the orthogonal complement of the space spanned by $c_{[S^c]}$.
\begin{proof}
See proof of claim \ref{claim:sc_nc2}
\end{proof}
\end{claim}

To evaluate the above probability, we condition on $c_{[K]}$ and $H_{[K]}$. For ease of notation, we will not explicitly write the conditioning. 

Let $\chi'_2(\lambda,d)$ denote the non-central chi-squared distributed random variable with non-centrality $\lambda$ and degrees of freedom $d$. That is, if $Z_i\distas{}\mathcal{N}(\mu_i,1),i\in[d]$ and $\lambda=\sum_{i\in[d]}\mu_i^2$, then $\chi'_2(\lambda,d)$ has the same distribution as that of $\sum_{i\in[d]}Z_i^2$.

Since $Z\distas{}\cn(0,I_n)$, we have $$Z-\frac{V_{n,S}}{1-V_{n,S}}\hc\distas{}\cn(-\frac{V_{n,S}}{1-V_{n,S}}\hc,I_n).$$ Hence 
\begin{\Ieee}{LLL}
\Pcp{S^c}\left(Z-\frac{V_{n,S}}{1-V_{n,S}}\hc\right)\\
\distas{}\cn(-\frac{V_{n,S}}{1-V_{n,S}}\Pcp{S^c}\hc,\Pcp{S^c}).
\end{\Ieee}
 Now using the fact that if $W=P+iQ \distas{}\cn(\mu,\Gamma,0)$ then 
\begin{equation}
\label{eq:gauss}
\begin{bmatrix} P \\ Q \end{bmatrix} \distas{} \mathcal{N}\left(\begin{bmatrix} Re(\mu) \\ Im(\mu) \end{bmatrix}, \frac{1}{2}\begin{bmatrix}
Re(\Gamma) & -Im(\Gamma)\\ Im(\Gamma) & Re(\Gamma) \end{bmatrix}\right)
\end{equation}
we can show the following
\begin{lemma}
\label{lem:1b}
Let $F=\norm{\frac{V_{n,S}}{1-V_{n,S}}\Pcp{S^c}\hc}^2$ and $n'=n-K+t$. Conditioned on $H_{[K]}$ and $c_{[K]}$, we have
\begin{IEEEeqnarray*}{LLL}
\label{eq:16}
\norm{\Pcp{S^c}\left(Z-\frac{V_{n,S}}{1-V_{n,S}}\hc\right)}^2 \distas{}\frac{1}{2}\chi'_2\left(2F,2n'\right)\IEEEyesnumber
\end{IEEEeqnarray*}
Hence its conditional expectation is
\begin{equation}
\label{eq:17}
\mu=n'+F.
\end{equation}

\begin{proof}
See proof of claim \ref{claim:sc_nc3}.
\end{proof}
\end{lemma}

Let $U=\frac{V_{n,S}}{(1-V_{n,S})}\norm{\Pcp{S^c}\hc}^2-n'$. Hence $F=\frac{V_{n,S}}{1-V_{n,S}}(U+n')$. Note that $V_{n,S}$, $U$ ,$\lambda$ all depend on $t$ and $S$.

Letting $T=\frac{1}{2}\chi'_2(2F,2n')-(F+n')$, we have,
\begin{IEEEeqnarray*}{LLL}
\label{eq:18}
&& \Prb{\bigcup_{t,S}\norma{\Pcp{S^c}\biggr(Z-\frac{V_{n,S}}{1-V_{n,S}}\hc\biggr)}^2 \\
&& \geq \frac{V_{n,S}}{(1-V_{n,S})^2}\norma{\Pcp{S^c}\hc}^2} \\
&=&\Prb{\bigcup_{t,S}\norma{\Pcp{S^c}\biggr(Z-\frac{V_{n,S}}{1-V_{n,S}}\hc\biggr)}^2-\mu\geq U}\\
 &=& \Exa{\Prb{\biggr\{\bigcup_{t,S}\left\{T\geq U\right\}\biggr\}\biggr|c_{[K]},H_{[K]} }}.\IEEEyesnumber
\end{IEEEeqnarray*}

Next we use lemma \ref{lem:chi2} to bound \eqref{eq:18}. 

First, note that 
{\allowdisplaybreaks
\begin{IEEEeqnarray*}{LLL}
\label{eq:21}
U & =\frac{V_{n,S}}{(1-V_{n,S})}\norm{\Pcp{S^c}\hc}^2-n'\\
& =\frac{n'}{1-V_{n,S}}\left(V_{n,S} \left(1+\frac{\norm{\Pcp{S^c}\hc}^2}{n'}\right)-1\right)\\
& =n'U^1 \IEEEyesnumber
\end{IEEEeqnarray*}
}
where $$U^1=\frac{1}{1-V_{n,S}}\left(V_{n,S} W_S-1\right)$$ and $$W_S=\left(1+\frac{\norm{\Pcp{S^c}\hc}^2}{n'}\right).$$ Hence 
\begin{equation}
\label{eq:22}
F=\frac{V_{n,S}}{1-V_{n,S}}(U+n')=n'\frac{V_{n,S}}{1-V_{n,S}}(U^1+1).
\end{equation}

Let $\delta_1>0$. Let $E_{11}=\bigcap_{S,t}\{U^1>\delta_1\}$. From \eqref{eq:18} we have 
{\allowdisplaybreaks
\begin{IEEEeqnarray*}{LLL}
\label{eq:23}
&& \Prb{\bigcup_{t,S}T\geq U}\\
 & \leq & \sum_{t,S}\Exa{\Prb{\left\{T \geq U\right\}\biggr| c_{[K]},H_{[K]} }1[E_{11}]}\\
 &&+\Pb{E_{11}^c} \\
& \leq & \sum_{t,S}\Ex{\Pb{\left\{T\geq U\right\}\biggr|c_{[K]},H_{[K]} }1[U^1>\delta_1]} \\
&&+\Pb{E_{11}^c}\\
&\leq & \sum_{t,S} \mathbb{E}\left[e^{-n'f_n(U^1)}1[U^1>\delta_1]\right]\\
&&+\Pb{E_{11}^c} \IEEEyesnumber 
\end{IEEEeqnarray*}
}
where the last inequality follows from \eqref{eq:sc_nc25a}, and 
\begin{IEEEeqnarray*}{LLL}
\label{eq:24}
&& f_n(x) \\
&=& x+1+\frac{2V_{n,S}}{1-V_{n,S}}(1+x)\\
&& -\sqrt{1+\frac{2V_{n,S}}{1-V_{n,S}}(1+x)}\sqrt{2x+1+\frac{2V_{n,S}}{1-V_{n,S}}(1+x)}\\
\IEEEyesnumber 
\end{IEEEeqnarray*}

Now, from claim \ref{claim:sc_nc4}, we have that for $0<V_{n,S}<1$ and $x>0$, $f_n(x)$ is a monotonically increasing function of $x$

Hence we have
\begin{IEEEeqnarray*}{LLL}
\label{eq:25}
\Prb{\bigcup_{t,S}\left\{T\geq U\right\}} \leq \sum_{t,S}e^{-n' f_n(\delta_1)}+\Pb{E_{11}^c}. \IEEEyesnumber
\end{IEEEeqnarray*}

So, we have the following claim
\begin{claim}
\label{claim:4}
Let $A_S=\left\{V_{n,S} W_S-1 \leq \delta_1\right\}$ and $E_{12}=\bigcup_{t,S} A_S $. If $0<V_{n,S}<1$ for all $t\in[K]$, $S\subset[K]$ with $|S|=t$ then we have
\begin{IEEEeqnarray*}{LLL}
\label{eq:26a}
& \Prb{\bigcup_{t,S}\left\{G_S\geq V_{n,S}\right\}}  \leq \sum_{t,S}e^{-n'f_n(\delta_1)}+ \Pb{E_{12}}.\\ \IEEEyesnumber
\end{IEEEeqnarray*}
\begin{proof}
\begin{IEEEeqnarray*}{LLL}
\label{eq:26b}
 \Pb{\bigcup_{t,S}\left\{G_S\geq V_{n,S}\right\}}  &\leq &\sum_{t,S} e^{-n'f_n(\delta_1)}+ \mathbb{P}\left[E_{11}^c\right]\\
& \leq &\sum_{t,S}e^{-n'f_n(\delta_1)}+ \Pb{E_{12}}.\\ \IEEEyesnumber
\end{IEEEeqnarray*}
\end{proof}
\end{claim}

Now, we need to upper bound $\Pb{E_{12}}$.
 
 We have
\begin{IEEEeqnarray*}{LLL}
\label{eq:27}
\norma{\Pcp{S^c}\hc}^2=\sum_{i\in S}|H_i|^2\norm{\Pcp{S^c}c_i}^2\\
+ 2\sum_{i<j:i,j\in S}Re\left(\dop{\Pcp{S^c}c_i}{\Pcp{S^c}c_j}H_i\bar{H_j}\right).\IEEEyesnumber
\end{IEEEeqnarray*}
Further,
\begin{IEEEeqnarray}{LL}
\label{eq:28}
&\dop{\Pcp{S^c}c_i}{\Pcp{S^c}c_j}=\dop{c_i}{c_j}-\dop{\Pc{S^c}c_i}{\Pc{S^c}c_j}
%& \implies Re\dop{\Pcp{S^c}c_i}{\Pcp{S^c}c_j}=Re\dop{c_i}{c_j}-Re\dop{c_i}{\Pcp{S^c}c_j}.%
\end{IEEEeqnarray}
Hence we have
{\allowdisplaybreaks
\begin{IEEEeqnarray*}{LLL}
\label{eq:29}
&&\left|Re\dop{\Pcp{S^c}c_i}{\Pcp{S^c}c_j}\right| \\
&\leq & \left|\dop{\Pcp{S^c}c_i}{\Pcp{S^c}c_j}\right|\\
&\leq & \left|\dop{c_i}{c_j}\right|+\left|\dop{\Pc{S^c}c_i}{\Pc{S^c}c_j}\right| \\
& \leq & \left|\dop{c_i}{c_j}\right|+\norm{\Pc{S^c}c_j}\norm{\Pc{S^c}c_i}\\
&=& nP\left(\left|\dop{\hat{c}_i}{\hat{c}_j}\right| +\norm{\Pc{S^c}\hat{c}_i}\norm{\Pc{S^c}\hat{c}_j}\right)\IEEEyesnumber \IEEEeqnarraynumspace
\end{IEEEeqnarray*}
}
where hats denote corresponding normalized vectors. Since these unit vectors are high dimensional, their dot products and projection onto a smaller, fixed dimension surface is very small. Indeed, we have the following two lemmas.

\begin{lemma}
\label{lem:2}
If $e_1,e_2\distas{iid} Unif(\cs^{n-1})$, then for any $\delta_2>0$, we have
\begin{equation}
\label{eq:30}
\Pb{|\dop{e_1}{e_2}|>\delta_2}\leq 4e^{-\frac{n\delta_2^2}{2}}
\end{equation}
\begin{proof}
First, lets take $e_1,e_2\distas{iid}S^{n-1}$. Let $x$ be a fixed unit vector in $\mathbb{R}^n$. Due to symmetry, we have $\Pb{\dop{e_1}{x}\geq 0}=1/2$. Hence, by Levy's Isoperimetric inequality on the sphere \citep{ledoux2001concentration}, we have
\begin{equation}
\label{eq:31}
\Pb{\dop{e_1}{x}>\delta_2}\leq e^{-n\delta_2^2/2}.
\end{equation}
Again by symmetry, and then taking $x$ as $e_2$, we have
\begin{equation}
\label{eq:32}
\Pb{\left|\dop{e_1}{e_2}\right|>\delta_2}\leq 2e^{-n\delta_2^2/2}.
\end{equation}

Now uniform distribution on $\cs^{n-1}$ is same as the uniform distribution on $S^{2n-1}$, and for complex vectors $z_1=x_1+iy_1$ and $z_2=x_2+iy_2$ we have $Re\dop{z_1}{z_2}=x_1^T x_2+y_1^T y_2=(x_1,y_1)^T (x_2,y_2)$. Hence if $e_1,e_2\distas{iid}\cs^{n-1}$, and $u_1,u_2\distas{iid}S^{2n-1}$ then $Re\dop{e_1}{e_2}$ has same law as $\dop{u_1}{u_2}$. Hence we have
\begin{equation}
\label{eq:33}
\Pb{|Re\dop{e_1}{e_2}|>\delta_2}\leq 2e^{-\frac{2n\delta_2^2}{2}}.
\end{equation}

Also, $Im\dop{z_1}{z_2}=x_1^T y_2-y_1^T x_2$. Hence $Im\dop{e_1}{e_2}$ has the same law as $Re\dop{e_1}{e_2}$. Hence we have 
{\allowdisplaybreaks
\begin{IEEEeqnarray*}{LLL}
\label{eq:34}
&\Pb{|\dop{e_1}{e_2}|>\delta_2} \\
&=\Pb{|\dop{e_1}{e_2}|^2>\delta_2^2} \\
& =\Pb{|Re\dop{e_1}{e_2}|^2+|Im\dop{e_1}{e_2}|^2>\delta_2^2}\\
&\leq \Pb{|Re\dop{e_1}{e_2}|>\frac{\delta_2}{\sqrt{2}}}+\Pb{|Im\dop{e_1}{e_2}|>\frac{\delta_2}{\sqrt{2}}}\\
& \leq 4e^{-\frac{n\delta_2^2}{2}}.\IEEEyesnumber
\end{IEEEeqnarray*}
}

\end{proof}
\end{lemma}

Next we have a similar lemma for low dimensional projections from \citep[Lemma 5.3.2]{vershynin2018high}
\begin{lemma}[\citep{vershynin2018high}]
\label{lem:3}
Let $x\distas{}Unif(S^{n-1})$ and $P$ be a projection to an $m$ dimensional subspace of $\mathbb{R}^n$. Then for any $\delta_3>0$, we have
\begin{equation}
\label{eq:35}
\Pb{\left|\norm{Px}-\sqrt{\frac{m}{n}}\right|>\delta_3}\leq 2e^{-cn\delta_3^2}
\end{equation}
where $c$ is some absolute constant. Hence, by symmetry, the result remains true if $P$ is a uniform random projection, independent of $x$.
\end{lemma}

Now we need to prove that a similar result holds for the complex variable case as well. We have the following lemma %(I am not fully sure if its correct)
\begin{lemma}
\label{lem:4}
Let $z\distas{}Unif{\cs^{n-1}}$ and $P$ be a projection to an $m$ dimensional subspace $V$ of $\mathbb{C}^n$. Then for any $\delta_3>0$, we have
\begin{equation}
\label{eq:36}
\Pb{\left|\norm{Pz}-\sqrt{\frac{m}{n}}\right|>\delta_3}\leq 2e^{-2cn\delta_3^2}
\end{equation}
where $c$ is some absolute constant. Hence, by symmetry, the result remains true if $P$ is a uniform random projection, independent of $z$.
\begin{proof}
Consider $\norm{Pz}$. Let $U$ be the unitary change of basis matrix which converts $V$ to first $m$ coordinates. Hence $\norm{Pz}=\norm{UPz}$. Therefore we can just consider the orthogonal projection onto first $m$ coordinates. Hence the projection matrix $P$ is real. Let $e_1,...,e_m$ be the standard basis corresponding to the first $m$ coordinates. Let $A$ be the $n\times m$ matrix whose columns are $e_1,...,e_m$. Then $P=AA^*$ ($*$ denotes conjugate transpose). Since $A$ is real, we have $Re(Pz)=AA^* Re(z)$ and $Im(Pz)=AA^*Im(z)$. Now, if $z\distas{}Unif(\cs^{n-1})$ then $Re(z)$ has same law as $Im(z)$. Hence $Re(Pz)$ has same law as $Im(Pz)$. Further $A^*=A^T$. Also note that, if $z=x+iy$ then $\norm{Pz}^2=z^* AA^* Z=x^T A A^T x+y^T AA^Ty= \begin{bmatrix} x^T & y^T \end{bmatrix} \begin{bmatrix} AA^T & 0 \\ 0 & AA^T \end{bmatrix} \begin{bmatrix} x \\ y \end{bmatrix}= \norm{\hat{P}\begin{bmatrix} x\\y \end{bmatrix}} $ where $\hat{P}$ denotes the orthogonal projection from $\mathbb{R}^{2n}$ to a $2m$ dimensional subspace. Hence $\norm{Pz}^2$ has the same law as that of the projection of a uniform random vector on $S^{2n-1}$ to a $2m$ dimensional subspace. Hence using lemma \ref{lem:3}, we have
\begin{equation}
\label{eq:37}
\Pb{\left|\norm{Pz}-\sqrt{\frac{m}{n}}\right|>\delta_3}\leq 2e^{-2cn\delta_3^2}
\end{equation}
\end{proof}
\end{lemma}

Since $H_i\distas{}\cn(0,1)$, we have $|H_i|^2\distas{}\frac{1}{2}\chi_2(2)=\text{exp}(1)$ where $\chi_2(d)$ denotes the chi-squared distribution with $d$ degrees of freedom and $\text{exp}(1)$ represents an exponentially distributed random variable with rate $1$. Therefore, for $\nu\geq 0$,
%\begin{lemma}
%\label{lem:5}
%For $\nu\geq 0$,
\begin{equation}
\label{eq:38}
%\Pb{|H_i|^2\geq \nu}=\Pb{\frac{1}{2}\chi_2(2)\geq \nu}\leq e^{-(\nu-\sqrt{2\nu-1})}
\Pb{|H_i|^2\geq \nu}=e^{-\nu}
\end{equation}

%\end{lemma}
\iffalse
Also for $i\neq j$ 
\begin{IEEEeqnarray*}{LLL}
\label{eq:39}
\Pb{|H_i H_j|>\nu} \\
\leq  \Pb{|H_i|^2>\nu}+\Pb{|H_j|^2>\nu}=2\Pb{|H_i|^2>\nu}\IEEEyesnumber
\end{IEEEeqnarray*}
\fi
%\end{lemma}

Now, we are in a position to bound $\Pb{E_{12}}$.

 For $S\subset[K]$ with $|S|=t$, define the events $E_2$, $E_3$ and $E_4$ as follows:
\begin{IEEEeqnarray*}{LLL}
\label{eq:40}
& E_2=\bigcap_{i<j:i,j\in [K]}\left\{|\dop{\hat{c}_i}{\hat{c}_j}|\leq \delta_2\right\}\IEEEyessubnumber\\
& E_3(S,t)=\bigcap_{i\in S}\left\{\left|\norm{\Pc{S^c}\hat{c}_i}-\sqrt{\frac{K-t}{n}}\right|\leq\delta_3\right\}\IEEEyesnumber \IEEEyessubnumber\\
& E_4=\bigcap_{i\in [K]}\left\{|H_i|^2\leq \nu\right\}\IEEEyessubnumber
\end{IEEEeqnarray*}
where we choose $\delta_2=n^{-\frac{1}{3}}=\delta_3$ and $\nu=n^{\frac{1}{4}}$. Hence we have

\begin{IEEEeqnarray*}{LLL}
\label{eq:41}
\Pb{E_{12}} &\leq &   \Pb{\bigcup_{t,S}\left( A_S\cap E_2\cap E_3\cap E_4\right)} \\
&&+ \Pb{E^c_2}+\Pb{ E^c_4}+\sum_{t,S} \Pb{ E^c_3(S,t)}.\IEEEyesnumber
\end{IEEEeqnarray*}

Using lemmas \ref{lem:2} and \ref{lem:4} and eq. \eqref{eq:38}, we have 
\begin{IEEEeqnarray*}{LLL}
\label{eq:42}
\Pb{E^c_2}+\Pb{ E^c_4}+\sum_{t,S} \Pb{ E^c_3(S,t)}\\
\leq 2K(K-1)e^{-\frac{n\delta_2^2}{2}} +Ke^{-\nu}+\sum_{t,S}2te^{-cn\delta_3^2}.\IEEEyesnumber
\end{IEEEeqnarray*}
Note that the above quantity goes to $0$ as $n\to \infty$ due to the choice of $\delta_2$, $\delta_3$ and $\nu$. Also, the choice of parameters is not the optimum. Nevertheless, this is enough to prove the result.

Let $$\delta_{1,t}=\left(\delta_2+\left(\delta_3+\sqrt{\frac{K-t}{n}}\right)^2\right)$$ $$\delta_{2,t}=\left(\delta_3+\sqrt{\frac{K-t}{n}}\right)^2$$

 Observe that on the sets $E_2$, $E_3$ and $E_4$, we have from \eqref{eq:29}
\begin{IEEEeqnarray*}{LLL}
\label{eq:43}
 \left|Re\dop{\Pcp{S^c}c_i}{\Pcp{S^c}c_j} H_i\bar{H}_j\right|\leq  \nu\delta_{1,t} =  O(n^{-\frac{1}{12}}) \IEEEeqnarraynumspace\IEEEyesnumber \IEEEyessubnumber\\
 |H_i|^2\norm{\Pc{S^c}\hat{c}_i}^2\leq \nu \delta_{2,t}=O(n^{-\frac{5}{12}})\IEEEyessubnumber \IEEEeqnarraynumspace
\end{IEEEeqnarray*}

 So we have
{\allowdisplaybreaks
\begin{IEEEeqnarray*}{LLL}
\label{eq:44}
&& \Prb{\bigcup_{t,S}\left(A_S\cap E_2 \cap E_3 \cap E_4\right)} \\
 &\leq & \Prb{\bigcup_{t,S}\biggr\{V_{n,S}\biggr[1+\frac{nP}{n'}\biggr\{\sum_{i\in S}|H_i|^2\norm{\hat{c}_i}^2 \\
 &&-t\nu\delta_{2,t}  -t(t-1)\delta_{1,t} \biggr\}\biggr]-1\leq \delta_1\biggr\}}\\
&\leq & \Prb{\bigcup_{t,S}\biggr\{V_{n,S}\left[1+\frac{nP}{n'}\sum_{i\in S}|H_i|^2\right] \\
&&\leq 1+\delta_1+O(n^{-\frac{1}{12}})\biggr\}}\\
& \leq &\Pb{\bigcup_{t,S}\left\{\ln\left[1+P\sum_{i\in S}|H_i|^2\right]\leq V'_{n,S}\right\}} \IEEEyesnumber
\end{IEEEeqnarray*}
}
where $V'_{n,S}=\tilde{V}_{n,S}+\ln(1+\delta_1+O(n^{-1/12}))$, and $O$ depends on $K$ and $t$.

Let $\delta_n=\ln(1+\delta_1+O(n^{-1/12}))$. We have $\frac{\log M_S }{n-K}=\left(\sum_{i\in S}(R_i-\eta_i)\right) (1+o(1))$ and $s_t=O\left(\frac{\log n}{n}\right)$. 

By the choice of $M^{(n)}_i$, for sufficiently large $n $, sufficiently small $\delta$ and $\delta_1$, we have

{\allowdisplaybreaks
\begin{IEEEeqnarray*}{LLL}
\label{eq:46}
&&\Prb{\bigcup_{t,S} \biggr\{\ln\left[1+P\sum_{i\in S}|H_i|^2\right]\leq V'_{n,S}\biggr\}} \\
 &= & \Prb{\bigcup_{t,S}\ln\biggr[1+P\sum_{i\in S}|H_i|^2\biggr] \leq s_t+\frac{\ln M_S }{n-K}+\delta+\delta_n}\\
& =&\Prb{\bigcup_{t,S}\biggr\{\log\biggr[1+P\sum_{i\in S}|H_i|^2\biggr]\leq s_t\log_2 (e)+\\
&&\biggr(\sum_{i\in S}(R_i-\eta_i)\biggr) (1+o(1))+(\delta+\delta_n)\log_2 (e)\biggr\}}\\
&\leq &\Prb{\bigcup_{t,S}\biggr\{\log\biggr[1+P\sum_{i\in S}|H_i|^2\biggr]\leq \biggr(\sum_{i\in S}R_i\biggr) \biggr\}}\IEEEyesnumber
\end{IEEEeqnarray*}
}

Finally combining everything, we have
{\allowdisplaybreaks
\begin{IEEEeqnarray*}{LLL}
\label{eq:47}
&&\epsilon_n \\
&\leq &\Prb{\bigcup_{t,S}\biggr\{\log\biggr[1+P\sum_{i\in S}|H_i|^2\biggr]\leq s_t\log_2 (e)\\
&&+\biggr(\sum_{i\in S}(R_i-\eta_i)\biggr) (1+o(1))+(\delta+\delta_n)\log_2 (e)\biggr\}}\\
&& +2K(K-1)e^{-\frac{n^{1/3}}{2}} +Ke^{-n^{1/4}}\\
&&+ \sum_{t,S}\biggr[2te^{-cn^{1/3}}+e^{-\delta(n-K)}+e^{-nf_n(\delta_1)}\biggr].\IEEEyesnumber
\end{IEEEeqnarray*}
}
Therefore for this choice of $\left(M^{(n)}_i\right)$, from \eqref{eq:46} we have
{\allowdisplaybreaks
\begin{IEEEeqnarray*}{LLL}
\label{eq:48}
\limsup\limits_{n\to\infty}\epsilon_n \\ 
\leq \Pb{\bigcup_{t,S}\left\{\log\left[1+P\sum_{i\in S}|H_i|^2\right]\leq \left(\sum_{i\in S}R_i\right)\right\} }\\
\leq \epsilon \IEEEyesnumber
\end{IEEEeqnarray*}
}

Since $\eta_i>0$ were arbitrary, we are done. That is \eqref{eq:1b} is also satisfied.

\end{proof}
\subsection{Per-user error}
\label{app:per_user}
\begin{proof}[Proof of theorem \ref{th:pu1}]
We need to show that there exists a sequence of $\left((M_1^{(n)},M_2^{(n)},...,M_K^{(n)}),n,\epsilon_n\right)_{PU}$ codes with the decoder given by \eqref{eq:pu1} and \eqref{eq:pu2} such that 
\begin{IEEEeqnarray}{LL}
\liminf_{n\to\infty}\frac{1}{n}\log\left(M_{i}^{(n)}\right)\geq R_i,\forall i\in [K]\label{eq:pu4a}\\
\limsup_{n\to\infty}\epsilon_n\leq \epsilon\label{eq:pu4b}.
\end{IEEEeqnarray}

Let $P_e^S (R)<\epsilon$ and $\eta_i>0,i\in[K]$. Choose $M_i^{(n)}=\lceil e^{n(R_i-\eta_i)}\rceil, \forall i\in [K]$. We use random coding with Gaussian codebooks: user $i$ generates $M_i$ codewords $\{c_{j}^{i}:j\in[M_i]\}\distas{iid}\cn(0,P'_n I_n)$ independent of other users, where $P'_n=\frac{P}{1+n^{-\frac{1}{3}}}$. Here $\cn(\mu,\Sigma)$ denotes the complex normal distribution with mean $\mu$, covariance $\Sigma$ and pseudo-covariance $0$. For the (random) message $W_{i}\in[M_i]$, user $i$ transmits $X_i=c^{i}_{W_{i}} 1\{\norm{c^{i}_{W_{i}}}^2>nP\}$. The channel model is given in \eqref{eq:sys1} and the decoder is given by \eqref{eq:pu1} and \eqref{eq:pu2}. The per-user probability of error is given by \eqref{eq:def2}
\begin{equation}
\label{eq:pu5}
P_e= \Ex{\frac{1}{K}\sum_{j=1}^{K}1\left\{W_j\neq \left(g_{D}(Y)\right)_j\right\}}.
\end{equation}

Similar to the proof of \citep[Theorem 1]{polyanskiy2017perspective}, we change the measure over which $\mathbb{E}$ is taken in \eqref{eq:pu5} to the one where $X_i=c^{i}_{W_{i}}$ at the cost of adding a total variation distance. Hence the probability of error under this change of measure becomes
\begin{IEEEeqnarray*}{LLL}
\label{eq:pu6}
P_e\leq & p_1 +p_0
\end{IEEEeqnarray*}

with 
\begin{IEEEeqnarray}{LLL}
p_0=K\Pb{\norm{w}^2>n\frac{P}{P'_n}}\label{eq:pu7}\\
p_1=\Ex{\frac{1}{K}\sum_{j=1}^{K}1\left\{W_j\neq \left(g_{D}(Y)\right)_j\right\}}\label{eq:pu8}
\end{IEEEeqnarray}
where $w\distas{}\cn(0,I_n)$ and, with abuse of notation, $\mathbb{E}$ in $p_1$ is taken over the new measure. It can be easily seen that by the choice of $P'_n$ and lemma \ref{lem:chi2}, $p_0\to 0$ as $n\to\infty$. From now on, we exclusively focus on bounding $p_1$.

$p_1$ can also be written as 
\begin{IEEEeqnarray*}{LL}
\label{eq:pu9}
p_1& =\frac{1}{K}\Ex{\sum_{i\in D}1\left\{W_j\neq \left(g_{D}(Y)\right)_j\right\}+|D^c|}\\
&=1-\frac{\Ex{D}}{K}+\frac{1}{K}\Ex{\sum_{i\in D}1\left\{W_j\neq \left(g_{D}(Y)\right)_j\right\}}\IEEEyesnumber
\end{IEEEeqnarray*}
because, for $i\in D^c$, $1\left\{W_j\neq \left(g_{D}(Y)\right)_j\right\}=1,\ a.s$. Define $p_2$ as

\begin{equation}
\label{eq:pu10}
p_2=\Pb{\sum_{i\in D}1\left\{W_j\neq \left(g_{D}(Y)\right)_j\right\}>0}.
\end{equation}
So, it's enough to show that $p_2\to 0$ as $n\to\infty$. This is because, if $p_2\to 0$, then the non-negative random variables $A_n=\sum_{i\in D}1\left\{W_j\neq \left(g_{D}(Y)\right)_j\right\}$ converge to $0$ in probability. Since $A_n\leq K,\ a.s$, we have, by dominated convergence, $\Ex{A_n}=\Ex{\sum_{i\in D}1\left\{W_j\neq \left(g_{D}(Y)\right)_j\right\}}\to 0 $. To this end, we upper bound $p_2$.

Let $c=(c_1\in\mathcal{C}_1,...,c_K\in\mathcal{C}_K)$ be the tuple of sent codewords. Let $K_1=|D|$. Let $c_{(D)}$ denote the ordered tuple corresponding to indices in $D$. That is, if $i_1<i_2<...<i_{K_1}$ are the elements of $D$, then $(c_{(D)})_{j}=c_{i_j},\forall j\in[K_1]$. Let $B_S=\left\{ \norm{\Pd{S}{S^c}Y}^2>\norm{\Pc{D}Y}^2 \right\}$. Then $p_2$ can also be written as
\begin{IEEEeqnarray}{LLL}
\label{eq:pu11}
p_2 &=&\Prb{\sum_{i\in D}1\left\{W_j\neq \left(g_{D}(Y)\right)_j\right\}>0}\\
&=&\Pb{\exists S\subset D, S\neq \emptyset: \forall i\in S, (g_D(Y))_i\neq W_i}\\
& =& \Prb{\exists c'_{(D)}\neq c_{(D)}: \norm{P_{c'_{[D]}}Y}^2>\norm{\Pc{D}Y}^2}\\
%&= &\Pb{\vphantom{\norm{\Pd{S}{S^c}Y}^2>\norm{\Pc{D}Y}^2}\exists S\subset D, S\neq \emptyset s.t \forall i\in S, \exists c'_i\neq c_i:\right.\\
%&& \left. \norm{\Pd{S}{S^c}Y}^2>\norm{\Pc{D}Y}^2 } \\
&=& \Prb{\bigcup _{t\in[K_1]}\bigcup_{\substack{S\subset D\\ |S|=t}} \bigcup_{\substack{c'_i\in \mathcal{C}_i\setminus\{c_i\}\\ i\in S}} B_S}.
\end{IEEEeqnarray}

Let $\delta>0$, $G_S\equiv g(Y,c_{[K]},S,D)=\frac{\norm{Y}^2-\norm{\Pc{D}Y}^2}{\norm{Y}^2-\norm{\Pc{S^c}Y}^2}$, $M_{S}=\prod_{j\in S}(M_j-1)$, $s_t=(t-1)\frac{\ln(n-K_1+t-1)}{n-K_1}$, $r_t=s_t+\frac{\ln M_S }{n-K_1}$, $\tilde{V}_{n,S}=r_t+\delta$ and $V_{n,S}=e^{-\tilde{V}_{n,S}}$. Denote $\bigcup_{t\in[K_1]}\bigcup_{\substack{S\subset D\\ |S|=t}}$ as $\bigcup_{t,S,K_1}$, similarly for $\bigcap$ and $\sum$. Further, denote $\bigcup_{t\in[K]}\bigcup_{\substack{S\subset [K]\\ |S|=t}}$ as $\bigcup_{t,S}$, again, similarly for $\bigcap$ and $\sum$.

 Note that, since $D$ is random, both $M_S$ and $V_{n,S}$ are random. But in the symmetric case only $M_S$ is not random. Now, following steps similar to \eqref{eq:8}, \eqref{eq:9}, \eqref{eq:11} and \eqref{eq:13}, we have
\begin{IEEEeqnarray*}{LLL}
p_2 & \leq & \Exa{\sum_{t,S,K_1}e^{-(n-K_1)\delta} }+\Prb{\bigcup_{t,S,K_1} G_S\geq V_{n,S}} \label{eq:pu12a}\IEEEyesnumber\\
&\leq &\sum_{t,K}e^{-(n-K)\delta} +\Prb{\bigcup_{t,S,K_1} G_S\geq V_{n,S}}. \label{eq:pu12b}\IEEEyesnumber\\
\end{IEEEeqnarray*}
So, the first term goes to $0$ as $n\to\infty$.

Let $Z_D=Z+\sum_{i\in D^c}H_i c_i$. It can be easily seen that, similar to \eqref{eq:14}, we have
\begin{IEEEeqnarray*}{LLL}
\label{eq:pu13}
&&\Prb{G_S\geq V_{n,S}} \\
&\leq & \Prb{\norma{(1-V_{n,S})\Pcp{S^c}Z_D-V_{n,S}\Pcp{S^c}\hc}^2 \\
&&\geq V_{n,S}\norma{\Pcp{S^c}\hc}^2}.\IEEEyesnumber
\end{IEEEeqnarray*}
 
 Now, conditional of $H_{[K]}$ and $c_{[D]}$, $Z_D\distas{}\cn(0,(1+P'_n\sum_{i\in D^c}|H_i|^2))$. Hence
 \begin{\Ieee}{LLL}
 \Pcp{S^c}\biggr(Z_D-\frac{V_{n,S}}{1-V_{n,S}}\hc\biggr) \distas{}\\
\cn\biggr(-\frac{V_{n,S}}{1-V_{n,S}}\Pcp{S^c}\hc,(1+P'_n\sum_{i\in D^c}|H_i|^2)\Pcp{S^c}\biggr)
 \end{\Ieee}
  Therefore
 \begin{IEEEeqnarray*}{LLL}
\label{eq:pu14}
\norm{\Pcp{S^c}\left(Z_D-\frac{V_{n,S}}{1-V_{n,S}}\hc\right)}^2 \\
\distas{}\phd\frac{1}{2}\chi'_2\left(2F,2 n'\right)\IEEEyesnumber
\end{IEEEeqnarray*}
where 
\begin{IEEEeqnarray}{LLL}
F=\frac{\norm{\frac{V_{n,S}}{1-V_{n,S}}\Pcp{S^c}\hc}^2}{\phd}\label{eq:pu15a}\\
n'=n-K_1+t\label{eq:15b}.
\end{IEEEeqnarray}

Let 
\begin{IEEEeqnarray}{LLL}
U=\frac{V_{n,S}}{(1-V_{n,S})}\frac{\norm{\Pcp{S^c}\hc}^2}{\phd}-n' \label{eq:pu16a}\\
U^1=\frac{1}{1-V_{n,S}}\left(V_{n,S}W_S-1\right)\label{eq:pu16b}\IEEEeqnarraynumspace
\end{IEEEeqnarray}
where $W_S= \left(1+\frac{\norm{\Pcp{S^c}\hc}^2}{n'\phd}\right)$

Hence $U=n'U^1$ and $F=\frac{V_{n,S}}{1-V_{n,S}}n'(1+U^1)$. So, similar to \eqref{eq:18}, we have
\begin{IEEEeqnarray*}{LLL}
\label{eq:pu17}
& \Pb{\bigcup_{t,S,K_1} \left\{G_S\geq V_{n,S}\right\}}\leq \Pb{\bigcup_{t,S,K_1}\left\{T\geq U\right\}}\IEEEyesnumber
\end{IEEEeqnarray*}
where $T=\frac{1}{2}\chi'_2(2F,2n')-(F+n')$.

Let $\delta_1>0$ and $E_{11}=\bigcap_{t,S,K_1}\{U^1>\delta_1\} \in \sigma (H_{[K]},c_{[D]})$. 

Now, similar to \eqref{eq:25}, we have

\begin{IEEEeqnarray*}{LLL}
\label{eq:pu18}
\Prb{\bigcup_{t,S,K_1}\left\{ T\geq U\right\}}\\
 \leq \Exa{\sum_{t,S,K_1}e^{-n' f_n(\delta_1)}}+\Pb{E_{11}^c}. \IEEEyesnumber
\end{IEEEeqnarray*}
where $f_n$ (now a random function) was defined in \eqref{eq:24}. So, again by claim \ref{claim:sc_nc4} and dominated convergence, the first term in \eqref{eq:pu18} converges to $0$ as $n\to \infty$. Next, we upper bound the second term $\Pb{E^c_{11}}$.

Let $A_S=\left\{V_{n,S} W_S-1 \leq \delta_1\right\}$ and $E_{12}=\bigcup_{t,S,K_1}A_S$. Similar to \eqref{eq:26b}, we have

\begin{IEEEeqnarray*}{LLL}
\label{eq:pu19}
\Pb{E^c_{11}}=\Pb{\bigcup_{t,S,K_1}\{U^1\leq \delta_1\}}\leq \Pb{E_{12}}.\\
\IEEEyesnumber \IEEEeqnarraynumspace
\end{IEEEeqnarray*}

Let $\hat{c}_i=c_i/\norm{c_i}$. Let $\delta_2>0$, $\delta_3>0$, $\delta_4>0$ and $\nu>1$. Define the events 
\begin{IEEEeqnarray*}{LLL}
\label{eq:pu20}
& E_2=\bigcap_{i<j:i,j\in [K]}\left\{|\dop{\hat{c}_i}{\hat{c}_j}|\leq \delta_2\right\} \IEEEyesnumber \IEEEyessubnumber\\
& E_3(S,t)=\bigcap_{i\in S}\left\{\left|\norm{\Pc{S^c}\hat{c}_i}-\sqrt{\frac{K_1-t}{n}}\right|\leq\delta_3\right\} \IEEEyessubnumber\\
& E_4=\bigcap_{i\in [K]}\left\{|H_i|^2\leq \nu\right\}\IEEEyessubnumber\\
& E_5=\bigcap_{i\in [K]}\left\{|\norm{c_i}-\sqrt{nP'_n}\leq \delta_4 \sqrt{nP'}|\right\}\IEEEyessubnumber
\end{IEEEeqnarray*}
and choose $\delta_2=O(n^{-\frac{1}{3}})=\delta_3=\delta_4$ and $\nu=O(n^{1/4})$.

Using these events we can bound $\Pb{E^c_{11}}$ as
\begin{IEEEeqnarray*}{LLL}
\label{eq:pu21}
\Pb{E^c_{11}} &\leq &  \Prb{\bigcup_{t,S,K_1}\biggr( A_S\cap E_2\cap E_3(S,t)\cap E_4\cap E_5\biggr)} \\
&&+\Pb{ E^c_2}+\Pb{ E^c_4}+\Pb{E^c_5}\\
&&+ \Exa{\sum_{t,S,K_1} \Pb{E^c_3(S,t)|H_{[K]}}}.\IEEEyesnumber
\end{IEEEeqnarray*}

From \citep[Theorem 3.1.1]{vershynin2018high}, we have 
\begin{IEEEeqnarray}{LL}
\label{eq:pu22}
\Pb{E_5^c}\leq 2Ke^{-c_1 n\delta_4^2}
\end{IEEEeqnarray}
for some constant $c_1>0$.
So, from lemma \ref{lem:2}, lemma \ref{lem:4}, \eqref{eq:38} and \eqref{eq:pu22}, we have

\begin{IEEEeqnarray*}{LLL}
\label{eq:pu23}
\Pb{E^c_{11}} &\leq &  \Pb{\bigcup_{t,S,K_1}\left( A_S\cap E_2\cap E_3(S,t)\cap E_4\cap E_5\right)}  + \\
&& 2K(K-1)e^{-\frac{n\delta_2^2}{2}}+ +Ke^{-\nu}+2Ke^{-c_1n\delta_4^2}\\
&&+ \sum_{t,S}2te^{-cn\delta_3^2}.\IEEEyesnumber
\end{IEEEeqnarray*}

So, by the chose of $\delta_i, i\in\{2,3,4\}$ and $\nu$, the exponential terms in the last expression go to $0$ as $n\to \infty$. 

Let $N=\phd$, $$\delta_{1,t}=\left(\delta_2+\left(\delta_3+\sqrt{\frac{K_1-t}{n}}\right)^2\right),$$ $$\delta_{2,t}=\left(\delta_3+\sqrt{\frac{K_1-t}{n}}\right)^2.$$ 

 Let $SINR_n=\frac{{P'_{n}}\sum_{i\in S}|H_i|^2}{N}$. Now, arguing similar to \eqref{eq:44}, we get

{\allowdisplaybreaks
\begin{IEEEeqnarray*}{LLL}
\label{eq:pu24a}
&&  \Prb{\bigcup_{t,S,K_1}\biggr( A_S\cap E_2\cap E_3(S,t)\cap E_4\cap E_5\biggr)} \\
&\leq & \Prb{\bigcup_{t,S,K_1} \biggr\{ V_{n,S}\biggr[1+\biggr\{\frac{n(1-\delta_4)^2}{n'}SINR_n \\
&&-(1+\delta_4)^2\biggr(\frac{n{P'_{n}}}{n-K} t\nu\delta_{2,t} +\frac{n{P'_{n}} t(t-1)}{n-K}\delta_{1,t}\biggr) \biggr\}\biggr] \\
&&-1\leq\delta_1\biggr\}}\IEEEyesnumber\\
& \leq &\Prb{\bigcup_{t,S,K_1}\biggr\{\ln\left[1+SINR_n\right]\leq  V'_{n,S}\biggr\}}\label{eq:pu24b} \IEEEyesnumber
\end{IEEEeqnarray*}
}
where $V'_{n,S}=\tilde{V}_{n,S}+\ln(1+\delta_1+O(n^{-1/12}))$. 

Let $\delta_n=\ln(1+\delta_1+O(n^{-1/12})$. We have $\frac{\log M_S }{n-K}=\left(\sum_{i\in S}(R_i-\eta_i)\right) (1+o(1))$. There for sufficiently large $n$ and sufficiently small $\delta$ and $\delta_1$, we have $V'_{n,S}\leq \sum_{i\in S}R_i\ a.s.$ Hence 
\begin{IEEEeqnarray*}{LLL}
\Pb{\bigcup_{t,S,K_1}\log\left[1+SINR_n\right]\leq \log_2(e) V'_{n,S}}\\
 \leq \Pb{\bigcup_{t,S,K_1}\log\left[1+SINR_n\right]\leq  \sum_{i\in S}R_i}. \label{eq:pu25}\IEEEyesnumber
\end{IEEEeqnarray*}

But we know that $P'_n\to P$, and on $D$, from \eqref{eq:pu1} we have 
\begin{equation}
\label{eq:pu26}
\sum_{i\in S}R_i<\log\left(1+\frac{P\sum_{i\in S}|H_{i}|^2}{1+P\sum_{i\in D^c}|H_i|^2}\right)\ a.s.
\end{equation}

Hence the probability in \eqref{eq:pu25} goes to $0$ as $n\to \infty$. 

So combining everything from \eqref{eq:pu12b}, \eqref{eq:pu19}, \eqref{eq:pu23}, \eqref{eq:pu24a}, \eqref{eq:pu24b}, \eqref{eq:pu25} and \eqref{eq:pu26}, we get $p_2\to 0$ as $n\to \infty$. Therefor $p_1\to 1-\frac{\Ex{D}}{K}$ as $n\to \infty$. Hence we have
\begin{IEEEeqnarray}{LLL}
\label{eq:pu27}
\epsilon_n=P_e \to 1-\frac{\Ex{D}}{K} <\epsilon.
\end{IEEEeqnarray}

Hence $\limsup_{n\to\infty}\epsilon_n\leq \epsilon$. Further, since $\eta_i>0$ were arbitrary, we can ensure $\liminf_{n\to \infty}\frac{1}{n}\log M_i^{(n)}\geq R_i,\forall i\in [K]$.

\end{proof}

\subsection{Proof of proposition \ref{prop:conv}}
\label{app:per_user_converse}
\begin{proof}
 We prove the second upper bound in \eqref{eq:pu_conv}. This is based on a single-user converse using the genie argument. Formally, since we consider per-user error, it is enough to look at the event that a particular user is not decoded. Let $W_i\distas{iid} unif[M]$ be the message of user $i$. The channel \eqref{eq:sys1} can be written as $Y=H_1 X_1 +\hat{Z}+Z$ where $\hat{Z}=\sum_{i=2}^{K}H_i X_i$ denotes the interference.  Let $L(Y)$ be the decoder output. Also, let $L(Y,\hat{Z})$ be the decoder output when it has knowledge of $\hat{Z}$. Hence a converse bound $\Pb{W_1\neq \left(L(Y)\right)_1}\geq \epsilon$ is implied by $\Pb{W_1\neq \left(L(Y,\hat{Z})\right)_1}\geq \epsilon$ for all $L(\cdot,\cdot)$. Since $Y-\hat{Z}$ is a sufficient statistic of $(Y,\hat{Z})$ for $W_1$, we have, equivalently, $\Pb{W_1\neq \left(L(Y-\hat{Z})\right)_1}\geq \epsilon$ for all $L(\cdot)$. Letting $\hat{Y}=Y-\hat{Z}$, this is equivalent to a converse for the channel $\hat{Y}=H_1 X_1+Z$: $\Pb{W_1\neq \left(L(\hat{Y})\right)_1}\geq \epsilon$ for all $L(\cdot)$. This is just the usual single user converse, and hence the bound is given by $R\leq C_{\epsilon}=\sup\{\xi:\Pb{\log_2(1+P|H_1|^2)\leq \xi}\leq\epsilon\}=\log_2(1-P\ln(1-\epsilon))$\citep{yang2014quasi}.

\end{proof}

\section{Proofs of certain claims}
\label{app:2}
\begin{proof}[Proof of claim \ref{claim:sc_nc2}]
We have $\norm{Y}^2-\norm{P_{A_0}Y }^2=\norm{\PP_{A_0}\hat{Z}}^2\leq \norm{\hat{Z}}^2-\norm{P_{A_1}\hat{Z}}^2=\norm{\PPA \hat{Z}}^2$.

 Also, $\norm{\PPA Y}^2=\norm{\PPA\hcc +\PPA \hat{Z}}^2$. Hence we have
{\allowdisplaybreaks
\begin{IEEEeqnarray*}{LLLL}
\label{eq:A1}
p_2 &=& \Prb{\bigcup_{t,S,K_1}\biggr\{\frac{\norm{Y}^2-\norm{P_{A_0}Y}^2}{\norm{Y}^2-\norm{P_{A_1}Y}^2}\geq V_{n,t}\biggr\}} \\
& =& \Prb{\bigcup_{t,S,K_1}\biggr\{\norm{\hat{Z}}^2-\norm{P_{A_0}\hat{Z}}^2\\
&&\geq V_{n,t}\norma{\PPA\hcc+\PPA \hat{Z}}^2\biggr\}}\\
& \leq &\Prb{\bigcup_{t,S,K_1}\biggr\{\norm{\PPA \hat{Z}}^2\\
&&\geq V_{n,t}\norma{\PPA\hcc+\PPA \hat{Z}}^2\biggr\}}\\
& =&\Prb{\bigcup_{t,S,K_1}\biggr\{(1-V_{n,t})\norm{\PPA \hat{Z}}^2 \\
&&-2V_{n,t} Re \dop{\PPA \hat{Z}}{\PPA \hcc}\\
&&\geq V_{n,t}\norma{\PPA\hcc}^2 \biggr\}}\\
& =& \Prb{\bigcup_{t,S,K_1}\biggr\{(1-V_{n,t})^2\norm{\PPA \hat{Z}}^2  \\
&&-2V_{n,t} (1-V_{n,t}) Re \dop{\PPA \hat{Z}}{\PPA\hcc}\\
&&\geq V_{n,t} (1-V_{n,t})\norma{\PPA\hcc}^2 \biggr\}}\\
& =&\Prb{\bigcup_{t,S,K_1}\biggr\{\norma{(1-V_{n,t})\PPA \hat{Z}-V_{n,t}\PPA\hcc}^2 \\
&&\geq V_{n,t}\norma{\PPA\hcc}^2\biggr\}} \IEEEyesnumber
\end{IEEEeqnarray*}
} 
 
\end{proof}

\begin{proof}[Proof of claim \ref{claim:sc_nc3}]
Conditional of $H_{[K]}$ and $A_0$, $$\hat{Z}\distas{}\cn\biggr(0,\phsa\biggr).$$ Hence 
\begin{\Ieee}{LLL}
\label{eq:A1a}
\PPA\biggr(\hat{Z}-\frac{V_{n,t}}{1-V_{n,t}}\hcc\biggr)\distas{}\\
\cn\biggr(-\frac{V_{n,t}}{1-V_{n,t}}\PPA\hcc,\phsa\PPA\biggr)\\
\Ieeen
\end{\Ieee}

Now, the rank of $\PPA$ is $n-K_1+t$ because the vectors in $A_1$ are linearly independent almost surely. Let $\mathcal{U}$ be a unitary change of basis matrix that rotates the range space of $\PPA$ to the space corresponding to first $(n-K_1+t)$ coordinates. Then

{\allowdisplaybreaks
\begin{IEEEeqnarray*}{LLL}
\label{eq:A2}
&&\norma{\cn\biggr(-\frac{V_{n,t}}{1-V_{n,t}}\PPA\hcc,\\
&&\phsa\PPA\biggr)}^2 \\
&=&\norma{\mathcal{U} \biggr(\cn\biggr(-\frac{V_{n,t}}{1-V_{n,t}}\PPA\hcc,\\
&& \phsa\PPA\biggr)\biggr)}^2 \\
&=&\norma{\cn\biggr(-\frac{V_{n,t}}{1-V_{n,t}}\mathcal{U}\PPA\hcc,\\
&&\phsa\mathcal{U}\PPA\mathcal{U}^*\biggr)}^2. \IEEEyesnumber
\end{IEEEeqnarray*}
}

Observe that $\mathcal{U} \PPA\mathcal{U}^*$ is a diagonal matrix with first $(n-K_1+t)$ diagonal entries being ones and rest all $0$. Also, if $W=P+iQ \distas{}\cn(\mu,\Gamma)$ (with pseudo-covariance being $0$) then 
\begin{equation}
\label{eq:gauss1}
\begin{bmatrix} P \\ Q \end{bmatrix} \distas{} \mathcal{N}\left(\begin{bmatrix} Re(\mu) \\ Im(\mu) \end{bmatrix}, \frac{1}{2}\begin{bmatrix}
Re(\Gamma) & -Im(\Gamma)\\ Im(\Gamma) & Re(\Gamma) \end{bmatrix}\right).
\end{equation}

Using this and the definition of non-central chi-squared distribution the claim follows. 

\end{proof}

\begin{proof}[Proof of Claim \ref{claim:sc_nc4}]
We have
{\allowdisplaybreaks
\begin{IEEEeqnarray*}{LLL}
\label{eq:A3}
&&f_n(x)\\
&= & x+1+\frac{2V_{n,t}}{1-V_{n,t}}(1+x)\\
&&-\sqrt{1+\frac{2V_{n,t}}{1-V_{n,t}}(1+x)}\sqrt{2x+1+\frac{2V_{n,t}}{1-V_{n,t}}(1+x)}\\
&=&\frac{1}{1-V_{n,t}}\biggr[(1+V_{n,t})(x+1)\\
&&-2\sqrt{V_{n,t}}\sqrt{\left(x+\frac{(1+V_{n,t})^2}{4V_{n,t}}\right)^2-\frac{(1-V_{n,t}^2)^2}{16V_{n,t}^2}}\biggr]\IEEEyesnumber
\end{IEEEeqnarray*}
}

Hence

\begin{IEEEeqnarray*}{LLL}
\label{eq:A4}
f'(x)&=& \frac{1}{1-V_{n,t}}\left[1+V_{n,t}-2\sqrt{V_{n,t}}\frac{a}{\sqrt{a^2-b^2}}\right]\\
&=& \frac{1}{1-V_{n,t}}\biggr(\sqrt{V_{n,t}}-\sqrt{\frac{a+b}{a-b}}\biggr)\cdot\\
&&\biggr(\sqrt{V_{n,t}}-\sqrt{\frac{a-b}{a+b}}\biggr)\\
\IEEEyesnumber
\end{IEEEeqnarray*}

where $a=\left(x+\frac{(1+V_{n,t})^2}{4V_{n,t}}\right)$ and $b=\frac{1-V_{n,t}^2}{4V_{n,t}}$. Also $a>0$ and $b>0$. Further $a+b>a-b$ and 
{\allowdisplaybreaks
\begin{IEEEeqnarray*}{LLL}
\sqrt{V_{n,t}}<\sqrt{\frac{a-b}{a+b}}=\sqrt{\frac{V_{n,t}(1+V_{n,t}+2x)}{1+V_{n,t}+2V_{n,t} x}}\\
\iff 2V_{n,t} x +1 +V_{n,t} < 2x+1+V_{n,t}\\
\iff 0<V_{n,t}<1
\end{IEEEeqnarray*}
}
which is true. Hence both the factors in \eqref{eq:A4} are negative. Therefore $f'(x)>0$.
\end{proof}

\section{Maximal per-user error}
\label{app:pupe-max}
In this section we briefly describe relations between maximal per-user error (PUPE-max) defined in \eqref{eq:pupe_max} and PUPE. First, we represent our system as in \eqref{eq:sys_comp_sens}
\begin{equation}
\label{eq:sys_comp_sens1}
Y=AHW+Z.
\end{equation}

Let $P_{e,i}(A)=\Pb{W_i\neq \hat{W}_i }$. We are interested in bounding the variance of $P_{e,i}(A)$ so that $$\Ex{P_{e,u}^{\max}(A)}=\Ex{\max\{P_{e,i}(A):i\in[K]\}}$$ can be related to $\Ex{P_{e,i}(A)}=\Ex{P_{e,u}}$ due to symmetry on users by random codebook generation. Consider two coupled systems
\begin{\Ieee}{LLL}
Y=AHW+Z\label{eq:pupe_max1}\Ieeen\\
Y'=A'HW +Z \label{eq:pupe_max2}\Ieeen
\end{\Ieee}
where $A$ and $A'$ are fixed so that the channels are dependent on these. 

Now we have
\begin{\Ieee}{LLL}
\label{eq:pupe_max3}
|P_{e,i}(A)-P_{e,i}(A')|&\leq & d_{TV}(P_{Y,H,W},P_{Y',H,W})\\
&\leq & \sqrt{\frac{1}{2}D\left(P_{Y,H,W}||P_{Y',H,W}\right)}\Ieeen
\end{\Ieee}
where $d_{TV}(P,Q)=\sup\{|P(A)-Q(A)|:A\, \mathrm{is\, measurable}\}$ is the total variation distance between measures $P$ and $Q$, $D(P||Q)=\mathbb{E}_{P}\left[\ln\frac{dP}{dQ}\right]$ is the Kullback-Leibler divergence (in nats) and the last inequality is the Pinsker's inequality (see \citep{verdu2014total}). Now using properties of $D$ (see \citep[Theorem 2.2]{polyanskiy2017lecture}) 
\begin{\Ieee}{LLL}
\label{eq:pupe_max4}
&&D\left(P_{Y,H,W}||P_{Y',H,W}\right)\\
& =& D(P_{Y|H,W}||P_{Y'|H,W}|P_{H,W})\\
&=& \int_{H,W}D(P_{Y|H=h,W=w}||P_{Y'|H=h,W=w})dP_{H,W}(h,w)\\
\Ieeen
\end{\Ieee}

Now note that conditioned on $H=h, W=w$, we have $Y\distas{}\cn(Ahw,I_n)$ and $Y'\distas{}\cn(A'hw,I_n)$. Hence a simple computation shows that $D(P_{Y|H=h,W=w}||P_{Y'|H=h,W=w})=\norm{Ahw-A'hw}^2$. Therefore we have

\begin{\Ieee}{LLL}
\label{eq:pupe_max5}
D\left(P_{Y,H,W}||P_{Y',H,W}\right)&=&\Ex{\norm{(A-A')HW}^2}.\Ieeen
\end{\Ieee}

Now let $B=A-A'$ and $X=HW$. Then 
\begin{\Ieee}{LLL}
\label{eq:pupe_max6}
\Ex{\norm{BX}^2}=\sum_{i\in[n]}\Ex{\sum_{j,k\in [KM]}B_{i,j}\bar{B}_{i,k}X_i\bar{X}_k}\Ieeen
\end{\Ieee}

Note that $\Ex{X_j\bar{X}_k}$ is zero if $j\neq k$ and it is $1/M$ otherwise. Hence 
\begin{\Ieee}{LLL}
\label{eq:pupe_max7}
\Ex{\norm{BX}^2}=\frac{1}{M}\sum_{i\in[n]}\sum_{j\in [KM]}B_{i,j}\bar{B}_{i,j}=\frac{1}{M}\norm{B}_F^2.\\
\Ieeen
\end{\Ieee}

Therefore 
\begin{\Ieee}{LLL}
\label{eq:pupe_max8}
D\left(P_{Y,H,W}||P_{Y',H,W}\right)=\frac{1}{M}\norm{A-A'}_F^2.\Ieeen
\end{\Ieee}

So combining this with \eqref{eq:pupe_max3}, we obtain
\begin{\Ieee}{LLL}
\label{eq:pupe_max9}
|P_{e,i}(A)-P_{e,i}(A')|\leq \sqrt{\frac{1}{2M}}\norm{A-A'}_F.\Ieeen
\end{\Ieee}

Now let each entry of $A$ and $A'$ to be distributed iid as $\cn(0,P)$ where $P=\PT/K$. Further, let $\tilde{A}=\sqrt{\frac{K}{\PT}}A$ and $\tilde{A'}=\sqrt{\frac{K}{\PT}}A'$. So the entries of $\tilde{A}$ and $\tilde{A'}$ are iid $\cn(0,1)$. Therefore, with slight abuse of notation, we can rewrite \eqref{eq:pupe_max9} as 
\begin{\Ieee}{LLL}
\label{eq:pupe_max10}
|P_{e,i}(\tilde{A})-P_{e,i}(\tilde{A'})|\leq \sqrt{\frac{\PT}{2MK}}\norm{\tilde{A}-\tilde{A'}}_F.\Ieeen
\end{\Ieee}

Hence the function $P_{e,i}$ is Lipschitz with Lipschitz constant $L=\sqrt{\frac{\PT}{2MK}}$. By concentration of Lipschitz functions of Gaussian random vectors \citep[Theorems 5.5, 5.6]{boucheron2013concentration}, we have that $P_{e,i}(\tilde{A})$ is sub-Gaussian with 
\begin{\Ieee}{LLL}
\label{eq:pupe_max11}
\mathrm{Var}(P_{e,i}(A))\leq 4L^2=\frac{2\PT}{KM}.\Ieeen
\end{\Ieee}

Hence, using bounds on the expected maximum of sub-Gaussian random variables (see~\citep[Section 2.5]{boucheron2013concentration}), we obtain
\begin{\Ieee}{LLL}
\label{eq:pupe_max12}
&&\Ex{\max_{i\in [K]}P_{e,i}(A)}\\
& \leq & \Ex{P_{e,u}}+\sqrt{\mathrm{Var}(P_{e,i}(A))\ln K}\\
&= &\Ex{P_{e,u}} +\sqrt{\frac{2\PT}{M}\frac{\ln K}{K}}\xrightarrow{K\to\infty} \Ex{P_{e,u}}.\Ieeen
\end{\Ieee}

Therefore, a random coding argument along with \eqref{eq:pupe_max12} shows that PUPE-max has same asymptotics as PUPE in the linear scaling regime. For FBL performance, if each user sends $k=100$ bits then $M=2^k$ and hence $\Ex{P_{e,u}^{\max}}\approx \Ex{P_{e,u}}$

\section*{Acknowledgements}
\txb{We thank the anonymous reviewer for suggesting to apply the scalar AMP to derive explicit achievability bounds in the many MAC case (resulting in Section \ref{sec:amp_ach}) and pointing out to use  \eqref{eq:eta_inf} to approximate the $\eta^*$.} We also thank Prof. A. Montanari for the footnote in
Section~\ref{sec:curious} regarding rigorization of the replica-method prediction for PUPE.

% Bibliography.

%\begin{singlespace}
 % \begin{small}
  \bibliographystyle{IEEEtran}
	\bibliography{IEEEabrv,refs}
 % \end{small}
%\end{singlespace}

\begin{IEEEbiographynophoto}{Suhas S Kowshik}
is a PhD student in the department of Electrical Engineering and Computer Science (EECS) at Massachusetts Institute of Technology (MIT), Cambridge, MA, USA. He obtained his S.M. degree from EECS, MIT in 2018, and B.Tech and M.Tech dual degrees from Indian Institute of Technology Madras, Chennai, India in 2016. His research interests include information theory, communication, machine learning and optimization. 
\end{IEEEbiographynophoto}

\begin{IEEEbiographynophoto}{Yury Polyanskiy}
Yury Polyanskiy (S'08-M'10-SM'14) is an Associate Professor of Electrical Engineering and Computer Science and a member of IDSS and LIDS at MIT. Yury received M.S. degree in applied mathematics and physics from the Moscow Institute of Physics and Technology, Moscow, Russia in 2005 and Ph.D. degree in electrical engineering from Princeton University, Princeton, NJ in 2010. His research interests span information theory, statistical learning, error-correcting codes, wireless communication and fault tolerance.
Dr. Polyanskiy won the 2020 IEEE Information Theory Society James Massey Award, 2013 NSF CAREER award and 2011 IEEE Information Theory Society Paper Award.
\end{IEEEbiographynophoto}

\end{document}